\documentclass[11pt,a4paper,DIV16,headsepline,abstracton]{scrartcl}
\pdfoutput=1

\usepackage{amsmath,amssymb,amsthm}
\allowdisplaybreaks
\usepackage{graphics,graphicx}
\usepackage{wrapfig}
\usepackage[margin=1.15in]{geometry}
\usepackage{dsfont}
\usepackage{enumerate}
\usepackage{mathrsfs}
\usepackage{mathtools}
\usepackage{paralist}
\usepackage{color}
\usepackage{transparent}
\usepackage[hidelinks]{hyperref}
\usepackage[affil-it]{authblk}

\usepackage{svg}

\usepackage{scrpage2}

\usepackage{microtype}

\usepackage{booktabs}

\usepackage[numbers]{natbib}
\usepackage{chapterbib}

\usepackage{url}

\usepackage{pgfplots}

\usepackage{authblk}

\usepackage[titletoc,title]{appendix}

\numberwithin{equation}{section}

\numberwithin{equation}{section}

\theoremstyle{plain}
\newtheorem{thm}{Theorem}[section]
\newtheorem{propn}[thm]{Proposition}
\newtheorem{lemma}[thm]{Lemma}
\newtheorem{cor}[thm]{Corollary}

\theoremstyle{definition}
\newtheorem{defn}[thm]{Definition}

\theoremstyle{remark}
\newtheorem{rk}[thm]{Remark}

\newcommand{\RR}{\mathbb{R}}

\newcommand{\NN}{\mathbb{N}}
\newcommand{\ZZ}{\mathbb{Z}}
\newcommand{\CC}{\mathbb{C}}

\newcommand{\dd}{\mathrm{d}}
\newcommand{\norm}[1]{\left\lVert #1\right\rVert}

\newcommand{\Aa}{\mathcal A}

\newcommand{\MM}{\mathcal M}
\newcommand{\HH}{\mathcal H}

\newcommand{\LL}{\mathcal L}
\newcommand{\Ss}{\mathcal S}
\newcommand{\JJ}{\mathcal J}
\newcommand{\tL}{\tilde{\LL}}
\newcommand{\uL}{\underline L}
\newcommand{\uH}{\underline H}
\newcommand{\OO}{\mathcal O}
\newcommand{\Vmod}{\tilde V^{h}}

\newcommand{\del}{\partial}

\renewcommand{\epsilon}{\varepsilon}

\renewcommand{\phi}{\varphi}
\renewcommand{\theta}{\vartheta}
\renewcommand{\rho}{\varrho}

\newcommand*\Laplace{\mathop{}\!\mathbin\bigtriangleup}
\newcommand{\SLaplace}{\Laplace \kern-3mm /\kern+1mm}
\newcommand{\SGrad}{\nabla \kern-2.9mm /\kern+1mm}
\newcommand{\sta}[1]{#1^{\ast}}
\renewcommand{\Im}{\mathrm{Im}\,}
\renewcommand{\Re}{\mathrm{Re}\,}

\newcommand{\e}{\operatorname{e}}
\newcommand{\im}{\mathrm{i}}

\title{Unstable mode solutions\\ to the Klein-Gordon equation\\ in Kerr-anti-de\,Sitter~spacetimes}
\author{Dominic Dold\footnote{\,\,\,Cambridge Centre for Analysis, Department of Pure Mathematics and Mathematical Statistics, University of Cambridge, Wilberforce Road, Cambridge CB3 0AG, United Kingdom\\Email address: \tt{\href{mailto:D.Dold@maths.cam.ac.uk}{D.Dold@maths.cam.ac.uk}}}} 
\date{}

\begin{document}

\maketitle
\vspace{-1.1cm}
\begin{abstract}
For any cosmological constant $\Lambda=-3/\ell^2<0$ and any $\alpha<9/4$, we find a Kerr-AdS spacetime $(\MM,g_{\mathrm{KAdS}})$, in which the Klein-Gordon equation $\Box_{g_{\mathrm{KAdS}}}\psi+\alpha/\ell^2\psi=0$ has an exponentially growing mode solution satisfying a Dirichlet boundary condition at infinity. The spacetime violates the Hawking-Reall bound $r_+^2>|a|\ell$.
We obtain an analogous result for Neumann boundary conditions if $5/4<\alpha<9/4$. 
Moreover, in the Dirichlet case, one can prove that, for any Kerr-AdS spacetime violating the Hawking-Reall bound, there exists an open family of masses $\alpha$ such that the corresponding Klein-Gordon equation permits exponentially growing mode solutions.
 Our result adopts methods of Shlapentokh-Rothman developed in \citep{ShlapentokhGrowing} and provides the first rigorous construction of a superradiant instability for negative cosmological constant.
\end{abstract}

\tableofcontents
 
\section{Introduction}
\label{sec:IntroductionI}

\subsection{The Klein-Gordon equation in asymptotically anti-de\,Sitter spacetimes}

The Einstein vacuum equations
\begin{align}
\label{eqn:Einstein}
R_{\mu\nu}-\frac{1}{2}Rg_{\mu\nu}+\Lambda g_{\mu\nu}=0
\end{align}
with cosmological constant $\Lambda$ can be understood as a system of second-order partial differential equations for the metric tensor $g$ of a four-dimensional spacetime $(\MM,g)$. Solutions with negative cosmological constant have drawn considerable attention in recent years, mainly due to the conjectured instability of these spacetimes. For more details, see \citep{AndersonUniqueness}, \citep{DafermosHolzegelInstability}, \citep{Bizon}, \citep{DiasGravitational}, \citep{DiasHorowitzMarolfSantos}, \citep{HolzegelLukSmuleviciWarnick} and references therein.

\begin{wrapfigure}{r}{.25\textwidth}
	\hspace{0.2cm}
	\def\svgwidth{80pt}
	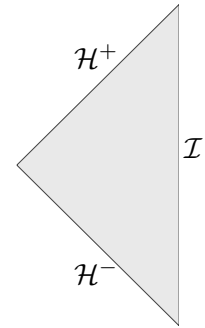
	\caption*{\footnotesize Penrose diagram of the exterior of the Kerr-AdS spacetime \label{figure:Penrose}}
\end{wrapfigure}

In appropriate coordinates, (\ref{eqn:Einstein}) forms a system of non-linear wave equations. A first step in understanding the global dynamics of solutions to (\ref{eqn:Einstein}) -- and thus eventually answering the question of stability -- is the study of linear wave equations on a fixed background. For $\Lambda<0$, efforts have focused on understanding the dynamics of the Klein-Gordon equation
\begin{align}
\label{eqn:KG}
 \Box_g\psi+\frac{\alpha}{\ell^2}\psi=0
\end{align}
for an asymptotically AdS metric $g$ with cosmological constant $\Lambda=-3/\ell^2$ and a mass term $\alpha$ satisfying the Breitenlohner-Freedman bound $\alpha<9/4$ \citep{BreitenlohnerFreedman}, which is required for well-posedness of the equation -- see \citep{WarnickMassive}, \citep{HolzegelWell} and \citep{Vasy}.  The conformally coupled case $\alpha=2$ encompasses scalar-type metric perturbations around an exact AdS spacetime \citep{IshibashiWald3}.

For $g$ being the metric of an exact AdS spacetime, the massive wave equation (\ref{eqn:KG}) allows for time-periodic solutions due to the timelike nature of null and spacelike infinity $\mathcal I$; in particular, general solutions to (\ref{eqn:KG}), while remaining bounded, do not decay.
The behaviour of solutions to (\ref{eqn:KG}) on black-hole spacetimes is very different. Given a Kerr-AdS spacetime with parameters $\ell$, $M$ and $a$ satisfying $|a|<\ell$, define the Hawking-Reall Killing vector field
\begin{align*}
K:=T+\frac{a\Xi}{r_+^2+a^2}\Phi,
\end{align*}
where, using Boyer-Lindquist coordinates, $T=\del_t$ and $\Phi=\del_{\tilde\phi}$; see Section~\ref{subsec:KerrAdS} for definitions of $\Xi$ and $r_+$. The vector field $K$ is the (up to normalisation) unique Killing vector field that is null on the horizon $\HH$ and non-spacelike in a neigbourhood of $\HH$.  It is globally timelike in the black hole exterior if the Hawking-Reall bound $r_+^2>|a|\ell$ is satisfied. If the bound is violated, $K$ becomes non-timelike far away from the horizon.

\begin{wrapfigure}{r}{.4\textwidth}
	\hspace{0.2cm}
	\def\svgwidth{250pt}
	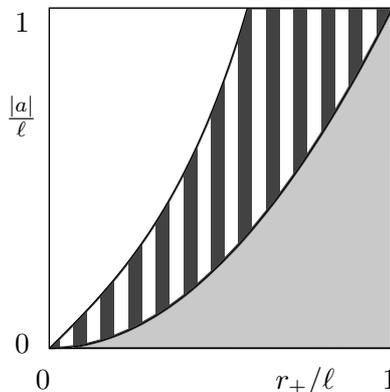
	\caption*{\footnotesize For $r_+<\ell$, the two shaded regions represent the set of admissible parameters for $|a|/\ell$ and $r_+/\ell$. Within the plain gray area (bottom right), the Hawking-Reall bound is satisfied, whereas it is violated in the striped (intermediate) domain.\label{figure:parameters}}
\end{wrapfigure}

 In \citep{HawkingReall}, Hawking and Reall use the existence of a globally causal $K$ for $r_+^2>|a|\ell$ to argue towards the stability of these spacetimes.
 Indeed, uniform boundedness of solutions to (\ref{eqn:KG}) in the full regime $\alpha<9/4$ was proved for $r_+^2>|a|\ell$ in \citep{HolzegelMassive} and \citep{HolzegelWarnickBoundedness}. Moreover, in \citep{HolzegelSmuleviciDecay}, it was shown that solutions with the fastest radial decay (Dirichlet conditions at infinity) in fact decay  logarithmically in time\footnote{\,\,\,Slightly stronger restrictions on $\alpha$ and the spacetime parameters were imposed in \citep{HolzegelSmuleviciDecay} for technical reasons, but the result is believed to hold in full generality by virtue of \citep{HolzegelWarnickBoundedness}.} and \citep{HolzegelSmuleviciQuasi} proves that this logarithmic bound is sharp. 

 For spacetimes violating the Hawking-Reall bound, the global behaviour of solutions to (\ref{eqn:KG}) has not been investigated rigorously, but it was argued in the physics literature -- see \citep{CardosoDias}, \citep{CardosoDiasLemosYoshida}, \citep{CardosoDiasYoshida} and \citep{DiasGravitational} -- that at least for small black holes, i.\,e. for $|a|\ll\ell$ and $|a|\ll r_+$, instability of solutions to (\ref{eqn:KG}) is to be expected if $r_+^2<|a|\ell$. As, in this regime, there is no Killing vector field which is globally timelike in the black hole exterior, this parallels the  situation of asymptotically flat Kerr spacetimes, where superradiance
 is present. For the present discussion, we will understand superradiance loosely as energy extraction from a rotating black hole. We will make this more precise in Lemma~\ref{lemma:negative_energy}.

\subsection{Unstable modes and superradiance in spacetimes with $\Lambda=0$}

The study of energy extraction from black holes in asymptotically flat spacetimes has a long history in the physics literature and two different, but related mechanisms have been proposed.
On the one hand, Press and Teukolsky \citep{PressTeukolsky} suggested that the leakage of energy through the horizon of a rotating black hole could be used to create a black hole bomb by placing a mirror around it. Superradiance would increase the radiation pressure on the mirror over time until it finally breaks, setting free all the energy at once.\footnote{\,\,\,In the asymptotically AdS case, infinity could serve as such a mirror due to the timelike character of spacelike and null infinity.} On the other hand, it was argued that energy could be extracted by the aid of massive waves acting as a natural mirror. This goes back to Zel'dovich \citep{Zeldovich} and was explored further by Starobinsky in \citep{Starobinskii}. Numerous heuristic and numerical studies on the superradiant behaviour of solutions to the Klein-Gordon equation followed, e.\,g. \citep{DamourDeruelleRuffini}, \citep{Zouros}, \citep{Detweiler}, \citep{Dolan2} and \citep{Dolan}. These studies found exponentially growing solutions to the massive wave equation on Kerr spacetimes.

Remarkably, this instability is not present at the level of the massless wave equation
\begin{align*}
\Box_g\psi=0,
\end{align*}
see \citep{DafermosRodnianskiShlapentokh}, where boundedness and decay for such solutions is proved in the full subextremal range $|a|<M$. Even though energy can potentially leak out of the black hole, superradiance can be overcome here as the superradiant frequencies in Fourier spaces are not trapped. In particular, in the context of scattering \citep{DafermosRodnianskiShlapentokhScattering}, a quantitative bound on the maximal superradiant amplification was shown.

In accordance with the above heuristic of massive waves acting as a natural mirror for a black hole bomb, this situation changes dramatically for the Klein-Gordon equation
\begin{align}
\label{eqn:KG_Yakov}
\Box_g\psi-\mu^2\psi=0
\end{align}
with scalar mass $\mu>0$.
A first rigorous construction of exponentially growing finite-energy solutions
in Kerr spacetimes was given by Shlapentokh-Rothman \citep{ShlapentokhGrowing}. The constructed solutions were modes. Mode solutions are solutions of the form
\begin{align}
\label{eqn:Mode}
\psi(t,r,\theta,\tilde\phi)=\e^{-\im\omega t}\e^{\im m\tilde\phi}S_{ml}(\cos\theta)R(r)
\end{align}
in Boyer-Lindquist coordinates $(t,r,\theta,\tilde\phi)$ for $\omega\in\CC$, $m\in\ZZ$ and $l\in\ZZ_{\geq|m|}$, where the smooth functions $S_{ml}$ and $R$ satisfy ordinary differential equations arising from the separability property of the wave equation in Boyer-Lindquist coordinates \citep{Carter_Separability}. We call a mode unstable if it is exponentially growing in time, i.\,e. if $\Im\omega>0$. Shlapentokh-Rothman showed that, for any given Kerr spacetime with $0<|a|<M$, 
there is  an open family of masses $\mu$ producing unstable modes with finite energy. The construction starts from proving existence of real modes and hence produces in particular periodic solutions. We will adopt this strategy.

\subsection{Unstable modes and superradiance in Kerr-AdS spacetimes}

Let us return to the Kerr-AdS case and 
connect the existence of unstable modes to superradiance. Recall that the energy-momentum tensor for the Klein-Gordon equation (\ref{eqn:KG}) is given by
\begin{align*}
 \mathbb{T}_{\mu\nu}:=\Re\left(\nabla_{\mu}\psi\overline{\nabla_{\nu}\psi}\right)-\frac{1}{2}g_{\mu\nu}\left(|\nabla\psi|^2-\frac{\alpha}{\ell^2}|\psi|^2\right)
\end{align*}
and that, for each vector field $X$, we obtain a current
\begin{align*}
 J_{\mu}^X:=\mathbb{T}_{\mu\nu}X^{\nu}.
\end{align*}
While in Kerr spacetimes, $T=\del_t$ (see Section~\ref{subsec:KerrAdS}) is the (up to normalisation) unique timelike Killing field at infinity, the family of vector fields $T+\lambda\Phi$ with $\Phi=\del_{\tilde{\phi}}$ is timelike near infinity in Kerr-AdS spacetimes if and only if
\begin{align}
\label{eqn:Lambda_range}
-\ell^{-2}\left(\ell+a\right)<\lambda<\ell^{-2}\left(\ell-a\right).
\end{align}
Hence, in this range of values for $\lambda$,  the conserved current $J_{\mu}^{T+\lambda\Phi}$ encapsulates the energy density of the scalar field measured by different (rotating) observers at infinity. The vector field $T+\lambda\Phi$ becomes spacelike or null at the horizon.

Recall that the Hawking-Reall vector field $K$ is tangent to the null generators of the horizon $\HH$.  Therefore the energy density radiated through the horizon is measured by
\begin{align*}
 J_{\mu}^{T+\lambda\Phi}K^{\mu}\Big\lvert_{\mathcal H}&=\Re\left(\left(T\psi+\lambda\Phi\psi\right)\overline{K\psi}\right)\Big\lvert_{\mathcal H}\\
  &=\Re\left(\left(T\psi+\lambda\Phi\psi\right)\overline{\left(T\psi+\frac{a\Xi}{r_+^2+a^2}\Phi\psi\right)}\right)\bigg\lvert_{\mathcal H}
\end{align*}
since $g(T+\lambda\Phi,K)=0$ on the horizon.
For mode solutions (\ref{eqn:Mode}), this yields
\begin{align}
\label{eqn:ModeEnergy}
 J_{\mu}^{T+\lambda\Phi}K^{\mu}\Big\lvert_{\HH}&=\left(|\omega|^2-\Re(\omega)\frac{ma\Xi}{r_+^2+a^2}+m\lambda\left(\frac{ma\Xi}{r_+^2+a^2}-\Re(\omega)\right)\right)|\psi|^2\bigg\lvert_{\HH}.
\end{align}
A non-trivial mode solution radiates energy away from the horizon if and only if the expression (\ref{eqn:ModeEnergy}) is negative
for all $\lambda$ in the range (\ref{eqn:Lambda_range}). The thusly characterised frequencies $\omega$ form the superradiant regime.

\begin{lemma}
	\label{lemma:negative_energy}
	Let $r_+^2<|a|\ell$. Let $\psi$ be a mode solution with $\omega(\epsilon)=\omega_R(\epsilon)+\im\epsilon$ for sufficiently small $\epsilon>0$, $\omega_R(\epsilon)\in\RR$ and $\omega_R(0)=ma\Xi/(r_+^2+a^2)$. If
	\begin{align}
	\label{eqn:superradiance}
	\omega_R(0)\frac{\del\omega_R}{\del\epsilon}(0)<0,
	\end{align}
	then $J_{\mu}^{T+\lambda\Phi}K^{\mu}\Big\lvert_{\HH}<0$ for sufficiently small $\epsilon>0$ and $\lambda$ in (\ref{eqn:Lambda_range}).
\end{lemma}

\begin{proof}
	Since $J_{\alpha}^{T+\lambda\Phi}K^{\alpha}=0$ at the horizon for $\epsilon=0$, it suffices to differentiate (\ref{eqn:ModeEnergy}) with respect to $\epsilon$ and evaluate at $\epsilon=0$. We see that the derivative is negative if and only if
	\begin{align*}
	\omega_R(0)\frac{\del\omega_R}{\del\epsilon}(0)-m\lambda\frac{\del\omega_R}{\del\epsilon}(0)<0.
	\end{align*}
	This, however, can be easily checked to hold using (\ref{eqn:Lambda_range}) and $r_+^2<|a|\ell$.
\end{proof}

\begin{rk}
	If $r_+^2>|a|\ell$, then $K$ induces an energy density at infinity and $J_{\mu}^KK^{\mu}\geq 0$, in accordance with the intuition of not being in the superradiant regime if the Hawking-Reall bound is satisfied.
\end{rk}

We will show that our constructed growing mode solutions -- as the modes of \citep{ShlapentokhGrowing} -- satisfy the assumptions of Lemma~\ref{lemma:negative_energy}. This corroborates our interpretation that the unstable modes are a linear manifestation of the superradiant properties of Kerr-AdS spacetimes.

\subsection{The Kerr-AdS family}
\label{subsec:KerrAdS}

Before stating our results, we introduce the Kerr-AdS family of spacetimes. For a more exhaustive presentation, we refer the reader to \citep{HolzegelSmuleviciDecay}. Kerr-AdS spacetimes depend on three parameters $(\ell,M,a)$, where $\ell$ is related to the cosmological constant $\Lambda$ via $\Lambda=-3/\ell^2$. The parameter $M>0$ represents the mass of the black hole and $a$, the angular momentum per unit mass, is assumed to satisfy $|a|<\ell$. This condition guarantees for the metric to be regular. Let
\begin{align*}
\Delta_-(r):=(r^2+a^2)\left(1+\frac{r^2}{\ell^2}\right)-2Mr.
\end{align*}
The polynomial $\Delta_-$ has two real roots, denoted by $r_-<  r_+$. We can write
\begin{align}
\label{eqn:Delta-}
\Delta_-(r)=\ell^{-2}(r-r_+)(r^3+r^2r_++r(r_+^2+a^2+\ell^2)-a^2\ell^2r_+^{-1}),
\end{align}
whence
\begin{align*}
\del_r\Delta_-(r_+)=\frac{1}{\ell^2}(3r_+^3+r_+a^2+r_+\ell^2-a^2\ell^2r_+^{-1}).
\end{align*}
This expression imposes some restrictions on the range of $|a|$ in terms of $r_+$ as shown in the following

\begin{lemma}
	\label{lemma:r+}
	If $r_+<\ell$, 
	\begin{align}
	\label{eqn:rlessr}
	a^2<r_+^2\frac{3\frac{r_+^2}{\ell^2}+1}{1-\frac{r_+^2}{\ell^2}}.
	\end{align}
	If $r_+\geq \ell$, $|a|$ can take any value in $[0,\ell)$.
\end{lemma}

\begin{proof}
	These statements follow from $\del_r\Delta_-(r_+)>0$, which is a necessary condition for $r_-<r_+$. Note also that $r_+\geq \ell$ implies $r_+^2\geq|a|\ell$.
\end{proof}

Therefore, under the restriction of Lemma~\ref{lemma:r+}, there is a bijection between Kerr-AdS spacetimes with parameters $(\ell,M,a)$ and spacetimes with parameters $(\ell,r_+,a)$.
 Henceforth we will use the shorthand notations $\MM_{\mathrm{KAdS}}(\ell,M,a)$ and $\MM_{\mathrm{KAdS}}(\ell,r_+,a)$ to denote Kerr-AdS spacetimes with parameters $(\ell,M,a)$ and $(\ell,r_+,a)$ respectively. The restriction of Lemma~\ref{lemma:r+} can be seen in the above figure.

Given $(\ell,M,a)$, a chart covering all of the domain of outer communication  is given by Boyer-Lindquist coordinates $(t,r,\theta,\tilde{\phi})\in\RR\times(r_+,\infty)\times S^2$. The metric in these coordinates is
\begin{align*}
g_{\mathrm{AdS}}&=-\frac{\Delta_--\Delta_{\theta}a^2\sin^2\theta}{\Sigma}\,\dd t^2-2\frac{\Delta_{\theta}(r^2+a^2)-\Delta_-}{\Xi\Sigma}a\sin^2\theta\,\dd t\,\dd\tilde{\phi}+\frac{\Sigma}{\Delta_-}\,\dd r^2\\
	&~~~~~~~~~~~~+\frac{\Sigma}{\Delta_{\theta}}\,\dd\theta^2+\frac{\Delta_{\theta}(r^2+a^2)^2-\Delta_-a^2\sin^2\theta}{\Xi^2\Sigma}\sin^2\theta\,\dd\tilde{\phi}^2,
\end{align*}
where
\begin{align*}
\Sigma=r^2+a^2\cos^2\theta,~~~~~~\Delta_{\theta}=1-\frac{a^2}{\ell^2}\cos^2\theta,~~~~~~\Xi=1-\frac{a^2}{\ell^2}.
\end{align*}

Since Boyer-Lindquist coordinates break down at $r=r_+$, we introduce Kerr-AdS-star coordinates $(\sta t,r,\theta,\phi)$. These are related to Boyer-Lindquist coordinates by
\begin{align*}
\sta t:=t+A(r)~~~\mathrm{and}~~~\phi:=\tilde{\phi}+B(r),
\end{align*}
where
\begin{align*}
\frac{\dd A}{\dd r}=\frac{2Mr}{\Delta_-(1+r^2/\ell^2)}\mathrm{~~~and~~~}\frac{\dd B}{\dd r}=\frac{a\Xi}{\Delta_-}.
\end{align*}
In these coordinates, the metric extends smoothly through $r=r_+$. One sees that the boundary $r=r_+$ of the Boyer-Lindquist patch is null and we shall call it the event horizon $\HH$.

Finally, we introduce the tortoise coordinate $\sta r$ which is related to $r$ by
\begin{align*}
\frac{\dd \sta r}{\dd r}=\frac{r^2+a^2}{\Delta_-(r)}
\end{align*}
with $\sta r(+\infty)=\pi/2$. We will denote the derivative with respect to $\sta r$ by $'$.

\subsection{Statement of the results}

The analysis in this paper yields two types of instability results:

\begin{itemize}
\item[A.] Given a cosmological constant $\Lambda$ and a mass $\alpha$, there is a Kerr-AdS spacetime for this $\Lambda$ in which (\ref{eqn:KG}) has a growing solution.
\item[B.] Given a Kerr-AdS spacetime violating the Hawking-Reall bound, there is a range for the scalar mass such that, in this spacetime, (\ref{eqn:KG}) has a growing solution.
\end{itemize}

To make this more precise, recall that mode solutions are Fourier modes that take the form
\begin{align*}
\psi(t,r,\theta,\tilde\phi)=\e^{-\im\omega t}\e^{\im m\tilde\phi}S_{ml}(\cos\theta)R(r)
\end{align*}
in Boyer-Lindquist coordinates $(t,r,\theta,\tilde\phi)$ for $\omega\in\CC$, $m\in\ZZ$ and $l\in\ZZ_{\geq|m|}$. Define $u(r):=(r^2+a^2)^{1/2}R(r)$. Use $\Ss_{\mathrm{mode}}(\alpha,\omega,m,l)$ to denote the set of all mode solutions with parameters $\omega,m,l$ to the Klein-Gordon equation with scalar mass $\alpha$. Set
\begin{align*}
\kappa^2:=9/4-\alpha.
\end{align*}

We require that all mode solutions are smooth. For the $S_{ml}$ this is ensured automatically by the definition -- see Section \ref{subsec:Spheroidal}. Hence we only need to impose a regularity condition on the function $u$, given parameters $\ell$, $r_+$, $a$, $m$ and $\omega$.

\begin{defn}[Horizon regularity condition]
	\label{defn:HRC}
A smooth function $f:\,(r_+,\infty)\rightarrow\CC$ satisfies the horizon regularity condition if
$f(r)=(r-r_+)^{\xi}\rho$ for a smooth function $\rho$ as well as a constant
				\begin{align}
				\label{eqn:defn_xi}
				\xi:=\im\frac{\Xi am-(r_+^2+a^2)\omega}{\del_r\Delta_-(r_+)}.
				\end{align}
\end{defn}

Henceforth we will only call a mode $\psi$ a mode solution to (\ref{eqn:KG}) if its radial part $R$ (and hence $u$) satisfies the horizon regularity condition. 
At infinity, we will study two different boundary conditions for $u$. 

\begin{defn}[Dirichlet boundary condition]
	\label{defn:BdyConds}
	Given a mass $\alpha<9/4$ (i.\,e. a $\kappa>0$), a smooth function $f:\,(r_+,\infty)\rightarrow\CC$ satisfies the Dirichlet boundary condition if 
	\begin{align*}
	r^{1/2-\kappa}f\rightarrow 0
	\end{align*}
	as $r\rightarrow\infty$.
	
	We say that a $\psi\in\Ss_{\mathrm{mode}}(\alpha,\omega,m,l)$ satisfies Dirichlet boundary condition if its radial part $u$ satisfies the Dirichlet boundary condition.
\end{defn}

Mode solutions satisfying these boundary conditions are analogous to the modes considered in \citep{ShlapentokhGrowing}. 

We are able to show the following result.

\begin{thm}
	\label{thm:oldD}
	Given a cosmological constant $\Lambda=-3/\ell^2$, a black hole radius $0<r_+<\ell$ and a scalar mass parameter $\alpha_0\in(-\infty,9/4)$, there are a spacetime parameter $a$ satisfying the regularity condition $|a|<\ell$, mode parameters $m$ and $l$ and a $\delta>0$ such that there are a smooth curve
	\begin{align*}
	(-\delta,\delta)\rightarrow\RR^2,~\epsilon\mapsto(\alpha(\epsilon),\omega_R(\epsilon))
	\end{align*}
	with
	\begin{align}
	\label{eqn:HRfrequency}
	\alpha(0)=\alpha_0~~~\mathrm{and}~~~\omega_R(0)=\frac{\Xi am}{r_+^2+a^2}
	\end{align}
	and corresponding
	mode solutions in $\Ss_{\mathrm{mode}}(\alpha(\epsilon),\omega_R(\epsilon)+\im\epsilon,m,l)$  satisfying the horizon regularity condition and Dirichlet boundary conditions.

	For all $\epsilon\in(0,\delta)$, these modes satisfy
	\begin{align*}
\frac{\dd\alpha}{\dd\epsilon}(0)>0	~~~\mathrm{and}~~~\omega_R(0)\frac{\del\omega_R}{\del\epsilon}(0)<0.
	\end{align*}
\end{thm}

\begin{rk}
	The $u$ in the theorem has finite energy and hence the spacetime parameters of the theorem must violate the Hawking-Reall bound as explained in the previous sections; this is explained further in Lemma~\ref{lemma:HR_positivity} and Remark~\ref{rk:HR_positivity}. By Lemma~\ref{lemma:r+}, we know that the $a$ can be located anywhere in the range
	\begin{align*}
	\frac{r_+^2}{\ell}<|a|<r_+\sqrt{\frac{3\frac{r_+^2}{\ell^2}+1}{1-\frac{r_+^2}{\ell^2}}}.
	\end{align*}
	We remark that our result does not restrict to small $|a|$. In fact, we can enforce $|a|$ to be as close to $\ell$ as we wish by choosing $r_+/\ell<1$ large.
\end{rk}

Lemma~\ref{lemma:negative_energy} implies that the constructed modes are superradiant and indicates that the instability is driven by energy leaking through the horizon.

Our next theorem builds on the first, but allows for the construction of an unstable superradiant mode with Dirichlet boundary conditions for each given $\alpha<9/4$.

\begin{thm}
	\label{thm:new}
	Let $\ell>0$ and $\alpha<9/4$. Then there is an $\MM_{\mathrm{KAdS}}(\ell,r_+,a)$ and a super\-radiant $\psi\in\Ss_{\mathrm{mode}}(\alpha,\omega_R+\im\epsilon,m,l)$ for an $\omega_R\in\RR$ and $\epsilon>0$ satisfying Dirichlet boundary conditions.
\end{thm}

The methods used in our proof also show the following statement:

\begin{cor}
\label{cor:new}
Let $\ell>0$, $\alpha<9/4$ and $0<r_+<\ell$. Then there is an $\epsilon>0$ such that for all $|a|\in(r_+^2/\ell,r_+^2/\ell+\epsilon)$, the Klein-Gordon equation with mass $\alpha$ has an exponentially growing mode solution in $\MM_{\mathrm{KAdS}}(\ell,r_+,a)$. 
\end{cor}

\begin{rk}
These results also apply to the massless wave equation, which is an important difference to the asymptotically flat case.
\end{rk}

Furthermore, although this will not be pursued explicitly in this paper, one can also show the analogue of Shlapentokh-Rothman's result in our setting by only adapting the proof slightly.

\begin{thm}
\label{thm:analogue_Yakov}
	Given a Kerr-AdS spacetime $\MM_{\mathrm{KAdS}}(\ell,r_+,a)$ satisfying, $|a|<\ell$, $r_+>0$ and $r_+^2<|a|\ell$ (and the restrictions of Lemma~\ref{lemma:r+}), there are mode parameters $m$ and $l$ as well as a $\delta>0$ such that, for each $\epsilon\in(-\delta,\delta)$, there is an open family of masses $\alpha(\epsilon)$ and a mode solution in $\Ss_{\mathrm{mode}}(\alpha(\epsilon),\omega_R(\epsilon)+\im\epsilon,m,l)$ satisfying Dirichlet boundary conditions with $\omega_R(0)$ as in (\ref{eqn:HRfrequency}).
\end{thm}

\begin{rk}
	\begin{compactenum}
	\item	Conversely, in the asymptotically flat Kerr case of \citep{ShlapentokhGrowing}, it is also possible to prove an analogue of Theorem~\ref{thm:new} instead of only the analogue of Theorem~\ref{thm:oldD}, using our strategy explained in the next section.
	\item To contrast our case to the asymptotically flat setting, we add three observations. First, in \citep{ShlapentokhGrowing}, the curve $\epsilon\mapsto(\mu(\epsilon),\omega_R(\epsilon)+\im\epsilon)$ must satisfy $\mu(0)^2>\omega_R(0)^2$. There is no equivalent condition for Kerr-AdS spacetimes as the instability is not driven by the interplay of frequency and mass, but by the violation of the Hawking-Reall bound. Second, in both cases, $\del\omega_R/\del\epsilon<0$ for small $\epsilon$, so $\omega_R(0)$ can be seen as the upper bound of the superradiant regime. Third, the result in Kerr holds for all $m\neq 0$, $l\geq |m|$. In contrast, our result is a statement about large $m=l$.
	\end{compactenum}
\end{rk}

It is known -- see \citep{HolzegelWarnickBoundedness} and references therein -- that, for $0<\kappa<1$, i.\,e. $5/4<\alpha<9/4$, we also have well-posedness for different boundary conditions at infinity. This underlies the following

\begin{defn}[Neumann boundary condition]
	\label{defn:BdyCondsNeumann}
	Given a mass $5/4<\alpha<9/4$ (i.\,e. $0<\kappa<1$), a smooth function $f:\,(r_+,\infty)\rightarrow\CC$ satisfies the Neumann boundary condition if
\begin{align*}
r^{1+2\kappa}\frac{\dd}{\dd r}\left(r^{\frac{1}{2}-\kappa}f\right)\rightarrow 0
\end{align*} as $r\rightarrow\infty$.
\end{defn}

Using the techniques of twisted derivatives, introduced in \citep{WarnickMassive}, we can prove versions of Theorems~\ref{thm:oldD} and \ref{thm:new} for Neumann boundary conditions. 
\begin{thm}
	\label{thm:oldN}
	Given a cosmological constant $\Lambda=-3/\ell^2$, a black hole radius $0<r_+<\ell$ and a scalar mass parameter $\alpha_0\in(5/4,9/4)$, there are a spacetime parameter $a$ satisfying the regularity condition $|a|<\ell$, mode parameters $m$ and $l$ and a $\delta>0$ such that there is a smooth curve
	\begin{align*}
	(-\delta,\delta)\rightarrow\RR^2,~\epsilon\mapsto(\alpha(\epsilon),\omega_R(\epsilon))
	\end{align*}
	with (\ref{eqn:HRfrequency}). Moreover, there are corresponding mode solutions in $\Ss_{\mathrm{mode}}(\alpha(\epsilon),\omega_R(\epsilon)+\im\epsilon,m,l)$  satisfying Neumann boundary conditions. If $\epsilon\in(0,\delta)$, then the modes satisfy
		\begin{align*}
		\frac{\dd\alpha}{\dd\epsilon}(0)>0	~~~\mathrm{and}~~~\omega_R(0)\frac{\del\omega_R}{\del\epsilon}(0)<0.
		\end{align*}
\end{thm}

\begin{thm}
	\label{thm:newN}
	Let $\ell>0$ and $5/4<\alpha<9/4$. Then there is an $\MM_{\mathrm{KAdS}}(\ell,r_+,a)$ and a super\-radiant $\psi\in\Ss_{\mathrm{mode}}(\alpha,\omega_R+\im\epsilon,m,l)$ for an $\omega_R\in\RR$ and $\epsilon>0$ satisfying Neumann boundary conditions.
\end{thm}

Let us conclude this section with a general remark on boundedness.
From \citep{HolzegelWarnickBoundedness}, we know that solutions to the Klein-Gordon equation with  Dirichlet  boundary conditions remain bounded for all $r_+^2>|a|\ell$. A similar statement holds for Neumann boundary conditions under more restrictive assumptions on the parameters. For $r_+^2=|a|\ell$, one can easily repeat the proof of the second theorem of \citep{HolzegelSmuleviciDecay} to see that there are no periodic solutions. One can potentially also extend the decay result of \citep{HolzegelSmuleviciDecay} to $r_+^2=|a|\ell$. Our results do not rule out boundedness in the entire parameter range in which $r_+^2<|a|\ell$ since we did not show that for \emph{any} given Kerr-AdS spacetime and \emph{any} $\alpha$, there are unstable mode solutions; they do, however, impose restrictions on the ranges of spacetime parameters and masses $\alpha$ in which boundedness could potentially hold. It is believed that, using more refined spectral estimates, our results can be shown to hold in the full regime $r_+^2<|a|\ell$, but we will not pursue this further.

\subsection{Outline of the proof}

The difficulty lies in the construction of the radial part $u$, for which we use the strategy of \citep{ShlapentokhGrowing}, which, as our present work shows, can be applied to more general settings than Kerr spacetimes. The technique contains two main steps.
\begin{itemize}
	\item[I.] Construct $u$ corresponding to a real frequency $\omega_0\in\RR$.
	\item[II.] Obtain a mode solution corresponding to a complex $\omega$ with $\Im\omega>0$ by varying spacetime and mode parameters.
\end{itemize}
We note that both steps are completely independent of each other, in particular step II does not rely on the method by which the periodic mode solution was constructed, but only requires existence of such a mode.

Let us first only deal with Dirichlet boundary conditions. To complete step I, $u$ needs to satisfy the radial ODE
\begin{align}
\label{eqn:angularODE_intro}
u''-(V-\omega_0^2)u=0
\end{align}
for the given boundary condition -- see Section~\ref{subsec:Spheroidal}. Lemma~\ref{thm:RestrReal} then already restricts $\omega_0$ to $\omega_+:=ma\Xi/(r_+^2+a^2)$. It is important to note that the boundary value problem does not admit nontrivial solutions in general.

\begin{lemma}
	If $u$ satisfies the Dirichlet boundary condition for real $\omega_0$ and $V-\omega_0^2\geq 0$, then $u=0$.
\end{lemma}

\begin{proof}
	Define $Q(r):=\Re(u'\overline u)$, note that $Q(r_+)=Q(\infty)=0$ and integrate $\dd Q/\dd r$.
\end{proof}

Hence, in a first step in Section~\ref{subsec:Potential}, we will find spacetime and mode parameters such that $V-\omega_0^2<0$ on some subinterval of $(r_+,\infty)$ for given $\ell$ and $\alpha_0$ by a careful analysis of the shape of the potential $V$ in Lemma~\ref{lemma:Vneg}. This requires proving an asymptotic estimate for the eigenvalues of the modified oblate spheroidal harmonics (Lemma~\ref{lemma:GroundState}). The spacetime parameters will necessarily violate the Hawking-Reall bound.

The radial ODE is the Euler-Lagrange equation of the functional
\begin{align}
\label{eqn:functional_intro}
\LL_a(f):=\int_{r_+}^{\infty}\left(\frac{\Delta_-}{r^2+a^2}\left\lvert\frac{\dd f}{\dd r}\right\lvert^2+(V-\omega^2)\frac{r^2+a^2}{\Delta_-}|f|^2\right)\,\dd r.
\end{align}
The functional is not bounded below, so we need to impose a norm constraint, which we choose to be $\norm{f/r}_{L^2(r_+,\infty)}=1$. Then Lemma~\ref{lemma:VariationInequalityD} gives a coercivity-type estimate. To carry out the direct method of the calculus of variations, we use the weighted Sobolev spaces that arise naturally from the functional -- see Section~\ref{sec:Real}. This setting of the minimisation problem then guarantees that the minimiser satisfies the correct boundary conditions. We remark that we will directly work with the functional (\ref{eqn:functional_intro}) instead of regularising first at the horizon and then taking the limit, as in \citep{ShlapentokhGrowing}.
 Then, in Lemma~\ref{lemma:ELD}, we obtain an ODE
\begin{align*}
u''-(V-\omega_0^2)u+\nu_a\frac{u}{r^2}=0
\end{align*}
with a Lagrange multiplier $\nu_a\leq 0$ that depends continuously on the spacetime parameter $a$. By varying $a$, we find an $\hat a$ such that $\nu_{\hat a}=0$ (Proposition~\ref{propn:a_hat}) and hence a solution to the radial ODE.

To carry out step II, we need the asymptotic analysis of (\ref{eqn:angularODE_intro}) that is worked out in Section~\ref{subsec:LocalAna}. There are two branches that asymptote $r^{-1/2+\kappa}$ and $r^{-1/2-\kappa}$, respectively, at infinity. Let $h_1$ denote the branch with slow decay and $h_2$ the one with fast decay. Then
\begin{align*}
u(r,\alpha,\omega)=A(\alpha,\omega)h_1(r,\alpha,\omega)+B(\alpha,\omega)h_2(r,\alpha,\omega).
\end{align*}  For the parameters from step I, $A(\alpha_0,\omega_0)=0$. By varying $\omega$ and $\alpha$ simultaneously in Section~\ref{subsec:PerturbingD}, the implicit function theorem yields a curve
\begin{align*}
\epsilon\mapsto(\omega_R(\epsilon)+\im\epsilon,\alpha(\epsilon))
\end{align*}
with $\omega_R(0)=\omega_0$ and $\alpha(0)=\alpha_0$ such that
\begin{align*}
A(\alpha(\epsilon),\omega(\epsilon))=0.
\end{align*}
along the curve. As $\Im\omega(\epsilon)>0$ for $\epsilon>0$, these modes grow exponentially whilst satisfying Dirichlet boundary conditions.
 In Section~\ref{subsec:Crossing}, we show that
\begin{align}
\label{epsilon_dependence_outline}
\omega_R(0)\frac{\del\omega_R}{\del\epsilon}(0)<0~~~\mathrm{and}~~~
\frac{\del\alpha}{\del\epsilon}(0)>0,
\end{align}
which proves Theorem~\ref{thm:oldD}. A careful analysis of the domain of the implicit function theorem in Section~\ref{subsec:Continuity} yields Theorem~\ref{thm:new}. Here, the analysis heavily exploits several continuity properties in the parameters. A difficulty is caused by $\hat a$ being defined as the infimum of an open set.

For Corollary~\ref{cor:new}, one observes that, by Lemma~\ref{lemma:Vneg}, once the Hawking-Reall bound is violated, one can always make the potential $V$ negative on some interval by choosing $|m|$ sufficiently large. This yields periodic modes for very small violation of the Hawking-Reall bound and hence growing modes by repeating the above argument.

The situation is more complicated if $u$ satisfies the Neumann boundary condition. Since, in this case, $u\sim r^{-1/2+\kappa}$ as $r\rightarrow\infty$, $\LL_a$ is not well-defined and hence cannot be used to produce periodic modes. 
To carry out the construction of step I, we use twisted derivatives as introduced in \citep{WarnickMassive} and used extensively in \citep{HolzegelWarnickBoundedness}. To find the minimiser via the variational argument, we also need to modify our function spaces and use twisted weighted Sobolev spaces. All details are given in Section~\ref{sec:RealDN}.

The main technical problems, however, arise in the second part of the argument.
The underlying reason is that the proofs for step II rely severely on establishing monotonicity properties for the functional when varying $\alpha$. Since the twisting necessarily depends on $\alpha$, proving monotonicity in $\alpha$ is more involved and indeed the monotonicity properties shown in the Neumann case are weaker; nevertheless, the ideas introduced in Section~\ref{subsec:PerturbingN} are sufficiently robust not only to construct the growing modes, but also to be applicable to showing (\ref{epsilon_dependence_outline}) and to transition from Theorem~\ref{thm:oldN} to Theorem~\ref{thm:newN}. It is also in the Neumann case, where the independence of steps I and II -- alluded to above -- is exploited.

\section{Preliminaries}
\label{sec:Prelim}

\subsection{The modified oblate spheroidal harmonics}
\label{subsec:Spheroidal}

Following \citep{HolzegelSmuleviciDecay}, we define the $L^2(\sin\theta\,\dd\theta\,\dd\tilde{\phi})$-self adjoint operator $P$ acting on $H^1(S^2)$-complex valued functions as
\begin{align*}
-P(\omega,\ell,a)f&=\frac{1}{\sin\theta}\del_{\theta}(\Delta_{\theta}\sin\theta\del_{\theta}f)+\frac{\Xi^2}{\Delta_{\theta}}\frac{1}{\sin^2\theta}\del^2_{\tilde{\phi}}f\\&~~~~~~~~+\Xi\frac{a^2\omega^2}{\Delta_{\theta}}\cos^2\theta f-2\im a\omega\frac{\Xi}{\Delta_{\theta}}\frac{a^2}{\ell^2}\cos^2\theta\del_{\tilde{\phi}}f.
\end{align*}
We also define
\begin{align*}
P_{\alpha}(\omega,\ell,a,\alpha):=\begin{cases}
P(\omega,\ell,a)+\frac{\alpha}{\ell^2}a^2\sin^2\theta		& \mathrm{if~}\alpha>0\\
P(\omega,\ell,a)-\frac{\alpha}{\ell^2}a^2\cos^2\theta		& \mathrm{if~}\alpha\leq 0.
\end{cases}
\end{align*}
For equivalent definitions in Kerr spacetime see \citep{DafermosRodnianskiSmalla} and also \citep{FinsterSchmid} for a more detailed discussion.
From elliptic theory \citep[cf.][]{HolzegelSmuleviciDecay}, we can make the following definitions: $P(\omega,\ell,a)$ has eigenvalues $\tilde{\lambda}_{ml}(\omega,\ell,a)$ with eigenfunctions $\e^{\im m\tilde{\phi}}\tilde S_{ml}(\omega,\ell,a,\cos\theta)$; $P_{\alpha}(\omega,\ell,a,\alpha)$ has eigenvalues ${\lambda}(\omega,\ell,a,\alpha)$ with eigenfunctions $\e^{\im m\tilde{\phi}} S_{ml}(\omega,\ell,a,\alpha,\cos\theta)$. The eigenfunctions form an orthonormal basis of $L^2(\sin\theta\,\dd\theta\,\dd\tilde{\phi})$. Below we will suppress $(\omega,\ell,a,\alpha)$ in the notation.

If $\alpha\leq 0$, $S_{ml}$ satisfies the angular ODE
\begin{align}
\begin{aligned}
\label{eqn:AngularODE1}
&\frac{1}{\sin\theta}\del_{\theta}\left(\Delta_{\theta}\sin\theta\del_{\theta}S_{ml}(\cos\theta)\right)-\bigg(\frac{\Xi^2}{\Delta_{\theta}}\frac{m^2}{\sin^2\theta}-\frac{\Xi}{\Delta_{\theta}}a^2\omega^2\cos^2\theta\\&~~~~~~~~~-2ma\omega\frac{\Xi}{\Delta_{\theta}}\frac{a^2}{\ell^2}\cos^2\theta-\frac{\alpha}{\ell^2}a^2\cos^2\theta\bigg)S_{ml}(\cos\theta)+\lambda_{ml} S_{ml}(\cos\theta)=0
\end{aligned}
\end{align}
for $\lambda_{m\ell}(\omega,\alpha,a)\in\CC$. If $\alpha>0$, the angular ODE takes the form
\begin{align}
\begin{aligned}
\label{eqn:AngularODE2}
&\frac{1}{\sin\theta}\del_{\theta}\left(\Delta_{\theta}\sin\theta\del_{\theta}S_{ml}(\cos\theta)\right)-\bigg(\frac{\Xi^2}{\Delta_{\theta}}\frac{m^2}{\sin^2\theta}-\frac{\Xi}{\Delta_{\theta}}a^2\omega^2\cos^2\theta\\&~~~~~~~~~-2ma\omega\frac{\Xi}{\Delta_{\theta}}\frac{a^2}{\ell^2}\cos^2\theta+\frac{\alpha}{\ell^2}a^2\sin^2\theta\bigg)S_{ml}(\cos\theta)+\lambda_{ml} S_{ml}(\cos\theta)=0.
\end{aligned}
\end{align}
Using these modified oblate spheroidal harmonics, one obtains that, for fixed $m$ and $l$, $u:=\sqrt{r^2+a^2}R$ satisfies the radial ODE
\begin{align}
\label{eqn:radial_ODE_prelim}
u''(r)+(\omega^2-V(r))u(r)=0,
\end{align}
with
\begin{align*}
V(r)&=V_{+}(r)+V_{0}(r)+V_{\alpha}(r)\\
V_{+}(r)&=-\Delta_-^2\frac{3r^2}{(r^2+a^2)^4}+\Delta_-\frac{5\frac{r^4}{\ell^2}+3r^2\left(1+\frac{a^2}{\ell^2}\right)-4Mr+a^2}{(r^2+a^2)^3}\\
V_{0}(r)&=\frac{\Delta_-(\lambda_{ml}+\omega^2a^2)-\Xi^2a^2m^2-2m\omega a\Xi(\Delta_- -(r^2+a^2))}{(r^2+a^2)^2}\\
V_{\alpha}(r)&=-\frac{\alpha}{\ell^2}\frac{\Delta_-}{(r^2+a^2)^2}(r^2+\Theta(\alpha)a^2).
\end{align*}
Here $\Theta(x)=1$ if $x>0$ and zero otherwise.
We will use the shorthand $\tilde V:=V-\omega^2$.
Recall that $'$ denotes an $r^{\ast}$-derivative.

To indicate the dependence upon $a$, we will often write $V_a$ and $\tilde V_a$ for $V$ and $\tilde V$ respectively.

\subsection{Local analysis of the radial ODE}
\label{subsec:LocalAna}

To see which boundary conditions are appropriate for $u$, we perform a local analysis of the radial ODE near the horizon $r=r_+$ and at infinity, using the following theorem about regular singularities, which we cite from \citep{Teschl}, but it can also be found in \citep{ShlapentokhGrowing} or \citep{Olver}.

\begin{thm}
	\label{thm:RegularSing}
	Consider the complex ODE
	\begin{align}
	\label{eqn:CplxODE}
	\frac{\dd^2H}{\dd z^2}+f(z,\nu)\frac{\dd H}{\dd z}+g(z,\nu)H=0.
	\end{align}
	Suppose $f$ and $g$ are meromorphic and have poles of order (at most) one and two, respectively, at $z_0\in\CC$. Let $f_0(\nu)$ and $g_0(\nu)$ be the coefficients of pole of order one and two, respectively, in the Laurent expansions. 
	Let $s_1(\nu)$ and $s_2(\nu)$ be the two solutions of the indicial equation 
	\begin{align*}
	s(s-1)+f_0(\nu)s+g_0(\nu)=0
	\end{align*}
	with $\Re(s_1)\leq\Re(s_2)$.
	
	If $s_2(\nu)-s_1(\nu)\notin\NN_0$, a fundamental system of solutions is given by
	\begin{align*}
	h_j(z,\nu)=(z-z_0)^{s_j(\nu)}\rho_j(z,\nu),
	\end{align*}
	where the functions $\rho_j$ are holomorphic and satisfy $\rho_j(z_0,\nu)=1$.
	
	If $s_2(\nu)-s_1(\nu)=m\in\NN_0$, a fundamental system is given by
	\begin{align*}
	h_1&=(z-z_0)^{s_1}\rho_1+c\log(z)h_2\\
	h_2&=(z-z_0)^{s_2}.
	\end{align*}
	The constant $c$ may be zero unless $m=0$.
	
	In both cases, the radius of convergence of the power series of $\rho_j$ is at least equal to the minimum of the radii of convergence of the Laurent series of $f$ and $g$.
\end{thm}

\subsubsection{The horizon}

Adopting the notation of the previous section, and, after expressing the radial ODE (\ref{eqn:radial_ODE_prelim}) with $r$-derivatives, we have
\begin{align*}
f=\frac{\del_r\Delta_-}{\Delta_-}-\frac{2r}{r^2+a^2},~~~~~~~~~~g=\frac{(r^2+a^2)}{\Delta_-^2}(\omega^2-V).
\end{align*}
Thus we obtain
\begin{align*}
f_0&=\lim_{r\rightarrow r_+}(r-r_+)f=1\\
g_0&=\lim_{r\rightarrow r_+}(r-r_+)^2\frac{(r^2+a^2)^2}{\Delta_-^2}(\omega^2-V)=\lim_{r\rightarrow r_+}\frac{(r-r_+)^2}{\Delta_-^2}\left(\omega (r^2+a^2)-\Xi a m\right)^2
=-\xi^2
\end{align*}
with
\begin{align*}
\xi:=\im\frac{\Xi a m-\omega(r_+^2+a^2)}{\del_r\Delta_-(r_+)}
\end{align*}
as $\del_r\Delta_-(r_+)>0$. Thus, the indicial equation is solved by $s=\pm\xi$. 

Therefore if $\xi\neq 0$, a local basis of solutions $u$ (or $R$) is given by 
\begin{align*}
\{(\cdot-r_+)^{\xi}\phi_1,(\cdot-r_+)^{-\xi}\phi_2\}
\end{align*}
for holomorphic functions $\phi_i$ satisfying $\phi_i(r_+)=1$. For $\xi=0$, a local basis is given by
\begin{align*}
\{\phi_1,\phi_1\left(1+c\log(\cdot-r_+)\right)\}
\end{align*}
for $\phi_1(r_+)=1$ and some constant $c$.

\begin{lemma}
	\label{lemma:SmoothHorizon}
	If $u$ extends smoothly to the horizon, then there is a smooth function $\rho:\,[r_+,\infty)\rightarrow \CC$ such that
	\begin{align*}
	u=(\cdot-r_+)^\xi\rho.
	\end{align*}
\end{lemma}

\begin{proof}
	Boyer-Lindquist coordinates break down at the horizon, so we need to change to Kerr-star coordinates. Then the solution $\psi$ takes the form
	\begin{align*}
	\psi(\sta t,r,\sta\phi,\theta)=\e^{-\im\omega(t-A(r))}\e^{\im m(\sta\phi-B(r))}S_{ml}(a\omega,\cos\theta)\frac{u(r)}{(r^2+a^2)^{1/2}},
	\end{align*}
	where
	\begin{align*}
	\frac{\dd A}{\dd r}=\frac{(r^2+a^2)(1+r^2/\ell^2)-\Delta_-}{\Delta_-(1+r^2/\ell^2)},~~~~~~\frac{\dd B}{\dd r}=\frac{a(1-a^2/\ell^2)}{\Delta_-}.
	\end{align*}
	Hence $R$ extends smoothly to the horizon if there is a smooth function $f$ such that
	\begin{align*}
	u(r)=\e^{-\im(\omega A(r)-m B(r))}f(r).
	\end{align*}
	Therefore the claim reduces to showing that
	\begin{align*}
	\rho(r):=(r-r_+)^{-\xi}\e^{-\im(\omega A(r)-m B(r))}
	\end{align*}
	is smooth. Since
	\begin{align*}
	\frac{\dd}{\dd r}\left(-\im(\omega A(r)-m B(r))\right)=\frac{\xi}{r-r_+}+\OO(1),
	\end{align*}
	we have
	\begin{align*}
	\rho(r)=\e^{-\xi\log(r-r_+)}\e^{\xi\log(r-r_+)+\OO(r-r_+)},
	\end{align*}
	which proves the claim.
\end{proof}

\begin{cor}
	\label{cor:horizon}
	Assume $u$ satisfies the horizon regularity condition. Then a local basis of solutions to the ODE at the horizon is given by
	\begin{align*}
	(\cdot-r_+)^{\xi}\rho
	\end{align*}
	for a holomorphic function $\rho$ defined around $r=r_+$.
\end{cor}

This asymptotic analysis at the horizon motivates the horizon regularity condition of Definition~\ref{defn:HRC}.

\subsubsection{Infinity}
\label{subsec:local_infinity}

The radial ODE has a regular singularity at $r=\infty$. To analyse it using the Theorem~\ref{thm:RegularSing}, we rewrite equation (\ref{eqn:CplxODE}) by introducing $x:=1/z$. This yields
\begin{align*}
\frac{\dd^2 H}{\dd x^2}+\left(\frac{2}{x}-\frac{f}{x^2}\right)\frac{\dd H}{\dd x}+\frac{g}{x^4}H=0.
\end{align*}
For the radial ODE, we have $x=1/r$. We obtain
\begin{align*}
f(x=0)=0,~~~~\lim_{x\rightarrow 0}\frac{f}{x}=2,~~~~
g(x=0)=0,~~~~\lim_{x\rightarrow 0}\frac{g}{x}=0,~~~~\lim_{x\rightarrow 0}\frac{g}{x^2}=\alpha-2.
\end{align*}
The indicial equation becomes
\begin{align*}
s^2-s+\alpha=0,
\end{align*}
which is solved by $s_{\pm}=\frac{1}{2}\pm\sqrt{\frac{9}{4}-\alpha}$. Set
\begin{align*}
\mathcal E:=\left\{\frac{9-k^2}{4}\,:\,k\in\NN\right\}.
\end{align*}
Then, for $\alpha\notin\mathcal E$, a local basis of solutions near infinity is given by
\begin{align*}
\{r^{-1/2+\sqrt{9/4-\alpha}}\rho_1(r),r^{-1/2-\sqrt{9/4-\alpha}}\rho_2(r)\}
\end{align*}
with functions $\rho_1,\rho_2$, smooth at $\infty$ and satisfying $\rho_1(\infty)=\rho_2(\infty)=1$. For $\alpha\in\mathcal E$, a local basis is given by
\begin{align*}
\left\{C_3r^{-1/2-\kappa}\log\frac{1}{r}+r^{-1/2+\kappa}\rho_2,r^{-1/2-\sqrt{9/4-\alpha}}\rho_2(r)\right\}.
\end{align*}
If $u$ extends smoothly to $r=r_+$ and we specify a boundary value $u(r_+)$, then the arguments of Section~\ref{subsec:uniqueness_continuity} show that $C_3$ has to be zero.

\begin{lemma}
	Let $u$ satisfy (\ref{eqn:radial_ODE_prelim}) on $(r_+,\infty)$ and extend smoothly to $r=r_+$, then, for large $r$, $u$ is a linear combination of
	\begin{align*}
	h_1(r,\alpha,\omega,a)&=r^{-1/2+\kappa}\rho_1(r,\alpha,\omega,a)\\
	h_2(r,\alpha,\omega,a)&=r^{-1/2-\kappa}\rho_2(r,\alpha,\omega,a)
	\end{align*}
	for functions $\rho_1$ and $\rho_2$ holomorphic at $r=\infty$ and satisfying $\rho_1(\infty)=\rho_2(\infty)=1$.
\end{lemma}

\begin{cor}
	If $u$ satisfies the horizon regularity condition and the Neumann boundary condition at infinity, then, for $5/4<\alpha<9/4$,
	\begin{align*}
	u=C_1h_1
	\end{align*}
	for a constant $C_1\in\CC$.
	
	If $u$ satisfies the horizon regularity condition and  the Dirichlet boundary condition at infinity, then, for all $\alpha<9/4$,
	\begin{align*}
	u(r)=C_2h_2
	\end{align*}
	for a constant $C_2\in\CC$.
\end{cor}

{

\begin{rk}
	The asymptotics near infinity do not change if we add $\nu(r^2+a^2)/(r^2\Delta_-)$ to $g$ as in Section~\ref{sec:Real}.
\end{rk}

\subsubsection{Uniqueness of solutions and dependence on parameters}
\label{subsec:uniqueness_continuity}

As one would expect, specifying one of the boundary conditions at infinity and choosing a value of $u$ at $r=r_+$ determines the solution to the radial ODE uniquely, which is being made more precise in the following standard lemma.

\begin{lemma}
	\label{lemma:uniqueness_ODE}
	Let $C_0\in\CC$. Then there is a unique classical solution to (\ref{eqn:radial_ODE_prelim}) on $(r_+,\infty)$ satisfying $u(r_+)=C_0$ and extending smoothly to $r=r_+$.
\end{lemma}

The continuous dependence of the solution $u$ on parameters is also well-known:

\begin{lemma}
	\label{lemma:continuous_parameters}
	Let $u_0$ be a unique solution to (\ref{eqn:radial_ODE_prelim}) for a certain set of parameters $(\alpha_0,\omega_0,a_0)$ with fixed $u(r_+)$ satisfying either the Dirichlet or Neumann boundary condition. Let there be a neighbourhood of these parameters such that for all $(\alpha,\omega,a)$ in said neighbourhood, there is a unique solution $u_{\alpha,\omega,a}$ with the same boundary conditions. Fix an $\hat r\in(r_+,\infty)$. Then
	\begin{align*}
	(\alpha,\omega,a)\mapsto u_{\alpha,\omega,a}(\hat r)
	\end{align*}
	is smooth.
\end{lemma}

Let $u$ be a solution to (\ref{eqn:radial_ODE_prelim}) that extends smoothly to the horizon. Fixing $u(r_+)$, we can uniquely define 
 reflection and transmission coefficients $A(\alpha,\omega,a)$ and $B(\alpha,\omega,a)$ via
 \begin{align}
 \label{eqn:reflection_transmission}
 u(r,\alpha,\omega,a)= A(\alpha,\omega,a) h_1(r,\alpha,\omega,a)+B(\alpha,\omega,a) h_2(r,\alpha,\omega,a)
 \end{align}
 for large $r$. Here $h_1$ and $h_2$ are the local basis near infinity from Section~\ref{subsec:local_infinity}.
 Let $W$ denote the Wronskian. {Then
 	\begin{align*}
 	A=\frac{W(u,h_2)}{W(h_1,h_2)}
 	\end{align*}
 and similarly for $B$.
 
 \begin{lemma}
 	\label{lemma:continuous_AB}
 	$A$ and $B$ are smooth in $\alpha$, $\omega$ and $a$.
\end{lemma}

\begin{proof}
	Note that $A$ and $B$ are independent of $r$ and apply Lemma~\ref{lemma:continuous_parameters}.
\end{proof}

\subsection{Detailed analysis of the potential}
\label{subsec:Potential}

From the analysis in \citep{HolzegelSmuleviciDecay} we know that the angular ODE has countably many simple eigenvalues $\lambda_{ml}$, labelled by $l=|m|,|m|+1,\ldots$ for any given $m\in\ZZ$, and corresponding real-valued eigenfunction $S_{ml}$. For later use, we need a bound from below which can be found in \citep{HolzegelSmuleviciDecay}, where it is proved under the assumption of the Hawking-Reall bound. We give the slight extension to our regime.

\begin{lemma}
Let $\omega\in\RR$. For $|a|<\ell$, the eigenvalues satisfy
\begin{align}
\begin{aligned}
\label{eqn:bound_below}
\lambda_{ml}+a^2\omega^2&\geq\Xi^2|m|(|m|+1)\\
\lambda_{ml}+a^2\omega^2&\geq\Xi^2|m|(|m|+1)+a^2\omega_+^2-C_{\ell,a}|m||\omega-\omega_+|,
\end{aligned}
\end{align}
where $C_{\ell,a}>0$ depends on $\ell$ and $a$ only and
\begin{align*}
\omega_+(\ell,r_+,a,m):=\frac{ma\Xi}{r_+^2+a^2}
\end{align*}
\end{lemma}

\begin{proof}
	We focus on the second inequality since the first one can be obtained similarly.
Let
\begin{align*}
\tilde P f:=-\frac{1}{\sin\theta}\del_{\theta}\left(\del_{\theta}\Delta_{\theta}\sin\theta\del_{\theta}f\right)+\Xi^2\frac{m^2}{\sin^2\theta}f.
\end{align*}
Then
\begin{align*}
\lambda_{ml}S_{ml}+a^2\omega^2S_{ml}&\geq\tilde PS_{ml}-\Xi^2\frac{m^2}{\sin^2\theta}S_{ml}+\Xi^2\frac{m^2}{\sin^2\theta}\frac{1}{\Delta_{\theta}}S_{ml}-\Xi\frac{a^2\omega^2}{\Delta_{\theta}}\cos^2\theta S_{ml}\\
	&~~~~~~~~-2ma\omega\frac{\Xi}{\Delta_{\theta}}\frac{a^2}{\ell^2}\cos^2\theta S_{ml}+a^2\omega_+^2S_{ml}+a^2(\omega^2-\omega_+^2)S_{ml}\\
	&=: \tilde PS_{ml}+P_cS_{ml}+a^2\omega_+^2S_{ml},
\end{align*}
where we have already used that the mass term is always nonnegative. We want to show that $P_c\geq 0$. We have the decomposition
\begin{align*}
P_c=P_c^++a^2(\omega^2-\omega_+^2)\left(1-\frac{\Xi}{\Delta_{\theta}}\cos^2\theta\right)-2ma\frac{\Xi}{\Delta_{\theta}}\frac{a^2}{\ell^2}(\omega-\omega_+)\cos^2\theta,
\end{align*}
where $P_c^+$ is the $\omega_+$-part (i.\,e. the part for $\omega=\omega_+$) with
\begin{align*}
P_c^+&=\frac{\Xi^2}{\Delta_{\theta}}\frac{m^2}{\sin^2\theta}\left(\frac{a^2}{\ell^2}\cos^2\theta+\frac{a^4}{(r_+^2+a^2)^2}\sin^4\theta-2\frac{a^2}{r_+^2+a^2}\frac{a^2}{\ell^2}\sin^2\theta\cos^2\theta\right).
\end{align*}
Interpreting the bracket as a function in $\theta$, we see that it has critical points only at $\theta=0,\pi/2,\pi$. Hence $P_c^+\geq 0$.
Therefore, we know
\begin{align*}
P_c&\geq a^2(\omega^2-\omega_+^2)\left(1-\frac{\Xi}{\Delta_{\theta}}\cos^2\theta\right)-2ma\frac{\Xi}{\Delta_{\theta}}\frac{a^2}{\ell^2}(\omega-\omega_+)\cos^2\theta\\
&\geq \frac{a^2}{\Delta_{\theta}} 2\omega_+(\omega-\omega_+)\sin^2\theta-2ma\frac{\Xi}{\Delta_{\theta}}\frac{a^2}{\ell^2}(\omega-\omega_+)\cos^2\theta\\
&\geq -\left(2\frac{a^2}{\Delta_{\theta}}|\omega_+|\sin^2\theta+2|m||a|\frac{\Xi}{\Delta_{\theta}}\frac{a^2}{\ell^2}\cos^2\theta\right)|\omega-\omega_+|
\end{align*}
To obtain the estimate, one only needs to integrate by parts on the sphere. The $\tilde P$ term yields
\begin{align*}
\int_{0}^{\pi}\tilde P S_{ml}\cdot\overline{S}_{ml}&=\int_{0}^{\pi}\left(\Delta_{\theta}|\del_{\theta}S_{ml}|^2+\Xi^2\frac{m^2}{\sin\theta}|S_{ml}|^2\right)\,\dd\theta\\
	&\geq\Xi^2\int_{0}^{\pi}\left(|\del_{\theta}S_{ml}|^2+\frac{m^2}{\sin^2\theta}|S_{ml}|^2\right)\sin\theta\,\dd\theta\\
	&\geq\Xi^2|m|(|m|+1),
\end{align*}
where we compared with spherical harmonics via the min-max principle.
\end{proof}

We will also need an asymptotic upper bound on the ground state eigenvalue $\lambda_{mm}$. 
By the min-max principle, we know that
\begin{align*}
\lambda_{mm}&=\min_{u\in U,\norm{u}=1}\int_0^{\pi}\Bigg(\left[\Delta_{\theta}\left\lvert\frac{\dd u}{\dd\theta}\right\lvert^2+\frac{\Xi^2}{\Delta_{\theta}}\frac{m^2}{\sin^2\theta}|u|^2\right]-\frac{\Xi}{\Delta_{\theta}}a^2\omega^2\cos^2\theta|u|^2\\&~~~~~~~~~~-2ma\omega\frac{\Xi}{\Delta_{\theta}}\frac{a^2}{\ell^2}\cos^2\theta|u|^2-\frac{\alpha}{\ell^2}a^2\cos^2\theta|u|^2\Bigg)\sin\theta\,\dd\theta\\
&\leq \min_{u\in U,\,\norm{u}=1}\int_0^{\pi}\left(\Delta_{\theta}\left\lvert\frac{\dd u}{\dd \theta}\right\lvert^2+\frac{\Xi^2}{\Delta_{\theta}}\frac{m^2}{\sin^2\theta}|u|^2\right)\sin\theta\,\dd\theta+\left(|\alpha|+2|ma|\cdot|\omega-\omega_+|\right)\frac{a^2}{\ell^2}
 \end{align*}
for $U=\{(\sin\theta)^{|m|}\rho(\theta)~:~\rho~\mathrm{analytic}\}$, which is the subspace of $L^2$ which contains all $S_{ml}$ -- see Appendix~\ref{sec:AngularODE}.

\begin{lemma}
Let $n\in\NN$. Then
\begin{align*}
\int_0^{\pi}\sin^n\theta\,\dd \theta=\sqrt{\pi}\frac{\Gamma\left(\frac{n+1}{2}\right)}{\Gamma\left(\frac{n+2}{2}\right)}\mathrm{~~~and~~~}
\int_0^{\pi}\sin^n\theta\cos^2\theta\,\dd\theta=\frac{\sqrt{\pi}}{2}\frac{\Gamma\left(\frac{n+1}{2}\right)}{\Gamma\left(\frac{n+4}{2}\right)},
\end{align*}
in particular
\begin{align*}
\int_0^{\pi}\sin^n\theta\,\dd \theta\sim\sqrt{2\pi}n^{-1/2}~~~\mathrm{and}~~~
\int_0^{\pi}\sin^n\theta\cos^2\theta\,\dd\theta\sim\sqrt{2\pi}n^{-3/2}
\end{align*}
as $n\rightarrow\infty$.
\end{lemma}

\begin{proof}
The expressions follow by induction and the Stirling formula.
\end{proof}

Define
\begin{align*}
u_m:=\left(\pi^{-1/2}\frac{\Gamma\left({(|m|+2)}/{2}\right)}{\Gamma\left({(|m|+3)}/{2}\right)}\right)^{1/2}\sin^{|m|}\theta.
\end{align*}
Then $u_m\in U$ and $\norm{u_m}_{L^2((0,\pi);\sin\theta\,\dd\theta)}$=1. As $1-\cos^2\theta\leq\Delta_{\theta}\leq 1$, we hence know that
\begin{align*}
\lambda_{mm}\leq\int_0^{\pi}\left(\left\lvert\frac{\dd u_m}{\dd \theta}\right\lvert^2+\Xi^2\frac{m^2}{\sin^4\theta}|u_m|^2\right)\sin\theta\,\dd\theta+\left(|\alpha|+2|ma|\cdot|\omega-\omega_+|\right)\frac{a^2}{\ell^2}.
\end{align*}

\begin{lemma}
\label{lemma:GroundState}
\begin{align*}
\lim_{m\rightarrow\infty}\frac{\lambda_{mm}}{m^2}=\Xi^2
\end{align*}
\end{lemma}

\begin{proof}
By equation (\ref{eqn:bound_below}),
we already have $\lim_{m\rightarrow\infty}\lambda_{mm}/m^2\geq \Xi^2$. To prove the result, we compute
\begin{align*}
\frac{1}{m^2}\int_0^{\pi}\left(\left\lvert\frac{\dd u_m}{\dd \theta}\right\lvert^2+\Xi^2\frac{m^2}{\sin^4\theta}|u_m|^2\right)\sin\theta\,\dd\theta&=\pi^{-1/2}\frac{\Gamma\left(\frac{|m|+2}{2}\right)}{\Gamma\left(\frac{|m|+3}{2}\right)}\times\\
	&~~~~~\times \int_0^{\pi}\left(\cos^2\theta\sin^2\theta+\Xi^2\right)\sin^{2|m|-3}\,\dd\theta\\
	&\sim m^{1/2}m^{-3/2}+\Xi^2m^{1/2}m^{-1/2}
\end{align*}
by the previous lemma. Hence $\lim_{m\rightarrow\infty}\lambda_{mm}/m^2\leq\Xi^2$.
\end{proof}

\begin{lemma}
\label{lemma:Vneg}
Let $N,L>0$. Then, given $\ell>0$ and $\alpha<9/4$. Moreover assume the spacetime parameters $r_+$ and $a$ satisfy
\begin{align}
\label{eqn:aRange}
\frac{r_+^4-a^2\ell^2}{(r_+^2+a^2)^2}<-N.
\end{align}
Then
 there is an $m_0>0$ such that for all  mode parameters $|m|\geq m_0$ and $l=m$, we have
\begin{align}
\label{eqn:V_above}
 V-\omega^2\leq -N\frac{\Delta_-m^2\Xi^2}{(r^2+a^2)^2}
\end{align}
on an interval $(R_1,R_2)$ of length $L$ at $\omega=\omega_+(\ell,r_+,a,m)$.
\end{lemma}

\begin{proof}
Let us first rewrite the potential:
\begin{align}
\label{eqn:potential_rewritten}
 V-\omega^2&=V_++V_{\alpha}-\omega^2\notag\\\nonumber
  &~~~~~~~~~+\frac{\Delta_-}{(r^2+a^2)^2}\left(\Xi^2m^2-2\frac{a^2\Xi^2m^2}{r_+^2+a^2}+\frac{m^2a^4\Xi^2}{(r_+^2+a^2)^2}-\frac{m^2a^2\Xi^2}{\Delta_-}\frac{(r^2-r_+^2)^2}{(r_+^2+a^2)^2}\right)\\\notag
  &~~~~~~~~~+\frac{\Delta_-}{(r^2+a^2)^2}(\lambda-\Xi^2m^2)\\\notag
  &=V_++V_{\alpha}+\frac{\Delta_-}{(r^2+a^2)^2}(\lambda-\Xi^2m^2)+\frac{\Delta_-m^2\Xi^2}{(r^2+a^2)^2(r_+^2+a^2)^2}(r_+^4-a^2\ell^2)\\\notag
  &~~~~~~~~~+\frac{\Delta_-m^2\Xi^2}{(r^2+a^2)^2(r_+^2+a^2)^2}\frac{a^2}{\Delta_-}(r-r_+)[r(2r_+^2+a^2+\ell^2)-a^2\ell^2r_+^{-1}+r_+^3]\\\notag
  &=\frac{\Delta_-m^2\Xi^2}{(r^2+a^2)^2}\times\\\notag
  &~~~~~~~\Big[\frac{2r^2}{\Xi^2m^2\ell^2}+\frac{\Delta_-}{(r^2+a^2)^2}\frac{a^2}{m^2\Xi^2}+\frac{r^2-a^2}{(r^2+a^2)^2}\frac{2Mr}{m^2\Xi^2}-\frac{\alpha}{\ell^2}\frac{1}{m^2\Xi^2}\left(r^2+\Theta(\alpha)a^2\right)\\
  &~~~~~~+\left(\frac{\lambda}{\Xi^2m^2}-1\right)+\boxed{\frac{r_+^4-a^2\ell^2}{(r_+^2+a^2)^2}}\\\notag
  &~~~~~~+\frac{a^2(r-r_+)}{\Delta_-}[r(2r_+^2+a^2+\ell^2)-a^2\ell^2r_+^{-1}+r_+^3]\Big]
\end{align}
Note that
\begin{align*}
 [r(2r_+^2+a^2+\ell^2)-a^2\ell^2r_+^{-1}+r_+^3]|_{r=r_+}=\Delta'_-(r_+)>0.
\end{align*}
Moreover $\lambda-\Xi^2m^2>0$ by (\ref{eqn:bound_below}). Therefore, to obtain negativity, we will violate the Hawking-Reall bound in the boxed term. First we can choose $|m|$ large such that $\lambda/\Xi^2m^2-1$  is sufficiently small by Lemma~\ref{lemma:GroundState}. Since the term in the last line is decaying, we can find an $R_1$ such that the last term is bounded on $[R_1,R_1+L]$. By making $|m|$ possibly larger, the terms of the first line are also bounded on the interval.

Let 
\begin{align*}
\frac{r_+^4-a^2\ell^2}{(r_+^2+a^2)^2}\leq-N-\epsilon_m
\end{align*}
Now we choose $m$ sufficiently large such that
\begin{align*}
\frac{\lambda_{mm}}{\Xi^2m^2}-1&<\frac{\epsilon_m}{2}.
\end{align*}
There is an $R_1$ such that
\begin{align*}
 \frac{a^2(r-r_+)}{\Delta_-}[r(2r_+^2+a^2+\ell^2)-a^2\ell^2r_+^{-1}+r_+^3]<\frac{\epsilon_m}{2}
\end{align*}
for all $r\geq R_1$. Set $R_2:=R_1+L$ and choose $m$ such that
\begin{align*}
\frac{1}{m^2}\left(\frac{2r^2}{\Xi^2\ell^2}+\frac{\Delta_-}{(r^2+a^2)^2}\frac{a^2}{\Xi^2}+\frac{r^2-a^2}{(r^2+a^2)^2}\frac{2Mr}{\Xi^2}-\frac{\alpha}{\ell^2}\frac{r^2}{\Xi^2}\right)<\frac{\epsilon_m}{2}
\end{align*}
on $[R_1,R_2]$.
Putting everything together, the lemma follows.
\end{proof}

\begin{rk}
	\label{rk:HR_bound}
The same proof yields the analogous negativity results for $V-\omega^2+F$, where $F$ is any continuous function on $(r_+,\infty)$ that is independent of $m$. This will be used in Section~\ref{sec:RealDN}.
\end{rk}

Define the functional
\begin{align*}
 \LL_{\alpha,r_+,a}(f):=\int_{r_+}^{\infty}\left(\frac{\Delta_-}{r^2+a^2}\left\lvert\frac{\dd f}{\dd r}\right\lvert^2+\frac{(V-\omega^2)(r^2+a^2)}{\Delta_-}|f|^2\right)\,\dd r
\end{align*}
on $C_0^{\infty}(r_+,\infty)$. We often suppress some of the indices and write $\LL_a$ and $V_a$ in view of Section~\ref{sec:Real}.

\begin{lemma}
\label{lemma:FunctionalNegative}
 Choose $(r_+,a,\ell)$ and $(m,l)$ as in Lemma~\ref{lemma:Vneg}. Then there is a function $f\in C_0^{\infty}(r_+,\infty)$ such that
\begin{align*}
 \LL_{a}(f)<0.
\end{align*}
\end{lemma}

\begin{proof}
 We have the following estimate for the functional if $f$ is supported in $(R_1,R_2)$:
\begin{align*}
 \LL_{a}(f)\leq\int_{R_2}^{R_1}\left(\frac{\Delta_-(R_2)}{r_+^2+a^2}\left\lvert\frac{\dd f}{\dd r}\right\lvert^2-N\frac{m^2\Xi^2}{R_2^2+a^2}|f|^2\right)\,\dd r
\end{align*}
Choose an $f$ such that $f$ is 1 on $[R_1+L/4,R_2-L/4]$ and 0 outside of $(R_1,R_2)$. Furthermore we require that
\begin{align*}
 \left\lvert\frac{\dd f}{\dd r}\right\lvert\leq 2\frac{4}{L}.
\end{align*}
Hence
\begin{align*}
 \LL_{a}(f)\leq\frac{64\Delta_-(R_2)}{(r_+^2+a^2)L^2}-N\frac{m^2\Xi^2L}{2(R_2^2+a^2)}.
\end{align*}
If necessary, we can increase $m$ further to make the expression negative. 
\end{proof}

\begin{rk}
	\label{rk:Choice}
	
Fix $\ell>0$ and $0<r_+<\ell$. Choose $\alpha_0<9/4$ and $a=a_0$ such that (\ref{eqn:aRange}) holds. Then there is a non-empty open interval $I\subseteq(-\infty,9/4)$ with $\alpha_0\in I$, a non-empty open Interval $I'$ around $a_0$  and an $m_0$ such that Lemma~\ref{lemma:FunctionalNegative} holds on $[R_1,R_2]$ for all $\alpha\in I$, $a\in I'$ and $|m|\geq m_0$.
\end{rk}

We want to conclude this section by showing that $\LL_a$ is always non-negative if the Hawking-Reall bound is satisfied. We borrow the following Hardy inequality from \citep{HolzegelSmuleviciDecay}.

\begin{lemma}
	\label{lemma:Hardy}
	For any $r_{\mathrm{cut}}\geq r_+$, we have for a smooth function $f$ with $fr^{1/2}=o(1)$ at infinity that
	\begin{align*}
	\frac{1}{4\ell^2}\int_{r_{\mathrm{cut}}}^{\infty}|f|^2\,\dd r\leq \int_{r_{\mathrm{cut}}}^{\infty}\frac{\Delta_-}{r^2+a^2}\left\lvert\frac{\dd f}{\dd r}\right\lvert^2\,\dd r.
	\end{align*}
\end{lemma}

\begin{proof}
	We include a proof for the sake of completeness. Integrating by parts and applying the Cauchy-Schwarz inequality yields
	\begin{align*}
	\int_{r_{\mathrm{cut}}}^{\infty}\frac{\dd}{\dd r}(r-r_{\mathrm{cut}})|f|^2\,\dd r\leq 4\int_{r_{\mathrm{cut}}}^{\infty}(r-r_{\mathrm{cut}})^2\left\lvert\frac{\dd f}{\dd r}\right\lvert^2\,\dd r.
	\end{align*}
	The lemma follows by estimating $r_{\mathrm{cut}}\geq r_+$.
\end{proof}

Thus we can prove the

\begin{lemma}
	\label{lemma:HR_positivity}
	Let $r_+^2\geq |a|\ell$. Then there is an $m_0$ such that, for all $|m|\geq m_0$ and $f\in C_0^{\infty}(r_+,\infty)$,
	\begin{align*}
	\LL_a(f)\geq 0.
	\end{align*}
\end{lemma}

\begin{proof}
	Let us first assume $r_+^2>|a|\ell$. Then, noting that
	\begin{align*}
	V_+=\frac{2\Delta_-}{(r^2+a^2)^2}\frac{r^2}{\ell^2}+\frac{\Delta_-}{(r^2+a^2)^4}\left(a^4\Delta_-+(r^2-a^2)2Mr\right),
	\end{align*}
	on sees from (\ref{eqn:potential_rewritten}) that
	\begin{align*}
	\tilde V_a>\frac{2-\alpha}{\ell^2}\frac{\Delta_-}{(r^2+a^2)^2}r^2>-\frac{1}{4\ell^2}\frac{\Delta_-}{(r^2+a^2)^2}r^2.
	\end{align*}
	Using Lemma~\ref{lemma:Hardy}, we conclude $\LL_a(f)>0$. By continuity, we obtain $\LL_a(f)\geq 0$ for $r_+^2\geq |a|\ell^2$.
\end{proof}

An analogue of Lemma~\ref{lemma:FunctionalNegative} can be proved for the twisted functional used in Section~\ref{sec:RealDN}. For $0<\kappa<1$, define
\begin{align*}
\tL_a(f):=\int_{r_+}^{\infty}\left(\frac{\Delta_-}{r^2+a^2}r^{-1+2\kappa}\left\lvert\frac{\dd}{\dd r}\left(r^{\frac{1}{2}-\kappa}f\right)\right\lvert^2+\tilde V_a^{h}\frac{r^2+a^2}{\Delta_-}|f|^2\right)\,\dd r
\end{align*}
with $\tilde V_a^{h}$ as in Section~\ref{sec:RealDN}.

\begin{lemma}
	\label{lemma:FunctionalNegativeTwisted}
	Choose $(r_+,a,\ell)$ and $(m,l)$ as in Lemma~\ref{lemma:Vneg}. Then there is a function $f\in C_0^{\infty}(r_+,\infty)$ such that
	\begin{align*}
	\tL_{a}(f)<0.
	\end{align*}
\end{lemma}

\begin{proof}
	For $f\in C_0^{\infty}(r_+,\infty)$, $\LL_a(f)=\tL_a(f)$ by choice of $\Vmod_a$.
\end{proof}

\subsection{Periodic mode solutions}
\label{sec:RestrReal}

\begin{lemma}
	\label{thm:RestrReal}
 Suppose we have a $\psi\in\Ss_{\mathrm{mod}}(\alpha,\omega,m,l)$ such that $\omega\in\RR$. Then the following statements are true:
\begin{compactenum}
 \item[(i)] We have $m a\Xi-(r_+^2+a^2)\omega=0$, i.\,e. that $\omega=\omega_+(\ell,r_+,a,m)$.
\item[(ii)] We have $am\neq 0$.
\end{compactenum}
\end{lemma}

\begin{proof}
 We wish to show $(i)$. First let us only deal with the Dirichlet branch.
  Then $u$ is decaying at infinity. Define the microlocal energy current
\begin{align*}
 Q_T:=\Im\left(u'\overline{u}\right).
\end{align*}
We have $Q_T(\infty)=0$. Moreover
\begin{align*}
 \frac{\dd Q_T}{\dd r}&=\frac{\dd\sta r}{\dd r}\Im(u''\overline{u}+|u'|^2)=0
\end{align*}
by the radial ODE.
By Lemma~\ref{lemma:SmoothHorizon}, we obtain
\begin{align*}
 u'&=\frac{\dd r}{\dd\sta r}\left(\frac{\xi}{r-r_+}u(r)+(r-r_+)^{\xi}\rho'(r)\right)=\frac{\Delta_-}{r^2+a^2}\left(\frac{\xi}{r-r_+}u(r)+(r-r_+)^{\xi}\rho'(r)\right)
\end{align*}
and so
\begin{align*}
u'(r_+)=\im\frac{\Xi am-(r_+^2+a^2)\omega}{r_+^2+a^2}u(r_+).
\end{align*}
We conclude
\begin{align}
\label{eqn:microlocal_horizon}
 0=Q_T(r_+)=(r_+^2+a^2)\Im\left(\frac{\dd u}{\dd\sta r}(r_+)\overline{u(r_+)}\right)=(am\Xi-(r_+^2+a^2)\omega)|u(r_+)|^2.
\end{align}
If $u(r_+)=0$, then $u$ vanishes identically by Lemma~\ref{lemma:uniqueness_ODE}. Hence we conclude that
\begin{align*}
m a\Xi-(r_+^2+a^2)\omega=0.
\end{align*}
For the Neumann branch of the solution we observe that
\begin{align*}
Q_T=\Im\left(r^{-\frac{1}{2}+\kappa}\frac{\dd}{\dd\sta r}\left(r^{\frac{1}{2}-\kappa}u\right)\overline u\right).
\end{align*}
From the boundary condition, we immediately get  $Q_T(\infty)=0$ as well and the rest follows as above.

Part (ii) follows immediately from $(r_+^2+a^2)\omega=\Xi am$.
\end{proof}

\section{Growing mode solutions satisfying Dirichlet boundary conditions}
\label{sec:Growing_Dirichlet}

\subsection{Existence of real mode solutions}
\label{sec:Real}

We now fix $\ell>0$, $\alpha<9/4$ and $0<r_+<\ell$. 
 Recall the variational functional
\begin{align*}
\LL_a(f):=\int_{r_+}^{\infty}\left(\frac{\Delta_-}{r^2+a^2}\left\lvert\frac{\dd f}{\dd r}\right\lvert^2+\tilde V_a\frac{r^2+a^2}{\Delta_-}|f|^2\right)\,\dd r.
\end{align*}
for $\omega=\omega_+=am\Xi/(r_+^2+a^2)$.
 Define
\begin{align*}
\Aa:=\{a>0\,:\,\exists f\in C_0^{\infty}:\,\LL_a(f)<0\}.
\end{align*}
By Lemma~\ref{lemma:FunctionalNegative}, there is an $m_0$ such that $\Aa$ is non-empty for all $|m|\geq m_0$ and $l=m$. Fix $m$ and $l$ henceforth.

If the bound $r_+^2\geq |a|\ell$ is satisfied, then $\LL_a(f)\geq 0$ for all compactly supported $f$ by Lemma~\ref{lemma:HR_positivity}. Hence $\Aa$ is bounded below by a strictly positive infimum. Moreover $\Aa$ is open as $a\mapsto \LL_{a}(f)$ is continuous for any fixed $f$.

\begin{rk}
	We restrict ourselves to $a>0$, but we could have defined the set $\Aa$ to also include negative values of $a$.
\end{rk}

Our aim is to show that $\LL_a$ has a minimiser for $a\in \Aa$. We will apply the natural steps of the direct method of the calculus of variations. First, we will specify an appropriate function space, then we will show that the functional obeys a coercivity condition and that the functional is weakly lower semicontinuous. The existence of a minimiser follows by an application of compactness results.

For $U\subseteq(r_+,\infty)$ define the weighted norm
\begin{align*}
\norm{f}^2_{\uL^2(U)}=\int_U\frac{1}{r^2}|f|^2\,\dd r
\end{align*}
and the space
\begin{align*}
\uL^2(U):=\{f\mathrm{~measurable~}:~\norm{f}_{\uL^2(U)}<\infty\}.
\end{align*}
This is clearly a Hilbert space with the natural inner product $(\cdot,\cdot)_{\uL^2(U)}$.

For $U\subseteq(r_+,\infty)$, we define the weighted Sobolev space $\uH^1$ via the norm
\begin{align*}
\norm{f}_{\uH^1(U)}^2:=\int_U\left(|f|^2+r(r-r_+)\left\lvert\frac{\dd f}{\dd r}\right\lvert^2\right)\,\dd r
\end{align*}
Note that for $U\subseteq(r_+,\infty)$ compact, the $\uH^1$ norm is equivalent to the standard Sobolev norm. As usual, let $\uH^1_0(U)$ be the completion of $C_0^{\infty}(U)$ under $\norm{\cdot}_{\uH^1(U)}$.

\begin{lemma}
	\label{lemma:trace}
	Let $u\in\uH_0^1(r_+,\infty)$. Then is $u$ is also in $C(r_++1,\infty)$ (after possibly changing it on a set of measure zero) and 
	\begin{align*}
	\lim_{r\rightarrow\infty } u(r)=\lim_{r\rightarrow\infty }r^{1/2-\kappa} u(r)=0
	\end{align*}
	for all $\kappa>0$.
\end{lemma}

\begin{proof}
	Establishing the embedding $\uH_0^1(r_++1,\infty)\subseteq C(r_++1,\infty)$ is standard. Now take a sequence $(u_m)$ in $C_0^{\infty}$ such that $u_m\rightarrow u$ in $\uH^1_0$ and pointwise almost everywhere. Choose an $R$ such that $(u_m)$ converges pointwise there. For any $\beta<1/2$, we have
	\begin{align*}
	\left\lvert\lim_{r\rightarrow\infty}r^{\beta} u(r)\right\lvert&\leq R^{\beta}\left\lvert  u-u_m\right\lvert(R)+\int_R^{\infty}\left\lvert\del_r\left(r^{\beta}( u-u_m)\right)\right\lvert\,\dd r\\
		&\leq R^{\beta}\left\lvert  u-u_m\right\lvert(R)+C''\norm{ u-u_m}_{\uH^1}.
	\end{align*}
	Therefore, the claim follows.
\end{proof}

To establish a coercivity-type inequality, we use the Hardy inequality of Lemma~\ref{lemma:Hardy}:

\begin{lemma}
\label{lemma:VariationInequalityD}
Let $a\in \Aa$ be fixed. There exist constants $r_+<B_0<B_1<\infty$ and constants $C_0,C_1,C_2>0$, such that, for sufficiently large $m$, we have for all smooth functions $f$ with $fr^{1/2}=o(1)$ at infinity that
\begin{align*}
\int_{r_+}^{\infty}\left(\frac{\Delta_-}{r^2+a^2}\left\lvert\frac{\dd f}{\dd r}\right\lvert^2+C_01_{[B_0,B_1]^c}|f|^2\right)\,\dd r\leq C_1\int_{B_0}^{B_1}|f|^2\,\dd r+C_2\LL_{a}(f).
\end{align*}
 Here we can choose $C_2=1$ if $\alpha<2$.
\end{lemma}

\begin{rk}
Note that the dependence of the expression on $\alpha$ is via $\tilde V$ in $\LL_a$. Recall from Section~\ref{subsec:Spheroidal} that
\begin{align*}
\tilde V=V-\omega^2=V_++V_0+V_{\alpha}-\omega^2.
\end{align*}
\end{rk}

\begin{proof}
First, we have to study the potential again:
\begin{align*}
(r_+^2+a^2)\frac{V_+}{\Delta_-}(r_+)&=\frac{3\frac{r_+^4}{\ell^2}+r_+^2\left(1+\frac{a^2}{\ell^2}\right)-a^2}{(r_+^2+a^2)^2}\\
(r_+^2+a^2)\frac{V_{\alpha}}{\Delta_-}(r_+)&=-\frac{\alpha}{\ell^2}\frac{1}{r_+^2+a^2}(r_+^2+\Theta(\alpha)a^2)\\
(r_+^2+a^2)\frac{V_0-\omega_+^2}{\Delta_-}(r_+)&=\frac{\lambda+\omega_+^2a^2-2m\omega a\Xi}{r_+^2+a^2}\\
	&\geq \frac{\Xi^2m^2}{(r_+^2+a^2)^3}r_+^4
\end{align*}
Thus for sufficiently large $|m|$, the expression is greater than zero. Furthermore, note the asymptotics
\begin{align}
\label{eqn:Potential_asymptotics}
(r^2+a^2)\frac{\tilde V}{\Delta_-}\rightarrow\ell^{-2}(2-\alpha)
\end{align}
as $r\rightarrow\infty$.

We will deal with the cases $\alpha<2$ and $\alpha\geq 2$ separately. First, let $\alpha<2$. The function $\frac{r^2+a^2}{\Delta_-}\tilde V$ is only nonpositive on an interval $[R_1,R_2]$. Choose constants such that $r_+<B_1<R_1<R_2<B_2<\infty$. Set $C_0$ to be the minimum of $\frac{r^2+a^2}{\Delta_-}\tilde V$ on $(r_+,\infty)\backslash[B_1,B_2]$ and set $-C_1$ to be its minimum on $[B_1,B_2]$. This immediately yields the result.

Now let $\alpha \geq 2$. There exist $R_1,R_2$ such that $\frac{r^2+a^2}{\Delta_-}\tilde V$ is positive on $(r_+,R_1)$ and
\begin{align*}
\frac{r^2+a^2}{\Delta_-}\tilde V>-\frac{1}{4\ell^2}(1-\epsilon)
\end{align*}
on $(R_2,\infty)$ for an $\epsilon>0$ because of (\ref{eqn:Potential_asymptotics}). Hence
\begin{align*}
\int_{R_2}^{\infty}\frac{r^2+a^2}{\Delta_-}\tilde V|f|^2\,\dd r&>-\frac{1-\epsilon/2}{4\ell^2}\int_{R_2}^{\infty}|f|^2\,\dd r+\frac{\epsilon}{8\ell^2}\int_{R_2}^{\infty}|f|^2\,\dd r\\
	&\geq -\left(1-\frac{\epsilon}{2}\right)\int_{R_2}^{\infty}\frac{\Delta_-}{r^2+a^2}\left\lvert\frac{\dd f}{\dd r}\right\lvert^2\,\dd r+\frac{\epsilon}{8\ell^2}\int_{R_2}^{\infty}|f|^2\,\dd r
\end{align*}
by Lemma~\ref{lemma:Hardy}. Choose $B_1,B_2$ as before. Let $C$ be the minimum of $\frac{r^2+a^2}{\Delta_-}\tilde V$ on $(r_+,B_1)$. Let $\epsilon C_0/2$ be the minimum of $C$ and $\epsilon/(8\ell^2)$. Moreover, set  $-\epsilon C_1/2$ to be the minimum of $\frac{r^2+a^2}{\Delta_-}\tilde V$ on $[B_1,B_2]$, we obtain
\begin{align*}
\int_{r_+}^{\infty}\left(\epsilon \frac{\Delta_-}{r^2+a^2}\left\lvert\frac{\dd f}{\dd r}\right\lvert^2+\epsilon C_01_{[B_0,B_1]^c}|f|^2\right)\,\dd r\leq \epsilon C_1\int_{B_0}^{B_1}|f|^2\,\dd r+\LL_{a}(f)
\end{align*}
and hence the inequality.
\end{proof}

\begin{lemma}
	\label{lemma:SemictsD}
	The functional $\LL_a$ is weakly lower semicontinuous in $\uH^1(r_+,\infty)$ when restricted to functions of untit $\uL^2$ norm.
\end{lemma}

\begin{proof}
	As the functional is convex in the derivative, the statement is standard and a proof can be extracted from \citep[][\textsection 8]{Evans}. We note that the boundedness from below comes from the norm constraint. The $r$ weight deals with $(r_+,\infty)$ having non-finite measure.
\end{proof}

\begin{lemma}
\label{lemma:RegMinD}
Let $a\in\Aa$. Then there exists an $f_{a}\in \uH_{0}^1(r_+,\infty)$ with unit $\uL^2(r_+,\infty)$ norm such that $\LL_{a}$ achieves its infimum over 
\begin{align*}
\{f\in \uH_0^1(r_+,\infty)\,:\,\norm{f}_{\uL^2}=1\}
\end{align*}
on $f_{a}$.
\end{lemma}

\begin{proof}
By Lemma~\ref{lemma:VariationInequalityD},
\begin{align}
\label{eqn:boundLmuD}
\int_{r_+}^{\infty}\left(\frac{\Delta_-}{r^2+a^2}\left\lvert\frac{\dd f}{\dd r}\right\lvert^2+C_01_{[B_0,B_1]^c}|f|^2\right)\,\dd r\leq C_1\int_{B_0}^{B_1}|f|^2\,\dd r+C_2\LL_{a}(f)
\end{align}
holds for all $f\in \uH_0^{1}$.
From this, it is evident that
\begin{align*}
 \LL_{a}(f)>-\infty
\end{align*}
if $\norm{f}_{\uL^2}=1$,
whence
\begin{align*}
\nu_{a}	&=\inf\{\LL_{a}(f)\,:\,f\in \uH_0^{1},\,\norm{f}_{\uL^2}=1\}>-\infty.
\end{align*}
We can choose a minimising sequence of functions of compact support by density. Thus let $\{f_{a,n}\}$ be a sequence of smooth functions, compactly supported in $(r_+,\infty)$ with $\norm{f_{a,n}}_{\uL^2}=1$, such that
\begin{align*}
\LL_{a}(f_{a,n})\rightarrow\nu_{a}.
\end{align*}
The bound (\ref{eqn:boundLmuD}) implies that $\norm{f_{a,n}}_{\uH^1}$ is uniformly bounded. Thus by the Banach-Alaoglu theorem, it has a weakly convergent subsequence in $\uH_0^1(r_+,\infty)$.
Recall a simple version of Rellich-Kondrachov: $H^1[a,b]$ embeds compactly into $L^2[a,b]$. Hence by the equivalence of norms, the subsequence has a strongly in $L^2$ convergent subsequence on compact subsets of $(r_+,\infty)$.
 Relabelling, we have a sequence $\{f_{a,n}\}$ that converges to $f_{a}$ weakly in $\uH_0^1$ and strongly in $L^2$ on compact subsets of $(r_+,\infty)$. The space $\uH_0^1$ is a linear (hence convex) subspace of $\uH^1$ that is norm-closed. Every convex subset that is norm closed is weakly closed. Therefore, $f_a\in \uH_0^1$.

We claim that $\norm{f_{a}}_{\uL^2}=1$. We have
\begin{align*}
\left\lvert\norm{f_a}_{\uL^2}-1\right\lvert&\leq\left\lvert \norm{f_a}_{\uL^2(r_++1/N,N)}-\norm{f_{a,n}}_{\uL^2(r_++1/N,N)}\right\lvert\\
	&~~~~~~~~~+\left\lvert\norm{f_a}_{\uL^2(r_++1/N,N)^c}-\norm{f_{a,n}}_{\uL^2(r_++1/N,N)^c}\right\lvert
\end{align*}
Due to the $L^2$ convergence on compact subsets, the claim follows if
\begin{align*}
\lim_{N\rightarrow\infty}\lim_{n\rightarrow\infty}\norm{f_n}_{\uL^2((r_+,\infty)\backslash[r_++1/N,N]}=0.
\end{align*}
Suppose not. 
Then there is a $\rho$ such that, for any $N$, there are infinitely many of the $f_{a,n}$ such that
\begin{align*}
\norm{f_{a,n}}_{\uL^2((r_+,\infty)\backslash [r_++1/N,N])}\geq \rho>0,
\end{align*}
i.\,e. the norm must concentrate either near the horizon or near infinity.
Suppose first that
\begin{align*}
\norm{f_{a,n}}_{\uL^2(r_+,r_++\delta)}\geq \rho_1>0
\end{align*}
for infinitely many $f_{a,n}$ and any $\delta>0$. 
By (\ref{eqn:boundLmuD}), we have for $r\in(r_+,r_++1)$:
\begin{align*}
|f_{a,n}(r)|&\leq \int_{r}^{r_++1}\left\lvert\frac{\dd f_{a,n}}{\dd r'}\right\lvert\,\dd r'+\int_{r_++1}^{\infty}\left\lvert\frac{\dd f_{a,n}}{\dd r'}\right\lvert\,\dd r'\\
	&\leq \left(\int_{r}^{r_++1}\frac{1}{r'-r_+}\,\dd r'\right)^{1/2}\left(\int_{r}^{r_++1}(r'-r_+)\left\lvert\frac{\dd f_{a,n}}{\dd r'}\right\lvert^2\,\dd r'\right)^{1/2}\\
	&~~~~~~~~+\left(\int_{r_++1}^{\infty}\frac{1}{(r')^2}\,\dd r'\right)^{1/2}\left(\int_{r_++1}^{\infty}r^2\left\lvert\frac{\dd f_{a,n}}{\dd r'}\right\lvert^2\,\dd r'\right)^{1/2}\\
	&\leq C\left(1+\sqrt{\log\frac{1}{r-r_+}}\right)
\end{align*}
for a constant $C>0$. Since $r\mapsto \sqrt{|\log (r-r_+)|}$ is integrable on compact subsets of $[r_+,\infty)$, we obtain $\norm{f_{a,n}}_{\uL^2(r_+,r_++\delta)}\rightarrow 0$ as $\delta\rightarrow 0$, a contradiction.
Hence we only need to exclude the case that the norm is bounded away from zero for large $r$. Thus, suppose that
\begin{align*}
\norm{f_{a,n}}_{\uL^2(R_0,\infty)}\geq \rho_2>0
\end{align*}
for infinitely many $f_{a,n}$ and any $R_0>0$.  However,
\begin{align*}
R_0\rho_2\leq\norm{f_{a,n}}_{L^2(R_0,\infty)}\leq C'
\end{align*}
for a constant $C'>0$ by (\ref{eqn:boundLmuD}) and any $R_0$, a contradiction. This shows that
\begin{align*}
\nu_a=\inf\{\LL_a(f)\,:\,f\in\uH_0^1,\,\norm{f}_{\uL^2}=1\}.
\end{align*}

By the infimum property, we have
\begin{align*}
\nu_{a}\leq\LL_{a}(f_{a}).
\end{align*}
By Lemma~\ref{lemma:SemictsD}, we get
\begin{align*}
\LL_{a}(f_{a})\leq\liminf_{n\rightarrow\infty}\LL_{a}(f_{a,n}).
\end{align*}
As
\begin{align*}
\LL_{a}(f_{a,n})\rightarrow\nu_{a},
\end{align*}
the latter equals $\nu_{a}$.
Thus the minimum is attained by $f_{a}$.
\end{proof}

We would like to derive the Euler-Lagrange equation corresponding to this minimiser.

\begin{lemma}
	\label{lemma:ELD}
	The minimiser $f_{a}$ satisfies
	\begin{align}
	\label{eqn:ELregD}
	\int_{r_+}^{\infty}\bigg(\frac{\Delta_-}{r^2+a^2}\frac{\dd f_{a}}{\dd r}\frac{\dd\psi}{\dd r}+\tilde V_a\frac{r^2+a^2}{\Delta_-}f_{a}\psi\bigg)\,\dd r=-\nu_a\int_{r_+}^{\infty}\frac{f_a}{r^2}\psi\,\dd r
	\end{align}
	for all $\psi\in\uH_0^1(r_+,\infty)$.
\end{lemma}

\begin{proof}
The proof can be extracted from \citep{Evans}. The analogous proof for twisted derivatives is given for Lemma~\ref{lemma:ELDN}. 
	\end{proof}

\begin{propn}
	\label{propn:a_hat}
	There is an $\hat a$ and a corresponding non-zero function $f_{\hat a}\in C^{\infty}(r_+,\infty)$ such that 
	\begin{align*}
	\frac{\Delta_-}{r^2+\hat a^2}\frac{\dd}{\dd r}\left(\frac{\Delta_-}{r^2+\hat a^2}\frac{\dd f_{\hat a}}{\dd r}\right)-\tilde V_{\hat a}f_{\hat a}=0
	\end{align*}
	and $f_{\hat a}$ satisfies the horizon regularity condition and the Dirichlet boundary condition at infinity.
\end{propn}

\begin{proof}
	First we would like to show that $\nu_a$ is continuous in $a$. We will use the notation $\Delta_-^a$ to denote the $\Delta_-$ corresponding to $a$. Given $a_1$ and $a_2$, we have
	\begin{align*}
	\nu_{a_1}&=\LL_{a_1}(f_{a_1})\\
	&=\int_{r_+}^{\infty}\left(\frac{\Delta_-^{a_2}}{r^2+a_2^2}\left\lvert\frac{\dd f_{a_1}}{\dd r}\right\lvert^2+\tilde V_{a_2}\frac{r^2+a_2^2}{\Delta_-^{a_2}}|f_{a_1}|^2\right)\,\dd r\\
	&~~~~+\int_{r_+}^{\infty}\left[\left(\frac{\Delta_-^{a_1}}{r^2+a_1^2}-\frac{\Delta_-^{a_2}}{r^2+a_2^2}\right)\left\lvert\frac{\dd f_{a_1}}{\dd r}\right\lvert^2+\left(\tilde V_{a_1}\frac{r^2+a_1^2}{\Delta_-^{a_1}}-\tilde V_{a_2}\frac{r^2+a_2^2}{\Delta_-^{a_2}}\right)|f_{a_1}|^2\right]\,\dd r.
	\end{align*}
	Due to the continuity of $\LL_a(f)$ in $a$, the first line is greater or equal than $\nu_{a_2}$ if $a_1$ is sufficiently close to $a_2$. Since the coefficients in the second line are continuously differentiable in $a$, we can use the mean value theorem to obtain
	\begin{align*}
	\nu_{a_1}\geq \nu_{a_2}-C|a_1-a_2|\int_{r_+}^{\infty}\left((r-r_+)\left\lvert\frac{\dd f_{a_1}}{\dd r}\right\lvert^2+|f_{a_1}|^2\right)\,\dd r
	\end{align*}
	for some constant $C>0$. We obtain an analogous inequality reversing the r\^{o}les of $a_1$ and $a_2$. Using (\ref{eqn:boundLmuD}) and $\norm{f_{a}}_{\uL^2}=1$ yields 
	\begin{align*}
	|\nu_{a_1}-\nu_{a_2}|&\leq C|a_1-a_2|\int_{r_+}^{\infty}\left((r-r_+)\left\lvert\frac{\dd f_{a_1}}{\dd r}\right\lvert^2+|f_{a_1}|^2\right)\,\dd r\leq C'|a_1-a_2|.
	\end{align*}
	Since $\Aa\neq \emptyset$, we  set
	\begin{align*}
	\hat a:=\inf\Aa.
	\end{align*}
	As stated in the introduction to this section, $\Aa$ is open, so $\hat a\notin\Aa$. By continuity of $\nu_a$, this implies that $\nu_{\hat a}=0$.
	
	Now choose a sequence $a_n\rightarrow \hat a$ and corresponding minimisers $f_{a_n}\in\uH^1_0$ satisfying $\norm{f_{a_n}}_{L^2}=1$. Then, as in the proof of Lemma~\ref{lemma:RegMinD}, by Lemma~\ref{lemma:VariationInequalityD}, $f_{a_n}$ is bounded in $\uH^1$ and there is a subsequence (also denoted $(a_n)$) such that $f_{a_n}\rightarrow f_{\hat a}$ weakly in $\uH^1$ and strongly in $L^2$ on compact subsets for a $f_{\hat a}\in\uH^1_0$. Again by Lemma~\ref{lemma:VariationInequalityD} and the strong $L^2$ convergence on compact subsets, we see that $f_{\hat a}$ is non-zero. Moreover, we have sufficient decay towards infinity by Lemma~\ref{lemma:trace}. Hence we get the desired asymptotics.
	
	From the weak convergence of $(f_{a_n})$, Lemma~\ref{lemma:ELD} yields that $f_{\hat a}$ satisfies
	\begin{align*}
	\int_{r_+}^{\infty}\bigg(\frac{\Delta_-}{r^2+\hat a^2}\frac{\dd f_{\hat a}}{\dd r}\frac{\dd\psi}{\dd r}+\tilde V_{\hat a}\frac{r^2+\hat a^2}{\Delta_-}f_{\hat a}\psi\bigg)\,\dd r=0
	\end{align*}
	for all $\psi\in\uH_0^1(r_+,\infty)$.  For ordinary differential equations, weak solutions are classical solutions -- see for example \citep[][Ch.~1]{TaoDispersive}; so
	\begin{align*}
	\frac{\Delta_-}{r^2+\hat a^2}\frac{\dd}{\dd r}\left(\frac{\Delta_-}{r^2+\hat a^2}\frac{\dd f_{\hat a}}{\dd r}\right)-\tilde V_{\hat a}f_{\hat a}=0,
	\end{align*}
	from which we obtain that $f_{\hat a}\in C^{\infty}$.
	
	It remains to check the boundary condition at the horizon. The lower semi-continuity of convex functionals with respect to weak convergence implies that
	\begin{align*}
	\int_{r_+}^{\infty}\left(\frac{\Delta_-}{r^2+a_n^2} \left\lvert\frac{\dd f_{a_n}}{\dd r}\right\lvert^2+\tilde V_{a_n}\frac{r^2+a_n^2}{\Delta_-}|f_{a_n}|^2\right)\,\dd r\leq\nu_{a_n},
	\end{align*}
	whence
	\begin{align*}
	\int_{r_+}^{\infty}\left(\frac{\Delta_-}{r^2+\hat a^2} \left\lvert\frac{\dd f_{\hat a}}{\dd r}\right\lvert^2+\tilde V_{\hat a}\frac{r^2+\hat a^2}{\Delta_-}|f_{\hat a}|^2\right)\,\dd r\leq 0.
	\end{align*}
	Hence
	\begin{align}
	\label{eqn:Boundfa}
	\int_{r_+}^{\infty}\frac{\Delta_-}{r^2+\hat a^2}\left\lvert\frac{\dd f_{\hat a}}{\dd r}\right\lvert^2\,\dd r<\infty.
	\end{align}
	Near $r_+$, the local theory (Theorem~\ref{thm:RegularSing}) implies that there exist constants $A$, $B$ and non-zero analytic functions $\phi_i$ such that
	\begin{align*}
	f_{\hat a}=A\phi_1+B(\log(r-r_+)\phi_2+\phi_3).
	\end{align*}
	If $B\neq 0$, then
	\begin{align*}
	\int_{r_+}^{\infty}\frac{\Delta_-}{r^2+\hat a^2}\left\lvert\frac{\dd f_{\hat a}}{\dd r}\right\lvert^2\,\dd r=\infty,
	\end{align*}
	whence $B=0$. Hence $f_{\hat a}$ satisfies the horizon regularity condition.
\end{proof}

\begin{rk}
	\label{rk:HR_positivity}
	From Lemma~\ref{lemma:HR_positivity}, we already know that $|\hat a|\geq r_+^2/\ell$. In \citep{HolzegelSmuleviciDecay}, it is shown directly that, if the Hawking-Reall bound is satisfied, there are no periodic solutions. One can easily see that the proof generalises to the case when the Hawking-Reall bound is saturated. Thus we even obtain $|\hat a|>r_+^2/\ell$.
\end{rk}

\begin{cor}
	Assume our choice of parameters, $a=\hat a$ and $\omega=\omega_+$. Let $C_0\in\CC$. Then the radial ODE (\ref{eqn:radial_ODE_prelim}) has a unique solution satisfying $u(r_+)=C_0$ and the Dirichlet boundary condition at infinity.
\end{cor}

\subsection{Perturbing the Dirichlet modes into the complex plane}
\label{subsec:PerturbingD}

We have shown that, for given $\ell>0$ and $\alpha<9/4$, there exists a real mode solution in a Kerr-AdS spacetime with parameters $(\ell,r_+,\hat a)$ and $\omega=\omega_R(0):=\Xi \hat am/(r_+^2+\hat a^2)$. Henceforth, we shall denote the chosen $\hat a$ simply by $a$.
Now we wish to vary $\omega$ and $\alpha$, keeping all the other parameters constant. 
Keeping $u(r_+,\omega,\alpha)$ fixed, satisfying $|u|(r_+,\omega,\alpha)=1$, the local theory yields a unique solution to the radial ODE of the form
\begin{align*}
u(r,\alpha,\omega)=A(\alpha,\omega)h_1(r,\alpha,\omega)+B(\alpha,\omega)h_2(r,\alpha,\omega)
\end{align*}
for large $r$, cf. Lemma~\ref{lemma:uniqueness_ODE} and (\ref{eqn:reflection_transmission}).
The functions $A$ and $B$ are smooth in $\omega$ and $\alpha$. Finding a mode solution is equivalent to finding a zero of $A$. We already have $A(\alpha(0),\omega_R(0))=0$. Write $A=A_R+\im A_I$.
Recall
\begin{align*}
 Q_T(r)=\Im(u'\overline u)
\end{align*}
and that
\begin{align*}
 \frac{\dd Q_T}{\dd r}(r)&=\frac{r^2+a^2}{\Delta_-}\Im\left(V-\omega^2\right)\\
 Q_T(r_+)&=\Xi am-\omega_R(r_+^2+a^2),
\end{align*}
where we have used $|u(r_+)|=1$.
We have
\begin{align*}
 Q_T(r)&=|A|^2\frac{\Delta_-}{r^2+a^2}\Im\left(\frac{\dd h_1}{\dd r}\overline{h_1}\right)+\frac{\Delta_-}{r^2+a^2}\Im\left(A\frac{\dd h_1}{\dd r}\overline{B h_2}\right)\\
    &~~~~~~~~~+\frac{\Delta_-}{r^2+a^2}\Im\left(B\frac{\dd h_2}{\dd r}\overline{A h_1}\right)+|B|^2\frac{\Delta_-}{r^2+a^2}\Im\left(\frac{\dd h_2}{\dd r}\overline{h_2}\right)
\end{align*}
and hence
\begin{align*}
  Q_T(\infty)&=\frac{1}{\ell^2}\left(-\frac{1}{2}+\sqrt{\frac{9}{4}-\alpha}\right)\Im(A\overline B)+\frac{1}{\ell^2}\left(-\frac{1}{2}-\sqrt{\frac{9}{4}-\alpha}\right)\Im(B\overline A)\\
  &=\frac{2}{\ell^2}\left(\frac{9}{4}-\alpha\right)^{1/2}\Im(A\overline B).
\end{align*}
due to the asymptotics of the $h_i$. We obtain
\begin{align*}
 \Xi am-(r_+^2+a^2)\omega_R+\int_{r_+}^{\infty}\frac{r^2+a^2}{\Delta_-}\Im\left(V-\omega^2\right)\,\dd r=\frac{2}{\ell^2}\left(\frac{9}{4}-\alpha\right)^{1/2}\Im(A\overline B).
\end{align*}
Now we differentiate at $\omega_R=\omega_R(0)$ and $\alpha=\alpha_0$ with respect to $\omega_R$ and $\alpha$:
\begin{align}
\label{eqn:det1}
 -(r_+^2+a^2)&=\frac{2}{\ell^2}\left(\frac{9}{4}-\alpha\right)^{1/2}\Im\left(\frac{\del A}{\del\omega_R}\overline B\right)
	=\frac{2}{\ell^2}\left(\frac{9}{4}-\alpha\right)^{1/2}\left(\frac{\del A_I}{\del\omega_R}B_R-\frac{\del A_R}{\del\omega_R}B_I\right)\\
\label{eqn:det2}
	    0&=\frac{2}{\ell^2}\left(\frac{9}{4}-\alpha\right)^{1/2}\Im\left(\frac{\del A}{\del\alpha}\overline B\right)
	    =\frac{2}{\ell^2}\left(\frac{9}{4}-\alpha\right)^{1/2}\left(\frac{\del A_I}{\del\alpha}B_R-\frac{\del A_R}{\del\alpha}B_I\right)
\end{align}
To extend the coefficient $A(\alpha,\omega_R)=0$ to complex $\omega$, we want to appeal to the implicit function theorem establishing
\begin{align*}
\det\begin{pmatrix}
      \frac{\del A_R}{\del\omega_R}	& \frac{\del A_R}{\del\alpha}\\
      \frac{\del A_I}{\del\omega_R}	& \frac{\del A_I}{\del\alpha}
    \end{pmatrix}
\neq 0.
\end{align*}
From equations (\ref{eqn:det1}) and (\ref{eqn:det2}) we see that this holds if
\begin{align*}
 \frac{\del A}{\del\alpha}(\alpha(0),\omega_R(0))\neq 0.
\end{align*}
This is true indeed:
\begin{lemma}
	\label{lemma:del_A__del_alphe}
\begin{align*} \frac{\del A}{\del\alpha}(\alpha(0),\omega_R(0))\neq 0.\end{align*}
\end{lemma}

\begin{proof}
 Suppose $\del A/\del\alpha=0$. Then we have
\begin{align*}
 \frac{\del u}{\del\alpha}(r,\omega_R(0),\alpha(0))&=\frac{\del B}{\del\alpha}(\omega_R(0),\alpha(0))h_2(r,\omega_R(0),\alpha(0))\\
 &~~~~~~~+B(\alpha(0),\omega_R(0))\frac{\del h_2}{\del\alpha}(r,\omega_R(0),\alpha(0)).
\end{align*}
Thus $\del u/\del \alpha$ is polynomially decreasing at infinity as $r^{-1/2-\sqrt{9/4-\alpha(0)}}$ and extends smoothly to $r=r_+$.
Defining the derivative $u_{\alpha}:=\del u/\del\alpha$, we get from the radial ODE
\begin{align*}
\frac{\Delta_-}{r^2+a^2}\frac{\dd}{\dd r}\left(\frac{\Delta_-}{r^2+a^2}\frac{\dd u_{\alpha}}{\dd r}\right)-\tilde Vu_{\alpha}=\left[\frac{\Delta_-}{(r^2+a^2)^2}\frac{\del\lambda}{\del\alpha}-\frac{1}{\ell^2}\frac{\Delta_-}{(r^2+a^2)^2}(r^2+\Theta(\alpha)a^2)\right]u.
\end{align*}
Multiplying by $\overline u$ and integrating by parts, we obtain at $\omega_R(0)$ and $\alpha(0)$
\begin{align}
\label{eqn:IBPLambda}
 \int_{r_+}^{\infty}\frac{\Delta_-}{(r^2+a^2)^2}\left(\frac{\del\lambda}{\del\alpha}-\frac{1}{\ell^2}(r^2+\Theta(\alpha)a^2)\right)|u|^2\,\dd r=0.
\end{align}
Now the two cases $\alpha\leq 0$ and $0<\alpha<9/4$ have to be treated separately. If $\alpha\leq 0$, then Proposition~\ref{propn:LambdaDeriv} readily gives $\del\lambda/\del\alpha<0$, so that $u$ would vanish identically.

For $\alpha>0$, we need to use the formula for $\del\lambda/\del\alpha$
from Proposition~\ref{propn:LambdaDeriv}.
Together with (\ref{eqn:IBPLambda}), this yields
\begin{align*}
\int_{r_+}^{\infty}\frac{\Delta_-}{(r^2+a^2)^2}\int_0^{\pi}\frac{1}{\ell^2}\left(-r^2-a^2\cos^2\theta\right)|S|^2\sin\theta|u|^2\,\dd\theta\,\dd r=0,
\end{align*}
whence we get the same contradiction.
\end{proof}

\subsection{Behaviour for small $\epsilon>0$ for Dirichlet boundary conditions}
\label{subsec:Crossing}

From the analysis of the previous section, we have a family of mode solutions $u(r,\epsilon)$ to the radial ODE parameters $(\omega(\epsilon),m,l,\alpha(\epsilon))$, where 
\begin{align*}
 \omega(\epsilon)=\omega_R(\epsilon)+\im\epsilon.
\end{align*}
The mode $u$ satisfies the horizon regularity condition and the Dirichlet boundary condition at infinity.
This proves the first part of Theorem~\ref{thm:oldD}. To prove the second part, we would like to study the behaviour of $\omega(\epsilon)$ and $\alpha(\epsilon)$ for small $\epsilon>0$.

To obtain the following statements, we potentially need to make $|m|$ even larger than in the previous sections.

\begin{propn}
	\label{lemma:DelOmega}
 If $|m|$ is sufficiently large, we have
\begin{align*}
 \omega_R(0)\frac{\del\omega_R}{\del\epsilon}(0)< 0.
\end{align*}
\end{propn}

\begin{proof}
 Define
\begin{align*}
 \tilde Q_T:=\Im\left(u'\overline{\omega u}\right).
\end{align*}
Let $\epsilon>0$. We have $\tilde Q_T(\infty)=0$. Moreover
\begin{align*}
  \tilde Q_T(r_+)=\Im\left(\frac{\xi}{r_+^2+a^2}\overline{\omega}\right)|u|^2(r_+)=0
\end{align*}
since $\xi$ has a positive real part (see (\ref{eqn:defn_xi})), $u\sim(r-r_+)^{\xi}$ and hence $|u|(r_+)=0$. Furthermore, using the radial ODE, one computes
\begin{align*}
 \frac{\dd\tilde Q_T}{\dd r}=-\epsilon\frac{\Delta_-}{r^2+a^2}\left\lvert\frac{\dd u}{\dd r}\right\lvert^2+\frac{r^2+a^2}{\Delta_-}\Im\left((V_{a}-\omega^2)\overline{\omega}\right)|u|^2.
\end{align*}
Hence
\begin{align}
\label{eqn:int_exact_equality}
 \int_{r_+}^{\infty}\left(\epsilon\frac{\Delta_-}{r^2+a^2}\left\lvert\frac{\dd u}{\dd r}\right\lvert^2-\frac{r^2+a^2}{\Delta_-}\Im\left((V_{a}-\omega^2)\overline{\omega}\right)|u|^2\right)\,\dd r=0
\end{align}
with
\begin{align*}
 -\Im((V_{a}-\omega^2)\overline{\omega})&=\frac{\epsilon}{(r^2+a^2)^2} \Big(V_+(r^2+a^2)^2+|\omega|^2(r^2+a^2)^2-\Xi^2a^2m^2\\&~~~~~~-\frac{\alpha}{\ell^2}\Delta_-(r^2+a^2\Theta(\alpha))\Big)-\frac{\Delta_-}{(r^2+a^2)^2}\Im((\lambda+a^2\omega^2)\overline\omega).
\end{align*}
From Proposition~\ref{propn:ImOmega}, we know that $-\Im(\lambda\overline\omega)>0$. Hence
\begin{align}
\begin{split}
\label{eqn:Im_V}
-\Im((V_{a}-\omega^2)\overline{\omega})&>\frac{\epsilon}{(r^2+a^2)^2} \Big(V_+(r^2+a^2)^2+|\omega|^2(r^2+a^2)^2-\Xi^2a^2m^2\\&~~~~~~~-\frac{\alpha}{\ell^2}\Delta_-(r^2+a^2\Theta(\alpha))\Big)-\frac{\Delta_-}{(r^2+a^2)^2}a^2\epsilon|\omega|^2.
\end{split}
\end{align}
We set
\begin{align*}
K(r):=|\omega|^2(r^2+a^2)^2-\Xi^2a^2m^2-\Delta_-a^2|\omega|^2.
\end{align*}
We have
\begin{align}
\begin{split}
\label{eqn:derivative_K}
\frac{\dd}{\dd r}K(r) 
&=|\omega|^2\left(4\left(1-\frac{a^2}{\ell^2}\right)r^3+2a^2M+2a^2\left(1-\frac{a^2}{\ell^2}\right)r\right)>0.
\end{split}
\end{align}
As already used in Section~\ref{sec:Real}, there is an $R>r_+$ such that, for $r\geq R$,
\begin{align*}
V_++V_{\alpha}>-\frac{1}{4\ell^2}\frac{\Delta_-}{r^2+a^2}.
\end{align*}
By an application of Lemma~\ref{lemma:Hardy}, we conclude
\begin{align}
\begin{split}
\label{eqn:int_K_large_R}
&\int_{R}^{\infty}\left(\epsilon\frac{\Delta_-}{r^2+a^2}\left\lvert\frac{\dd u}{\dd r}\right\lvert^2-\frac{r^2+a^2}{\Delta_-}\Im\left((V_{a}-\omega^2)\overline{\omega}\right)|u|^2\right)\,\dd r\\&~~~~~~~~~~~~>\int_R^{\infty}\frac{\epsilon}{(r^2+a^2)^2}K(r)|u|^2\,\dd r.
\end{split}
\end{align}
For the sake of contradition, suppose $K(r_+)\geq 0$. Then, by (\ref{eqn:derivative_K}), $K>0$ on $(r_+,\infty)$, whence we obtain strict positivity for (\ref{eqn:int_K_large_R}). 
As
\begin{align*}
|\omega(0)|^2=\frac{m^2a^2\Xi^2}{(r_+^2+a^2)^2}
\end{align*}
and as, for fixed $r_+$, $\hat a$ is bounded away from zero for all $m$,
\begin{align*}
|\omega(0)|^2\geq Cm^2.
\end{align*}
Since $\epsilon\mapsto\omega(\epsilon)$ is continuous, $|\omega|^2$ scales as $m^2$ for small $\epsilon$, so $\dd K/\dd r$ can be chosen as large as possible at $r=r_+$, in particular, it can be used to overcome the potentially non-positive derivative of the remaining terms of the right hand side of (\ref{eqn:Im_V}) on $(r_+,R)$. Then, (\ref{eqn:int_exact_equality}) implies $u=0$, a contradiction.

Hence $K(r_+)<0$ which is equivalent to
\begin{align*}
\omega_R(\epsilon)^2+\epsilon^2<\left(\frac{am}{r_+^2+a^2}\right)^2.
\end{align*}
This in turn is equivalent to the claim.
\end{proof}

In the following, we will fix an $|m|\geq m_0$ such that Proposition~\ref{lemma:DelOmega} holds.

\begin{rk}
	\label{rk:large_m}
	The choice of $m$ could have been made right at the beginning as the choice of $m_0$ in Lemma~\ref{lemma:Vneg} is independent of the largeness required for Proposition~\ref{lemma:DelOmega}.
\end{rk}

The next proposition shows that the mass $\alpha$ is at first increasing along the curve obtained by the implicit function theorem. The proof requires a technical lemma which is given at the end of this section.

\begin{propn}
	\label{propn:alpha_Dirichlet}
	Let $\alpha(0)<9/4$. Then
 \begin{align*}
  \frac{\del\alpha}{\del\epsilon}(0)>0.
 \end{align*}
\end{propn}

\begin{proof}
 Define
\begin{align*}
u_{\epsilon}=\frac{\del u}{\del\epsilon}.
\end{align*}
Then
\begin{align*}
 \frac{\del}{\del r}\left(\frac{\Delta_-}{r^2+a^2}\frac{\del u_{\epsilon}}{\del r}\right)-\frac{(V-\omega^2)(r^2+a^2)}{\Delta_-}u_{\epsilon}=\frac{\del}{\del\epsilon}\left(\frac{r^2+a^2}{\Delta_-}\tilde V\right)u.
\end{align*}
We would like to multiply this equation by $\overline u$ and then integrate by parts, but, at $r=r_+$, $u_{\epsilon}$ does not satisfy the boundary conditions of a mode solution. However, using $u\sim (r-r_+)^{-1/2-\kappa(\epsilon)}$ for all $\epsilon$ by Section~\ref{subsec:PerturbingD}, $u_{\epsilon}\sim \log r(r-r_+)^{-1/2-\kappa}$ and hence satisfies the Dirichlet boundary condition at infinity.

We know that
\begin{align*}
f(r,\epsilon):=\exp\left(-\im\frac{\Xi am-(r_+^2+a^2)\omega(\epsilon)}{\del_r\Delta_-}\log(r-r_+)\right)u(r,\epsilon)
\end{align*}
is smooth,
whence
\begin{align*}
\frac{\del f}{\del\epsilon}&=\im(r_+^2+a^2)\frac{\log(r-r_+)}{\del_r\Delta_-(r_+)}\left(\frac{\del\omega_R}{\del\epsilon}+\im\right)f(r,\epsilon)\\&~~~~~~~+\exp\left(-\im\frac{\Xi am-(r_+^2+a^2)\omega(\epsilon)}{\del_r\Delta_-(r_+)}\log(r-r_+)\right)u_{\epsilon}
\end{align*}
and
\begin{align*}
u_{\epsilon}(r,0)=\frac{r_+^2+a^2}{\del_r\Delta_-(r_+)}\left(1-\im\frac{\del\omega_R}{\del\epsilon}\right)\log(r-r_+)u+\frac{\del f}{\del\epsilon}(r,0).
\end{align*}
We have
\begin{align*}
\frac{\dd}{\dd r}\left(\frac{\Delta_-}{r^2+a^2}\frac{\dd u_{\epsilon}}{\dd r}\right)\overline u&=\frac{\dd}{\dd r}\left(\frac{\Delta_-}{r^2+a^2}\frac{\dd u_{\epsilon}}{\dd r}\overline u\right)-\frac{\dd}{\dd r}\left(\frac{\Delta_-}{r^2+a^2}u_{\epsilon}\frac{\dd \overline u}{\dd r}\right)\\&~~~~~~~~~~+u_{\epsilon}\frac{\dd}{\dd r}\left(\frac{\Delta_-}{r^2+a^2}\frac{\dd\overline u}{\dd r}\right)
\end{align*}
and
\begin{align*}
\frac{\dd u_{\epsilon}}{\dd r}\overline u-u_{\epsilon}\frac{\dd\overline u}{\dd r}&=\frac{r_+^2+a^2}{\del_r\Delta_-(r_+)}\left(1-\im\frac{\del\omega_R}{\del\epsilon}\right)\frac{1}{r-r_+}|u|^2\\&~~~~~~~~~~+\frac{r_+^2+a^2}{\del_r\Delta_-(r_+)}\left(1-\im\frac{\del\omega_R}{\del\epsilon}\right)\log(r-r_+)\Im\left(\frac{\dd u}{\dd r}\overline u\right).
\end{align*}
We conclude that
\begin{align*}
\frac{\Delta_-}{r^2+a^2}\log(r-r_+)\Im\left(\frac{\dd u}{\dd r}\overline u\right)
\end{align*}
is zero at $r=r_+$.
Thus evaluating the radial ODE at $\epsilon=0$, multiplying it by $\overline u$, taking real parts and integrating by parts yields
\begin{align}
\label{eqn:IntV}
-|u(r_+)|^2&=\int_{r_+}^{\infty}\frac{r^2+a^2}{\Delta_-}\Re\left(\frac{\del\tilde V}{\del\epsilon}\right)\Bigg\lvert_{\epsilon=0}|u|^2\,\dd r
\end{align}
For $\alpha\neq 0$, the derivative is given by
\begin{align}
\begin{split}
\label{eqn:real_part_V_epsilon}
\Re\left(\frac{\del\tilde V}{\del\epsilon}\right)\Bigg\lvert_{\epsilon=0} &=\frac{\Delta_-}{(r^2+a^2)^2}\left[\Re\left(\frac{\del\lambda}{\del\epsilon}\right)+2a^2\omega_R(0)\frac{\del\omega_R}{\del\epsilon}(0)\right]\\
	&~~~~-\frac{\Delta_-}{(r^2+a^2)^2}2ma\Xi\frac{\del\omega_R}{\del\epsilon}(0)-2\omega_R(0)\frac{\del\omega_R}{\del\epsilon}(0)\frac{r^2-r_+^2}{r^2+a^2}\\
	&~~~~-\frac{1}{\ell^2}\frac{\del\alpha}{\del\epsilon}\frac{\Delta_-}{(r^2+a^2)^2}(r^2+\Theta(\alpha)a^2)\\
	&=\frac{\Delta_-}{(r^2+a^2)^2}\left[\Re\left(\frac{\del\lambda}{\del\epsilon}-2r_+^2\omega_R(0)\frac{\del\omega_R}{\del\epsilon}(0)\right)-\frac{1}{\ell^2}\frac{\del\alpha}{\del\epsilon}(r^2+\Theta(\alpha)a^2)\right]\\
	&~~~~~~~~~~~-2\omega_R(0)\frac{\del\omega_R}{\del\epsilon}(0)\frac{r^2-r_+^2}{r^2+a^2}.
\end{split}
\end{align}
Noting that $\int_0^{\pi}|S|^2\sin\theta\,\dd\theta=1$ and using Lemma~\ref{propn:IfThen} to eliminate the dependence on $\lambda$ and then Proposition~\ref{lemma:DelOmega}, we conclude that $\del\alpha/\del\epsilon(0)$ needs to be positive to make the integrand of (\ref{eqn:IntV}) negative. The restriction to $\alpha\neq 0$ can by removed by continuity of the reflection and transmission coefficients $A$ and $B$.
\end{proof}

\begin{lemma}
	\label{propn:IfThen}
	At $\epsilon=0$, for $\alpha\leq 0$,
	\begin{align*}
	&\int_0^{\pi}\Bigg(2\left[\Xi a^2+(r_+^2+a^2)\frac{a^2}{\ell^2}\right]\frac{\cos^2\theta}{\Delta_{\theta}}\omega_R\frac{\del\omega_R}{\del\epsilon}+\frac{a^2}{\ell^2}\cos^2\theta\frac{\del\alpha}{\del\epsilon}+\Re\left(\frac{\del\lambda}{\del\epsilon}\right)\Bigg)|S|^2\sin\theta\,\dd\theta=0
	\end{align*}
	and, for $\alpha>0$,
	\begin{align*}
	&\int_0^{\pi}\Bigg(2\left[\Xi a^2+(r_+^2+a^2)\frac{a^2}{\ell^2}\right]\frac{\cos^2\theta}{\Delta_{\theta}}\omega_R\frac{\del\omega_R}{\del\epsilon}-\frac{a^2}{\ell^2}\sin^2\theta\frac{\del\alpha}{\del\epsilon}+\Re\left(\frac{\del\lambda}{\del\epsilon}\right)\Big)|S|^2\sin\theta\,\dd\theta=0.
	\end{align*}
\end{lemma}

\begin{proof}
	Let $\alpha\geq 0$. Set $S_{\epsilon}:=\del S/\del\epsilon$. Then, differentiating the angular ODE with respect to $\epsilon$,
	evaluating at $\epsilon=0$, multiplying by $\overline S$, taking the real part and integrating by part yields the claimed identity.
	An analogous computation yields the result for $\alpha>0$.
\end{proof}

\subsection{A continuity argument}
\label{subsec:Continuity}

We now deduce Theorem~\ref{thm:new} from Theorem~\ref{thm:oldD}. In this section, we fix $\ell>0$ and $\alpha_0<9/4$. In the previous sections, we have produced a curve $\epsilon\mapsto\alpha_0(\epsilon)$ of masses with $\del\alpha_0(0)/\del\epsilon>0$. This means that the constructed mode solutions will solve a radial ODE with a different scalar mass. 
This section formalises the intuitive idea of ``following up" the curves $\epsilon\mapsto\alpha(\epsilon)$ starting at an $\alpha$ close to $\alpha_0$ until one ``hits" the desired mass, which is made possible by $\del\alpha(0)/\del\epsilon>0$. The proof consists simply in establishing necessary continuity and carefully choosing neighbourhoods. This can be divided into two independent steps.
\begin{enumerate}
	\item We show that the function mapping $\alpha$ to the corresponding $\hat a$ is left-continuous.
	\item We show that, for $\alpha$ and corresponding $\hat a$ sufficiently close to $\alpha_0$ and the corresponding $\hat a_0$, the implicit function theorem guarantees a curve, starting at  $\alpha$ and the corresponding $\hat a$ and real frequency $\omega_+$, which exists ``long enough" to ``hit" $\alpha_0$.
\end{enumerate}

Note that, for any $f\in C_0^{\infty}$,$(\alpha,r_+,a)\mapsto \LL_{\alpha,r_+,a}(f)$
defines a continuous function. For a given $f\in C_0^{\infty}$, define the family of sets
\begin{align*}
\Aa_{\alpha,r_+}(f):=\{a>0\,:\,\LL_{\alpha,r_+,a}(f)<0\}
\end{align*}
and
\begin{align*}
\Aa_{\alpha,r_+}:=&\bigcup_{f\in C_0^{\infty}}\Aa_{\alpha,r_+}(f)\\
=&\{a>0\,:\,\exists f\in C_0^{\infty}:\,\LL_{\alpha,r_+,a}(f)<0\}.
\end{align*}

\begin{rk}
	$\Aa_{\alpha,r_+}$ corresponds to the set $\Aa$ from Section~\ref{sec:Real}. 
\end{rk}

Define the function
\begin{align*}
\Phi:\,(-\infty,9/4)\times(0,\infty)\rightarrow(0,\infty),~\Phi(\alpha,r_+):=\inf\Aa_{\alpha,r_+}
\end{align*}
if $\Aa_{\alpha,r_+}\neq\emptyset$.

\begin{lemma}
	\label{lemma:well-defined}
	Let $0<r_+<\ell$. Then there is an interval $I\subseteq(-\infty,9/4)$ with $\alpha_0\in I$  and an $m_0$ such that $\Phi(\cdot,r_+)$ is well-defined for all $\alpha\in I$ and $|m|\geq m_0$ .
\end{lemma}

\begin{proof}
	The set $\Aa_{\alpha,r_+}$ non-empty, open and bounded away from zero for all $\alpha\in I\subseteq (-\infty,9/4)$ by Remark~\ref{rk:Choice}, whence $\Phi(\cdot,r_+)$ is well-defined.
\end{proof}

We shall fix $r_+$ now. Moreover, we shall fix an $m\geq m_0>0$.

\begin{lemma}
	\label{lemma:non-increasing}
	The function $\Phi(\cdot,r_+)$ is non-increasing in $\alpha\in I$.
\end{lemma}

\begin{proof}
	Suppose $\Phi(\cdot,r_+)$ was not non-increasing. Then
	\begin{align*}
	\inf\Aa_{\alpha,r_+}<\inf\Aa_{\alpha',r_+}
	\end{align*}
	for some $\alpha<\alpha'$. Hence there is an $a>0$ with
	\begin{align*}
	\inf\Aa_{\alpha,r_+}<a<\inf\Aa_{\alpha',r_+}
	\end{align*}
	and an $f\in C_0^{\infty}$ such that
	\begin{align*}
	\LL_{\alpha,r_+,a}(f)<0\leq\LL_{\alpha',r_+,a}(f).
	\end{align*}
	This contradicts that $\LL_{\alpha,r_+,a}(f)>\LL_{\alpha',r_+,a}(f)$ for all $f$ and $a$ if $\alpha<\alpha'$.
\end{proof}

\begin{lemma}
	\label{lemma:CL}
	The function $\Phi(\cdot,r_+)$ is left-continuous at $\alpha_0$, i.\,e.
	\begin{align*}
	\lim_{\alpha\uparrow\alpha_0}\Phi(\alpha,r_+)=\Phi(\alpha_0,r_+).
	\end{align*}
\end{lemma}

\begin{proof}
	Suppose $\Phi(\cdot,r_+)$ was not left-continuous at $\alpha_0$. Then there is an $\epsilon>0$ such that, for all $\delta>0$, there is an $\alpha<\alpha_0$ with
	\begin{align*}
	\alpha_0-\alpha<\delta
	\end{align*}
	and
	\begin{align*}
	\Phi(\alpha,r_+)-\Phi(\alpha_0,r_+)\geq \epsilon.
	\end{align*}
	Then there is an $a$ between $\Phi(\alpha_0,r_+)$ and $\Phi(\alpha,r_+)$ such that there is an $f$ with $a\in\Aa_{\alpha_0,r_+}(f)$,
	but, for each $\delta$, there is an $\alpha$ with $a\notin\Aa_{\alpha,r_+}(f)$.
	Since $\LL_{\alpha_0,r_+,a}(f)<0$ and due to the continuity of $\LL_{\cdot}(f)$, there is a $\delta>0$ such that for all $\alpha_0-\alpha<\delta$, we have $\LL_{\alpha,r_+,a}(f)<0$, i.\,e. $a\in\Aa_{\alpha,r_+}(f)$, a contradiction.
\end{proof}

For $\alpha\in I$, we define
\begin{align}
\label{eqn:defn_omega_+}
\Omega_R(\alpha):=\frac{m\Phi(\alpha,r_+)\left(1-\frac{\Phi(\alpha,r_+)^2}{\ell^2}\right)}{r_+^2+\Phi(\alpha,r_+)^2}.
\end{align}
As shown, this is left-continuous at $\alpha_0$.

Now we turn to the second step. Recall that, for all $\alpha\in I$, there is a periodic Dirichlet mode with frequency $\omega=\Omega_R(\alpha)\in\RR$ in a Kerr-AdS spacetime with parameters $(\ell,r_+,\Phi(\alpha,r_+))$ by Proposition~\ref{propn:a_hat}. Using Section~\ref{subsec:PerturbingD}, we can find unstable Dirichlet mode solutions with frequency $\omega=\omega_R+\im\omega_I=\omega_R(\epsilon)+\im\epsilon$ (where $\omega_R(0)=\Omega_R(\alpha)$) to the Klein-Gordon equation with mass $\alpha(\epsilon)$ (where $\alpha(0)=\alpha_0$). As by Lemma~\ref{lemma:well-defined} the results of Section~\ref{subsec:Crossing} hold, we know that
\begin{align}
\label{eqn:sign_omega_alpha}
\frac{\del\alpha}{\del\epsilon}(0)>0,~~~~~~~~\frac{\del\omega_R}{\del\epsilon}(0)<0.
\end{align}

In this section, $B_{\rho}(x)$ will denote an open $\ell^{\infty}$ ball of radius $\rho$ centered around $x\in\RR^4$, i.\,e.
\begin{align*}
B_{\rho}(x):=\left\{y\in\RR^4\,:\,\max_{j=1,\ldots,4}|x_j-y_j|<\rho\right\}.
\end{align*}
We view the column vectors $(\alpha,\omega_R,\omega_I,a)^t$ as points in $\RR^4$.

Consider
\begin{align*}
D:=\det\begin{pmatrix}
\frac{\del A_R}{\del\omega_R}	& \frac{\del A_R}{\del\alpha}\\
\frac{\del A_I}{\del\omega_R}	& \frac{\del A_I}{\del\alpha}
\end{pmatrix}.
\end{align*}
It was shown in Section~\ref{subsec:PerturbingD} that $D(\alpha,\Omega_R(\alpha),0,\Phi(\alpha,r_+))\neq 0$ for all $\alpha\in I$. 

From Lemma~\ref{lemma:continuous_AB}, we know that $A$ is smooth in $\alpha$, $\omega$ and $a$. Hence there is an $L>0$ such that $D\neq 0$ in $B_L(\alpha,\Omega_R(\alpha_0),0,\Phi(\alpha_0,r_+))$ and such that, for all values $(\alpha,\Omega_R(\alpha),0,\Phi(\alpha,r_+))\in B_L(\alpha,\Omega_R(\alpha_0),0,\Phi(\alpha_0,r_+))$, we have $\alpha\in I$.

 Hence in this neighbourhood, the vector field
	\begin{align*}
	W:=-\begin{pmatrix}
	\frac{\del A_R}{\del\alpha} & \frac{\del A_R}{\del\omega_R} & 0 & 0\\
	\frac{\del A_I}{\del \alpha} & \frac{\del A_I}{\del\omega_R} & 0 & 0\\
	0 & 0 & -1 & 0\\
	0 & 0 & 0 & -1
	\end{pmatrix}^{-1}
	\begin{pmatrix}
	\frac{\del A_R}{\del\omega_I}\\
	\frac{\del A_I}{\del\omega_I}\\
	1\\
	0
	\end{pmatrix}
	\end{align*}
	is well-defined. It is this vector field whose integral curves describe the solutions given by the implicit function theorem as applied in Section~\ref{subsec:PerturbingD}. In particular, solving the ODE
	\begin{align*}
	\frac{\dd}{\dd \epsilon}
	(\alpha(\epsilon),\omega_R(\epsilon),\omega_I(\epsilon),a(\epsilon))^t
	=W(\alpha(\epsilon),\omega_R(\epsilon),\omega_I(\epsilon),a(\epsilon))
	\end{align*}
	with initial conditions $(\alpha,\Omega_R(\alpha),0,\Phi(\alpha,r_+))$ ($\alpha\in I$) gives the previously introduced $\alpha(\epsilon)$ and $\omega_R(\epsilon)$.
	
	Set $W:=(W^{\alpha},W^{\omega_R},W^{\omega_I},W^a)^t$. By (\ref{eqn:sign_omega_alpha}),
	\begin{align*}
	W^{\alpha}(\alpha_0,\Omega_R(\alpha_0),0,\Phi(\alpha_0,r_+))>0
	\end{align*}
	and
	\begin{align*}
	W^{\omega_R}(\alpha_0,\Omega_R(\alpha_0),0,\Phi(\alpha_0,r_+))<0.
	\end{align*}
	Let $\delta>0$. Again by smoothness of $A$, there are $\rho>0$ and $L'\leq L$ such that
	\begin{align*}
	\norm{W-W(\alpha_0,\Omega_R(\alpha_0),0,\Phi(\alpha_0,r_+))}_{\infty}<\delta~~\mathrm{and}~~
	W^{\alpha}\geq\rho,~W^{\omega_R}\leq -\rho
	\end{align*}
	in $B_{L'}(\alpha_0,\Omega_R(\alpha_0),0,\Phi(\alpha_0,r_+))$.
	
	We now study integral curves of $W$ in $B_{L'}(\alpha_0,\omega_R(\alpha_0),0,\Phi(\alpha_0,r_+))$. Let
	\begin{align*}
	\tau\mapsto\gamma(\tau,p)
	\end{align*}
	be the integral curve of $W$ with $\gamma(0,p)=p\in B_{L'}(\alpha_0,\Omega_R(\alpha_0),0,\Phi(\alpha_0,r_+))$.
	Define the map
	\begin{align*}
	&T:\,B_{2L'/3}(\alpha_0,\Omega_R(\alpha_0),0,\Phi(\alpha_0,r_+))\rightarrow\RR,\\
	&T(p)=\inf\left\{\tau>0\,:\,\gamma(\tau,p)\in\overline{B_{2L'/3}(\alpha_0,\Omega_R(\alpha_0),0,\Phi(\alpha_0,r_+))}^c\right\}.
	\end{align*}
	The set $\left\{\tau>0\,:\,\gamma(\tau,p)\in\overline{B_{2L'/3}(\alpha_0,\Omega_R(\alpha_0),0,\Phi(\alpha_0,r_+))}^c\right\}$ is non-empty since
	\begin{align}
	\label{eqn:del_omega_I}
	\frac{\dd}{\dd \tau}\omega_I(\tau)=1.
	\end{align}
	Therefore $T$ is well-defined.
	
	\begin{lemma}
		T is continuous.
	\end{lemma}
	
\begin{proof} For this proof use the abbreviation $B:=B_{2L'/3}(\alpha_0,\Omega_R(\alpha_0),0,\Phi(\alpha_0,r_+))$.
	Let $p_0\in B$, $\tau_0:=T(p_0)>0$. Let $0<\epsilon<\tau_0$ such that $\gamma(\tau,p_0)\in \overline{B}^c$ for $\tau_0+\epsilon\leq \tau\leq \tau_0+2\epsilon$, which exists by (\ref{eqn:del_omega_I}). Define
	\begin{align*}
	d_1:=\min\{\mathrm{dist}(\gamma(\tau,p_0),\del B)\,:\,0\leq \tau\leq \tau_0-\epsilon\}
	\end{align*}
	
	We claim that $d_1>0$. Suppose not. Then there is a $\tau'\in(0,\tau_0-\epsilon]$ such that $\mathrm{dist}(\gamma(\tau',p_0),\del B)=0$. Since
	\begin{align}
	\label{eqn:Vectorfield}
	W^{\alpha}\geq\rho,~~~W^{\omega_R}\leq-\rho,~~~W^{\omega_I}=1,~~~W^{a}=0,
	\end{align}
	whence $W$ is not parallel to any side of the boundary $\del B$ of the $\ell^{\infty}$ ball, this would imply that $\gamma(\tau,p_0)\in\overline B^c$ for a range of $\tau$'s in a small neighbourhood of $\tau'$. This, however, contradicts $\tau_0=T(p_0)$. Hence $d_1>0$.
	
	Furthermore define
	\begin{align*}
	d_2:=\min\{\mathrm{dist}(\gamma(\tau,p_0),\del B)\,:\,\tau_0+\epsilon\leq \tau\leq \tau_0+2\epsilon\}.
	\end{align*}
	Using (\ref{eqn:Vectorfield}), we can see again that $d_2>0$. Set $d:=\min(d_1,d_2)$.
	
	Set
	\begin{align*}
	G:=\{\gamma(\tau,p_0)\,:\,\tau\in[0,\tau_0+2\epsilon]\backslash[\tau_0-\epsilon,\tau_0+\epsilon]\}
	\end{align*}
	Since the solutions of linear ODEs depend continuously on the initial data, there is a $\delta>0$ such that, for all $p\in B_{\delta}(p_0)$, $\gamma(\tau,p)$ is in a $d/2$-neighbourhood of $G$ for all $\tau\in[0,\tau_0+2\epsilon]\backslash(\tau_0-\epsilon,\tau_0+\epsilon)$. Thus for all $p\in B_{\delta}(p_0)$, $T(p)>\tau_0-\epsilon$ and $T(p)<\tau_0+\epsilon$, i.\,e. $|T(p)-T(p_0)|<\epsilon$.
\end{proof}

	Hence there exists a $T_0\geq0$ such that
	\begin{align*}
	T(p)\geq T_0
	\end{align*}
	for all $p\in \overline{B_{L'/3}(\alpha_0,\Omega_R(\alpha_0),0,\Phi(\alpha_0,r_+))}$. As continuous functions attain their minimum on compact sets, $T_0$ can be chosen to be positive. This shows that all integral curves of $W$ starting in $B_{L'/3}(\alpha_0,\Omega_R(\alpha_0),0,\Phi(\alpha_0,r_+))$ exist for $0\leq \tau\leq T_0$ and remain in $B_{L'}(\alpha_0,\Omega_R(\alpha_0),0,\Phi(\alpha_0,r_+))$.
	
	We can prove the following
	\begin{lemma} Given $(\alpha(0),\omega_R(0),\omega_I(0),a(0))^t\in B_{L'/3}(\alpha_0,\Omega_R(\alpha_0),0,\Phi(\alpha_0,r_+))$ with
	\begin{align*}
	\alpha_0-T_0\rho\leq\alpha(0)<\alpha_0,
	\end{align*}
	let
	\begin{align*}
	s\mapsto (\alpha(s),\omega_R(s),\omega_I(s),a(s))^t
	\end{align*}
	be the integral curve starting at $(\alpha(0),\omega_R(0),\omega_I(0),a(0))^t$. Then there is a $\tau\in (0,T_0]$ such that $\alpha(\tau)=\alpha_0$.
	\end{lemma}
	\begin{proof} The ODE yields
	\begin{align*}
	\alpha(T_0)&=\alpha(0)+\int_0^{T_0}W^{\alpha}(\alpha(s),\omega_R(s),\omega_I(s),a(s))\,\dd s\\
		&\geq \alpha(0)+T_0\rho.
	\end{align*}
	If $\alpha_0-T_0\rho\leq\alpha(0)<\alpha_0$, then $\alpha(T_0)\geq\alpha_0$ and, by the intermediate value theorem, there is a $\tau\in (0,T_0]$ such that $\alpha(\tau)=\alpha_0$.
	\end{proof}
	
	The function $\Phi(\cdot,r_+)$ induces the curve
	\begin{align*}
	\Gamma:\,\alpha\mapsto(\alpha,\Omega_R(\alpha),0,\Phi(\alpha,r_+))^t
	\end{align*}
	for $\alpha\in I$; it is continuous on the left at $\alpha_0$. The result of the previous section says that along this curve, the implicit function theorem produces parameter curves that correspond to superradiant modes; these parameter curves are exactly the integral curves of $W$ starting on a point of $\Gamma$. Since $\Gamma$ is left-continuous,
	\begin{align*}
	\Gamma\cap B_{L'/3}(\alpha_0,\Omega_R(\alpha_0),0,\Phi(\alpha_0,r_+))\cap \{\alpha_0-T_0\rho\leq\alpha(0)<\alpha_0\}\neq\emptyset.
	\end{align*}
This shows Theorem~\ref{thm:new}

\section{Growing mode solutions satisfying Neumann boundary conditions}
\label{sec:Neumann_construction}

\subsection{Existence of real mode solutions}
\label{sec:RealDN}

In this section, we will construct growing mode solutions satisfying Neumann boundary conditions. Every result has a counterpart in Section~\ref{sec:Growing_Dirichlet}. In the following, whenever proofs will be short in detail, the reader can extract those from Section~\ref{sec:Growing_Dirichlet}. The two novel techniques in this section are the use of twisted derivatives with appropriately modified Sobolev spaces and a new Hardy inequality (Lemma~\ref{lemma:Hardy_twisted}). 

Fix $\ell>0$ and $r_+$.  In this section, we look at the range $5/4<\alpha<9/4$, i.\,e. $0<\kappa<1$, for Neumann boundary conditions.

To treat the Neumann case variationally, we need to modify the functional, so it becomes finite for Neumann modes. We achieve this by conjugating the derivatives by a power of r; more precisely, we consider the twisted derivative $h\frac{\dd}{\dd r}\left(h^{-1}\cdot\right)$, where $h=r^{-1/2+\kappa}$. This ``kills off" the highest order term of the Neumann branch. Moreover, squaring the twisted derivative term does not introduce any ``mixed terms" in $f$ and its derivative; it only produces a zeroth order term that also makes the potential finite.

Thus introduce the twisted variational functional
\begin{align*}
\tL_a(f):=\int_{r_+}^{\infty}\left(\frac{\Delta_-}{r^2+a^2}h^2\left\lvert\frac{\dd}{\dd r}\left(h^{-1}f\right)\right\lvert^2+\Vmod_a\frac{r^2+a^2}{\Delta_-}|f|^2\right)\,\dd r,
\end{align*}
where $\Vmod_a$ as in Appendix~\ref{sec:twisted_derivative}, i.\,e.
\begin{align*}
\Vmod_a=\tilde V_a+\left(\frac{1}{2}-\kappa\right)\frac{\Delta_-}{r^2+a^2}r^{\frac{1}{2}-\kappa}\frac{\dd}{\dd r}\left(\frac{\Delta_-}{r^2+a^2}r^{-\frac{3}{2}+\kappa}\right).
\end{align*}
By Lemma~\ref{lemma:twisting_positive_g}, $\Vmod_a=\OO(1)$ and $\Vmod_a$ is positive near infinity for sufficiently large $|m|$, which shall be assumed henceforth. Moreover, $\Vmod_a$ is chosen such that the twisted variational problem leads to the same Euler-Lagrange equation as the untwisted one.

For $U\subseteq(r_+,\infty)$, we define the twisted Sobolev norm
\begin{align*}
\norm{f}_{\uH_{\kappa}^1(U)}^2:=\int_U\left(\frac{1}{r^2}|f|^2+r(r-r_+)h^2\left\lvert\frac{\dd}{\dd r}\left(h^{-1}f\right)\right\lvert^2\right)\,\dd r.
\end{align*}
Note that for $U\subseteq(r_+,\infty)$ compact, the $\uH_{\kappa}^1$ norm is equivalent to the standard Sobolev norm.
For $U=(r_+,\infty)$, let $\uH^1_{\kappa}(U)$ be the completion of functions of the form
\begin{align}
\label{eqn:basis_twisting}
f(r)=r^{-\frac{1}{2}+\kappa}g(r)
\end{align}
under $\norm{\cdot}_{\uH_{\kappa}^1(U)}$, where $(x\mapsto g(1/x))\in C_0^{\infty}[0,1/r_+)$. Henceforth, we will sometimes refer to such a function $g$ as being ``compactly supported around infinity".

\begin{lemma}
	\label{lemma:SobolevNeumann}
	Let $f\in\uH^1_{\kappa}(r_+,\infty)$, then $f$ is also in $C(r_++1,\infty)$ and $r^{1/2-\kappa}f(r)$ is bounded.
\end{lemma}

\begin{proof}
	The existence of a continuous version follows from Sobolev embedding as in Lemma~\ref{lemma:trace}. Then, there exists a sequence $(f_n)\in C^{\infty}$ as in the definition such that $f_n\rightarrow f$ in $\uH^1_{\kappa}$. Let $\tilde R>R>r_+$ and let $f_n(R)$ converge to $f(R)$:
	\begin{align*}
	\left\lvert \tilde R^{1/2-\kappa}(f_n(\tilde R)-f(\tilde R))\right\lvert&\leq \left\lvert r^{1/2-\kappa}(f_n(R)-f(R))\right\lvert+\int_R^{\tilde R}\left\lvert\frac{\dd}{\dd r}\left(r^{1/2-\kappa}\left(f_n-f\right)\right)\right\lvert\,\dd r\\
	&\leq \left\lvert r^{1/2-\kappa}(f_n(R)-f(R))\right\lvert\\&~~~~~+\left(\int_R^{\infty}r(r-r_+)r^{-1+2\kappa}\left\lvert\frac{\dd}{\dd r}\left(f_n-f\right)\right\lvert^2\,\dd r\right)^{1/2}\times\\&~~~~~~~~~~\times\left(\int_R^{\infty}\frac{1}{r(r-r_+)r^{-1+2\kappa}}\,\dd r\right)^{1/2}
	\end{align*}
	Hence $r^{1/2-\kappa}\left(f_n(r)-f(r)\right)$ converges uniformly for all $r\geq R$. Hence we even have convergence at $r=\infty$. Since $\lim_{r\rightarrow\infty}r^{1/2-\kappa}f_n(r)\neq \infty$ for all $n$, we obtain the result. 
\end{proof}

As in Section~\ref{sec:Real}, choose mode parameters such that the conditions for Lemma~\ref{lemma:FunctionalNegativeTwisted} are satisfied.
Let
\begin{align}
\label{eqn:defn_A}
\Aa:=\{a>0\,:\,\exists\,(x\mapsto g(1/x))\in C_0^{\infty}[0,1/r_+):\,\tilde{\LL}_a(r^{-1/2+\kappa}g)<0\}.
\end{align}
Note that $\Aa$ is non-empty, open and bounded below.

\begin{lemma}
	\label{lemma:Hardy_twisted}
	For $r_{\mathrm{cut}}> r_++1$, $0<\kappa<1$ and a smooth function $f$ with $fr^{1/2-\kappa}=\OO(1)$ at infinity, we have that
	\begin{align*}
	\int_{r_{\mathrm{cut}}}^{\infty}\frac{|f|^2}{r^2}\,\dd r&=\frac{1}{2(1-\kappa)}r_{\mathrm{cut}}^{-1}\left(\frac{c}{2}\int_{r_{\mathrm{cut}}-1}^{r_{\mathrm{cut}}}|f|^2\,\dd r+\frac{1}{2c}\int_{r_{\mathrm{cut}}-1}^{r_{\mathrm{cut}}}\left\lvert\frac{\dd f}{\dd r}\right\lvert^2\,\dd r\right)\\&~~~~~+\frac{1}{(1-\kappa)^2}\int_{r_{\mathrm{cut}}}^{\infty}r^{-1+2\kappa}\left\lvert\frac{\dd}{\dd r}\left(r^{1/2-\kappa}f\right)\right\lvert^2\,\dd r
	\end{align*}
	for any $c>0$ sufficiently large.
\end{lemma}

\begin{proof}
	We compute:
	\begin{align*}
	\int_{r_{\mathrm{cut}}}^{\infty}\frac{|f|^2}{r^2}\,\dd r&=\frac{1}{2(1-\kappa)}r_{\mathrm{cut}}^{-1}|f|^2(r_{\mathrm{cut}})\\&~~~~~+\frac{1}{1-\kappa}\int_{r_{\mathrm{cut}}}^{\infty}r^{2\kappa-2}\Re\left(r^{1/2-\kappa}f\frac{\dd}{\dd r}\left( r^{1/2-\kappa}\overline f\right)\right)\,\dd r.
	\end{align*}
	Using the Cauchy-Schwarz inequality, one easily sees that
	\begin{align*}
	\int_{r_{\mathrm{cut}}}^{\infty}\frac{|f|^2}{r^2}\,\dd r&\leq \frac{1}{1-\kappa}r_{\mathrm{cut}}^{-1}|f|^2(r_{\mathrm{cut}})\\&~~~~~+\frac{1}{(1-\kappa)^2}\int_{r_{\mathrm{cut}}}^{\infty}r^{-1+2\kappa}\left\lvert\frac{\dd}{\dd r}\left(r^{1/2-\kappa}f\right)\right\lvert^2\,\dd r.
	\end{align*}
	Let $\chi\geq 0$ be a smooth function of compact support with $\chi(r_{\mathrm{cut}})=1$ and $\chi=0$ for $r\leq r_{\mathrm{cut}-1}$. Then
	\begin{align*}
	|f|^2(r_{\mathrm{cut}})&\leq|f|^2(r_{\mathrm{cut}})\chi(r_{\mathrm{cut}})\\
	&\leq \frac{c}{2}\int_{r_{\mathrm{cut}}-1}^{r_{\mathrm{cut}}}|f|^2\,\dd r+\frac{1}{2c}\int_{r_{\mathrm{cut}}-1}^{r_{\mathrm{cut}}}\left\lvert\frac{\dd f}{\dd r}\right\lvert^2\,\dd r
	\end{align*}
	for any $c>0$ sufficiently large.
\end{proof}

\begin{lemma}
	\label{lemma:VariationInequalityND}
	Let $a\in \Aa$ be fixed. There exist constants $r_+<B_0<B_1<\infty$ and $C_0,C_1>0$, such that, for large enough $m$, we have for all smooth functions $f$ for which the following integrals are defined that
	\begin{align*}
	\int_{r_+}^{\infty}\left(\frac{\Delta_-}{r^2+a^2}h^2\left\lvert\frac{\dd}{\dd r}\left(h^{-1}f\right)\right\lvert^2+C_01_{[B_0,B_1]^c}\frac{|f|^2}{r^2}\right)\,\dd r\leq C_1\int_{B_0}^{B_1}|f|^2\,\dd r+2\tL_{a}(f).
	\end{align*}
\end{lemma}

\begin{proof}
	The proof follows the strategy of Lemma~\ref{lemma:VariationInequalityD}. The analysis of the potential goes through as in Section~\ref{sec:Real} as the twisting part of $\Vmod_a$ does not depend on $m$ and has the right asymptotics. Thus we know that there is an $R_1$ such that
	\begin{align*}
	\frac{r^2+a^2}{\Delta_-}\Vmod_a>0
	\end{align*}
	on $(r_+,R_1)$. Moreover, there is an $R_2>R_1$ such that
	\begin{align*}
	\frac{r^2+a^2}{\Delta_-}\Vmod_a>-\frac{C}{2r^2}~~~\mathrm{and}~~~
	\frac{C}{(1-\kappa)^2}<\frac{\Delta_-}{r^2+a^2}
	\end{align*}
	for $r\geq R_2$.
	Hence
	\begin{align*}
	\int_{R_2}^{\infty}\frac{r^2+a^2}{\Delta_-}\Vmod_a|f|^2\,\dd r&\geq -\frac{1}{2}\int_{R_2-1}^{\infty}\frac{\Delta_-}{r^2+a^2}r^{-1+2\kappa}\left\lvert\frac{\dd}{\dd r}\left(r^{1/2-\kappa}f\right)\right\lvert^2\,\dd r-C'\int_{R_2-1}^{R_2}|f|^2\,\dd r
	\end{align*}
	for some large constant $C'>0$ by the Hardy inequality of Lemma~\ref{lemma:Hardy_twisted}. Choosing $B_0$,$ B_1$, $C_0$ and $C_1$ appropriately (as in the proof of Lemma~\ref{lemma:VariationInequalityD}), we obtain the inequality.
\end{proof}

\begin{lemma}
	\label{lemma:Semicts}
	The functional $\tL_a$ is weakly lower semicontinuous in $\uH_{\kappa}^1(r_+,\infty)$.
\end{lemma}

\begin{proof}
	See the comments to Lemma~\ref{lemma:SemictsD}.
\end{proof}

\begin{lemma}
	\label{lemma:RegMinim}
	Let $a\in\Aa$. 	
	There exists an $f_{a}\in\uH_{\kappa}^1(r_+,\infty)$ with norm  $\norm{f_{a}}_{\uL^2(r_+,\infty)}=1$ such that $\tL_a$ achieves its infimum in
	\begin{align*}
	\uH_{\kappa}^1(r_+,\infty)\cap\{\norm{f}_{\uL^2(r_+,\infty)}=1\} 
	\end{align*}
	on $f_{a}$.
\end{lemma}

\begin{proof} 
	The proof is similar to the one in Section~\ref{sec:Real}. We obtain a minimising sequence $(f_{a,n})$ that converges weakly in $\uH_{\kappa}^1$ and strongly in $L^2$ on compact subsets of $(r_+,\infty)$. In analogy to the Dirichlet, the $f_{a,n}$ are can be taken from a dense subset and can be chosen to be of the form $f_{a,n}=r^{-1/2+\kappa}g_{a,n}$ for $g_{a,n}$ smooth and compactly supported around infinity.
	
	We will show that the norm is conserved. Suppose not. 
	Then, for any $N$, there are infinitely many of the $f_{a,n}$ such that
	\begin{align*}
	\norm{f_{a,n}}_{\uL^2((r_+,\infty)\backslash [r_++1/N,N])}\geq \rho>0.
	\end{align*}
	Suppose
	\begin{align*}
	\norm{f_{a,n}}_{\uL^2(r_+,r_++\delta)}\geq \rho_1>0
	\end{align*}
	for infinitely many $f_{a,n}$ and any $\delta>0$.
	Because of the $L^2$ convergence on compact subsets, there is an $R$ such that $f_{a,n}(R)\rightarrow f_a(R)$ as $n\rightarrow\infty$, in particular $f_{a,n}(R)$ is bounded for all $n$.	
	By Lemma~\ref{lemma:VariationInequalityND}, we have for $r\in(r_+,R)$:
	\begin{align*}
	|r^{1/2-\kappa}f_{a,n}(r)|&\leq \int_{r}^{R}\left\lvert\frac{\dd }{\dd r'}\left(r'^{1/2-\kappa}f_{a,n}\right)\right\lvert\,\dd r'+R^{1/2-\kappa}f_{a,n}(R)\\
	&\leq \left(\int_{r}^{R}\frac{1}{r'-r_+}\,\dd r'\right)^{1/2}\times\\&~~~~~~~~~~\times\left(\int_{r}^{R}(r'-r_+)r'^{-1+2\kappa}\left\lvert\frac{\dd}{\dd r'}\left(r'^{1/2-\kappa}f_{a,n}\right)\right\lvert^2\,\dd r'\right)^{1/2}\\&~~~~~~~~~~~~~~~+R^{1/2-\kappa}f_{a,n}(R)\\
	&\leq C\left(1+\sqrt{\log\frac{R-r_+}{r-r_+}}\right)
	\end{align*}
	for a constant $C>0$. Since $r\mapsto \sqrt{|\log (r-r_+)|}$ is integrable on compact subsets of $[r_+,\infty)$, we obtain $\norm{f_{a,n}}_{\uL^2(r_+,r_++\delta)}\rightarrow 0$ as $\delta\rightarrow 0$, a contradiction.

	Hence we only need to exclude the case that the norm is bounded away from zero for large $r$. Thus, suppose that
	\begin{align*}
	\norm{f_{a,n}}_{\uL^2(R_0,\infty)}\geq \rho_2>0
	\end{align*}
	for infinitely many $f_{a,n}$ and any $R_0>0$.	
	Since $f_{a,n}(r_+)=0$, we have
	\begin{align*}
	r^{\frac{1}{2}-\kappa}|f_{a,n}|(r)&\leq \int_{r_+}^{r}\left\lvert\frac{\dd}{\dd r}\left(r'^{\frac{1}{2}-\kappa}f_{a,n}\right)\right\lvert\,\dd r'\\&\leq \left(\int_{r_+}^{\infty}\frac{1}{r^{1+2\kappa}}\,\dd r\right)^{1/2}\left(\int_{r_+}^{\infty}r^{1+2\kappa}\left\lvert\frac{\dd}{\dd r}\left(r^{\frac{1}{2}-\kappa}f_{a,n}\right)\right\lvert\,\dd r\right)^{1/2},
	\end{align*}
	which is uniformly bounded for all $n$. Hence
	\begin{align*}
	\int_{R_0}^{\infty}\frac{|f_{a,n}|^2}{r^2}\,\dd r\leq C'\int_{R_0}^{\infty}r^{-3+2\kappa}\,\dd r\rightarrow 0
	\end{align*}
	as $R_0\rightarrow\infty$, a contradiction.
	
	As in the proof of Lemma~\ref{lemma:RegMinD}, we have
	\begin{align*}
	\nu_a\leq\LL_a(f_a)\leq\liminf_{n\rightarrow\infty}\LL_a(f_{a,n})=\nu_a
	\end{align*}
	and the rest follows.
\end{proof}

We would like to derive the Euler-Lagrange equation corresponding to this minimiser.

\begin{lemma}
	\label{lemma:ELDN}
	The minimiser  $f_{a}$ satisfies
	\begin{align}
	\begin{aligned}
	\label{eqn:ELregNeumann}
	\int_{r_+}^{\infty}\bigg(\frac{\Delta_-}{r^2+a^2}h^2\frac{\dd}{\dd r}\left(h^{-1}f_{a}\right)\frac{\dd}{\dd r}\left(h^{-1}\psi\right)+\Vmod_a\frac{r^2+a^2}{\Delta_-}f_{a}\psi\bigg)\,\dd r=-\nu_a\int_{r_+}^{\infty}\frac{f_{a}}{r^2}\psi\,\dd r
	\end{aligned}
	\end{align}
	for all $\psi\in\uH_{\kappa}^1(r_+,\infty)$.
\end{lemma}

The proof of Lemma~\ref{lemma:ELDN} can be found in Appendix~\ref{sec:Twisted_Euler_Lagrange}.

\begin{propn}
	There is an $\hat a$ and a corresponding non-zero function $f_{\hat a}\in C^{\infty}(r_+,\infty)$ such that
	\begin{align*}
	\frac{\Delta_-}{r^2+\hat a^2}\frac{\dd}{\dd r}\left(\frac{\Delta_-}{r^2+\hat a^2}\frac{\dd f_{\hat a}}{\dd r}\right)-\tilde V_{\hat a}f_{\hat a}=0
	\end{align*}
	and $f_a$ satisfies the horizon regularity condition and the Neumann boundary condition at infinity.
\end{propn}

\begin{proof}
	As in Proposition~\ref{propn:a_hat}, we find an $f_{\hat a}\in\uH_{\kappa}^1$ such that
	\begin{align*}
	\int_{r_+}^{\infty}\bigg(\frac{\Delta_-}{r^2+\hat a^2}r^{-1+2\kappa}\frac{\dd}{\dd r}\left(r^{\frac{1}{2}-\kappa}f_{\hat a}\right)\frac{\dd}{\dd r}\left(r^{\frac{1}{2}-\kappa}\psi\right)+\Vmod_{\hat a}\frac{r^2+\hat a^2}{\Delta_-}f_{\hat a}\psi\bigg)\,\dd r=0.
	\end{align*}
	Choosing $\psi(r)=r^{-\frac{1}{2}+\kappa}g(r)$ with $g$ having compact support around infinity and integrating by parts, we obtain
	\begin{align*}
	\frac{\Delta_-}{r^2+\hat a^2}r^{-1+2\kappa}\frac{\dd}{\dd r}\left(r^{\frac{1}{2}-\kappa}f_{\hat a}\right)g\rightarrow 0
	\end{align*}
	as $r=\infty$ for all $g$ as in (\ref{eqn:basis_twisting}). This yields the asymptotics.
	
	Moreover, as in the proof of Proposition~\ref{propn:a_hat}, we retrieve the ODE.
	The boundary condition at the horizon follows analogously to Section~\ref{sec:Real}.
\end{proof}

\subsection{Perturbing the Neumann modes into the complex plane}
\label{subsec:PerturbingN}

In Section~\ref{sec:RealDN}, we constructed real mode solutions for $5/4<\alpha<9/4$ satisfying Neumann conditions.
For the growing radial parts, we proceed as in Section~\ref{sec:Real} with the difference that here, finding a mode solution is equivalent to finding a zero of $B$. The present case is considerably more difficult than the Dirichlet case. A first manifest difference is the asymmetry in the definitions of Dirichlet and Neumann boundary conditions since a Dirichlet mode has more decay than required by Definition~\ref{defn:BdyConds}. This means that if a function satisfies the Dirichlet boundary condition for a mass $\alpha_1$, it also does so for every $\alpha_2$ sufficiently close to $\alpha_1$. As Definition~\ref{defn:BdyCondsNeumann} is tighter, this is not true in the Neumann case.  Another difficulty stems from twisting as the dependence of the equations on $\alpha$ becomes more complicated.

We have already chosen $B(\alpha(0),\omega_R(0))=0$. Recall from Section~\ref{sec:RestrReal} that
\begin{align*}
Q_T=\Im\left(r^{-\frac{1}{2}+\kappa}\frac{\dd}{\dd\sta r}\left(r^{\frac{1}{2}-\kappa}u\right)\overline u\right),~~~~~ Q_T(r_+)=\Xi am-\omega_R(r_+^2+a^2).
\end{align*}
Hence, analogously to Section~\ref{subsec:PerturbingD}, the problem reduces to showing that
\begin{align*}
\frac{\del B}{\del\alpha}(\alpha(0),\omega_R(0))\neq 0.
\end{align*}
Again, for the sake of contradiction, suppose that this is not the case. 
Then, near infinity, we have
\begin{align*}
u_{\alpha}(r,\alpha(0),\omega_R(0))&=\frac{\del A}{\del\alpha}(\alpha(0),\omega_R(0)) h_1(r,\alpha(0),\omega_R(0))\\
&~~~~~~~+A(\alpha(0),\omega_R(0))\frac{\del h_1}{\del\alpha}(r,\alpha(0),\omega_R(0)).
\end{align*}
By the horizon regularity condition, $u\sim(r-r_+)^{\xi}$ near the horizon, $u_{\alpha}$ is smooth at $r=r_+$. However, $u_{\alpha}$ does not satisfy the Neumann condition at infinity as the second term behaves as $r^{-1/2+\kappa}u\log r$.

Let $f:\,(r_+,\infty)\rightarrow\CC$ be $C^1$ and piecewise $C^2$. Then the function
\begin{align*}
v(r):=u_{\alpha}(r)-\frac{\del\kappa}{\del\alpha}f(r)u(r)=u_{\alpha}(r)+\frac{1}{2\kappa}f(r)u(r),
\end{align*}
does satisfy the Neumann boundary condition if, for large $r$, $f(r)=\log r+\OO(r^{-\gamma})$, where $\gamma>0$.

From the radial ODE, we obtain
\begin{align*}
&\frac{\dd}{\dd r}\left(\frac{\Delta_-}{r^2+a^2}\frac{\dd u_{\alpha}}{\dd r}\right)-\frac{r^2+a^2}{\Delta_-}\tilde V_au_{\alpha}=\frac{1}{r^2+a^2}\left[\frac{\del\lambda}{\del\alpha}-\frac{1}{\ell^2}(r^2+a^2)\right]u.
\end{align*}	
Lemma~\ref{lemma:EL_twisted_g} yields a twisted version
\begin{align}
\begin{split}
\label{eqn:twisted_del_alpha_perturbation}
\frac{1}{h}\frac{\dd}{\dd r}\left(\frac{\Delta_-}{r^2+a^2}h^2\frac{\dd}{\dd r}\left(\frac{u_{\alpha}}{h}\right)\right)+\tilde V^h_a\frac{r^2+a^2}{\Delta_-}u_{\alpha}=\frac{1}{r^2+a^2}\left[\frac{\del\lambda}{\del\alpha}-\frac{1}{\ell^2}(r^2+a^2)\right]u.
\end{split}
\end{align}
We will use the previous twisting, i.\,e. $h=r^{-1/2+\kappa}$.

For the second term of $v$, we compute:
\begin{align*}
\frac{1}{h}\frac{\dd}{\dd r}\left(\frac{\Delta_-}{r^2+a^2}h^2\frac{\dd}{\dd r}\left(f\frac{u}{h}\right)\right) &=h^{-1}\frac{\dd}{\dd r}\left(\frac{\Delta_-}{r^2+a^2}h^2\frac{\dd f}{\dd r}\right)(h^{-1}u)\\&~~~~~+2h^{-1}\frac{\Delta_-}{r^2+a^2}h^2\frac{\dd f}{\dd r}\frac{\dd}{\dd r}(h^{-1}u)\\&~~~~~+h^{-1}f\frac{\dd}{\dd r}\left(\frac{\Delta_-}{r^2+a^2}h^2\frac{\dd}{\dd r}(h^{-1}u)\right)
\end{align*}
We add this to the equation to (\ref{eqn:twisted_del_alpha_perturbation}) multiplied by $2\kappa$. Then we multiply the resulting equation by $\overline u$ and integrate by parts, noting that $v$ satisfies the Neumann boundary condition. Hence we obtain
\begin{align}
\begin{split}
\label{eqn:Bzero_contradiction}
0&=\int_{r_+}^{\infty}\frac{2\kappa}{r^2+a^2}\left(\frac{\del\lambda}{\del\alpha}-\frac{a^2}{\ell^2}\right)|u|^2\,\dd r\\
&~~~~~-\int_{r_+}^{\infty}\left(\frac{2\kappa}{\ell^2}\frac{r^2}{r^2+a^2}|u|^2-\frac{\dd}{\dd r}\left(\frac{\Delta_-}{r^2+a^2}h^2\frac{\dd f}{\dd r}\right)\left\lvert h^{-1}u\right\lvert^2\right)\,\dd r\\
&~~~~~-2\int_{r_+}^{\infty}f\left(\frac{\Delta_-}{r^2+a^2}h^2\left\lvert\frac{\dd}{\dd r}\left(h^{-1}u\right)\right\lvert^2+\tilde V^h_a\frac{r^2+a^2}{\Delta_-}|u|^2\right)\,\dd r.
\end{split}
\end{align}
Our aim is to show that the right hand side of (\ref{eqn:Bzero_contradiction}) is negative, which yields the desired contradiction.

From Section~\ref{subsec:PerturbingD}, we already know that
\begin{align*}
\int_{r_+}^{\infty}\frac{2\kappa}{r^2+a^2}\left(\frac{\del\lambda}{\del\alpha}-\frac{a^2}{\ell^2}\right)|u|^2\,\dd r=-\int_{r_+}^{\infty}\frac{2\kappa}{r^2+a^2}\int_0^{\pi}\frac{a^2}{\ell^2}\cos^2\theta|S|^2\sin\theta|u|^2\,\dd\theta\,\dd r
\end{align*}
has the right sign. We set
\begin{align}
\label{f_perturbation}
f(r):=\begin{cases}
\log r+\frac{1}{2\kappa}\frac{R^{2\kappa}}{r^{2\kappa}},	& r\geq R\\
\log R+\frac{1}{2\kappa},	& r<R
\end{cases}
\end{align}
for an $R>r_++1$ to be determined. Note that $f$ is continuously differentiable. For $r> R$,
\begin{align*}
\frac{\dd f}{\dd r}(r)=\frac{1}{r}\left(1-\frac{R^{2\kappa}}{r^{2\kappa}}\right)>0,
\end{align*}
whence $f$ is monotonic.

First, we choose $R$ sufficiently large such that $\tilde V^h_a>0$ for $r>R$ according to Lemma~\ref{lemma:twisting_positive_g}, whence
\begin{align*}
&\int_{r_+}^{\infty}f\left(\frac{\Delta_-}{r^2+a^2}h^2\left\lvert\frac{\dd}{\dd r}\left(h^{-1}u\right)\right\lvert^2+\tilde V^h_a\frac{r^2+a^2}{\Delta_-}|u|^2\right)\,\dd r\\=&\left(\log R+\frac{1}{2\kappa}\right)\int_{r_+}^{\infty}\left(\frac{\Delta_-}{r^2+a^2}h^2\left\lvert\frac{\dd}{\dd r}\left(h^{-1}u\right)\right\lvert^2+\tilde V^h_a\frac{r^2+a^2}{\Delta_-}|u|^2\right)\,\dd r\\
&~~~~~~~+\int_R^{\infty}\left(\log\frac{r}{R}+\frac{1}{2\kappa}\left(\frac{R^{2\kappa}}{r^{2\kappa}}-1\right)\right)\left(\frac{\Delta_-}{r^2+a^2}h^2\left\lvert\frac{\dd}{\dd r}\left(h^{-1}u\right)\right\lvert^2+\tilde V^h_a\frac{r^2+a^2}{\Delta_-}|u|^2\right)\,\dd r,
\end{align*}
which is non-negative since the first integral with the constant coefficient is zero and the second integral is positive.
For $r>R$, one easily computes
\begin{align*}
\frac{2\kappa}{\ell^2}\frac{r^2}{r^2+a^2}-h^{-2}\frac{\dd}{\dd r}\left(\frac{\Delta_-}{r^2+a^2}h^2\frac{\dd f}{\dd r}\right)
=-\frac{2\kappa a^2/\ell^2+(2\kappa-2)}{r^2}-2\frac{R^{2\kappa}}{r^{2\kappa}}\frac{1}{r^2}+\OO(r^{-3})
\end{align*}
Therefore, there is a $C_1>0$ such that
\begin{align*}
\int_R^{\infty}\left(\frac{2\kappa}{\ell^2}\frac{r^2}{r^2+a^2}|u|^2-h^{-2}\frac{\dd}{\dd r}\left(\frac{\Delta_-}{r^2+a^2}h^2\frac{\dd f}{\dd r}\right)\right)|u|^2\,\dd r>-C_1\int_R^{\infty}\frac{|u|^2}{r^2}\,\dd r.
\end{align*}
We can prove the following Hardy inequality:
\begin{lemma}
	Let $u$ satisfy the Neumann boundary condition at infinity and let $\beta>0$. Then 
	\begin{align*}
	\int_R^{\infty}\frac{1}{r^{1+\beta}}\left\lvert r^{\frac{1}{2}-\kappa}u\right\lvert^2\,\dd r&\leq \lim_{r\rightarrow\infty}\frac{2\beta}{R^{\beta}}\left\lvert  r^{\frac{1}{2}-\kappa}u(r)\right\lvert^2\\&~~~~~~+4\beta^2 \int_R^{\infty}r^{2-2\kappa-\beta}\left(1-\left(\frac{r}{R}\right)^{\beta}\right)^2h^2\left\lvert\frac{\dd}{\dd r}\left(\frac{u}{h}\right)\right\lvert^2\,\dd r.
	\end{align*}
\end{lemma}

\begin{proof} We compute
	\begin{align*}
	\int_R^{\infty}\frac{1}{r^{1+\beta}}\left\lvert r^{\frac{1}{2}-\kappa}u\right\lvert^2\,\dd r&=\int_R^{\infty}\del_r\left(-\frac{\beta}{r^{\beta}}+\frac{\beta}{R^{\beta}}\right)\left\lvert r^{\frac{1}{2}-\kappa}u\right\lvert^2\,\dd r\\ 
	&\leq\frac{\beta}{R^{\beta}}\left\lvert r^{\frac{1}{2}-\kappa}u\right\lvert^2(\infty)+\int_R^{\infty}\frac{1}{2}\frac{1}{r^{1+\beta}}\left\lvert r^{\frac{1}{2}-\kappa}u\right\lvert^2\,\dd r\\&~~~~~~~~~~+2\beta^2 \int_R^{\infty}r^{1+\beta}\left(r^{-\beta}-R^{-\beta}\right)^2\left\lvert\frac{\dd}{\dd r}\left(r^{\frac{1}{2}-\kappa}u\right)\right\lvert^2\,\dd r,
	\end{align*}
		yielding the result.
\end{proof}

Since 
\begin{align*}
C_1\lim_{r\rightarrow\infty}\frac{2\beta}{R^{\beta}}\left\lvert  r^{\frac{1}{2}-\kappa}u(r)\right\lvert^2\rightarrow 0
\end{align*}
as $R\rightarrow\infty$, 
by choosing $R$ possibly larger, we obtain
\begin{align*}
C_1\frac{2\beta}{R^{\beta}}\lim_{r\rightarrow\infty}\left\lvert  r^{\frac{1}{2}-\kappa}u(r)\right\lvert^2<\int_{r_+}^{\infty}\frac{2\kappa}{r^2+a^2}\left(\frac{\del\lambda}{\del\alpha}-\frac{a^2}{\ell^2}\right)|u|^2\,\dd r.
\end{align*}
For convenience, set $C_2:=4\beta^2 C_1$. Since $\tilde V_a^h\sim r^{-2}$, we need to show that
\begin{align}
\begin{split}
\label{eqn:goal_perturbation}
C_2\int_R^{\infty}\frac{|u|^2}{r^2}\,\dd r<&\int_R^{\infty}\left(\log\frac{r}{R}+\frac{1}{2\kappa}\left(\frac{R^{2\kappa}}{r^{2\kappa}}-1\right)\right)\frac{\Delta_-}{r^2+a^2}h^2\left\lvert\frac{\dd}{\dd r}\left(h^{-1}u\right)\right\lvert^2\,\dd r.
\end{split}
\end{align}

We will deal with the two cases $0<\kappa\leq 1/2$ and $1/2<\kappa<1$ separately. Let us first consider $0<\kappa\leq 1/2$. We choose $\beta=2\kappa$. Note that in this case
\begin{align*}
\int_R^{\infty}\frac{|u|^2}{r^2}\,\dd r\leq \int_R^{\infty}\frac{1}{r^{1+2\kappa}}\left\lvert r^{1/2-\kappa}u\right\lvert^2\,\dd r.
\end{align*}

\begin{lemma}
	Let $C>0$. There is an $R$ such that, for all $r>R$,
	\begin{align*}
	Cr^{2-4\kappa}\left(1-\left(\frac{r}{R}\right)^{2\kappa}\right)^2\leq r^2\left(\log\frac{r}{R}+\frac{1}{2\kappa}\left(\left(\frac{R}{r}\right)^{2\kappa}-1\right)\right).
	\end{align*}
\end{lemma}

\begin{proof}
	It suffices to show that
	\begin{align*}
	Cr^{-4\kappa}\left(1-\left(\frac{r}{R}\right)^{2\kappa}\right)^2\leq \log\frac{r}{R}+\frac{1}{2\kappa}\left(\left(\frac{R}{r}\right)^{2\kappa}-1\right).
	\end{align*}
	As this holds at $r=R$, it suffices to show the statement for the derivatives. Substituting $x:=r^{2\kappa}$, we need to show
	\begin{align*}
	0\leq x^2-(R^{2\kappa}+4\kappa C R^{-2\kappa})x+4\kappa C=(x-R^{2\kappa})(x-4\kappa C R^{-2\kappa}).
	\end{align*}
	Therefore, the result holds if $R^{4\kappa}>4\kappa C$.
\end{proof}

This lemma immediately yields (\ref{eqn:goal_perturbation}) for $0<\kappa\leq 1/2$. Let us now turn to $1/2<\kappa<1$. Here we choose $\beta=2-2\kappa$.

\begin{lemma}
	Let $C>0$ and $1/2<\kappa<1$. There is an $R$ such that, for all $r>R$,
	\begin{align*}
	C\left(1-\left(\frac{r}{R}\right)^{2-2\kappa}\right)^2\leq r^2\left(\log\frac{r}{R}+\frac{1}{2\kappa}\left(\left(\frac{R}{r}\right)^{2\kappa}-1\right)\right).
	\end{align*}
\end{lemma}

\begin{proof}
	As equality holds for $r=R$, it suffices to consider the derivatives, i.\,e. we would like to establish
	\begin{align*}
	&2r\left(\log\frac{r}{R}+\frac{1}{2\kappa}\left(\left(\frac{R}{r}\right)^{2\kappa}\right)\right)+r^2\left(\frac{1}{r}-\left(\frac{R}{r}\right)^{2\kappa}\frac{1}{r}\right)\\
	&~~~~~~~~~~~~~~~~~~~-2C\left(\left(\frac{r}{R}\right)^{2-2\kappa}-1\right)(2-2\kappa)R^{-2+2\kappa}r^{1-2\kappa}\geq 0.
	\end{align*}
	This again holds for $r=R$, so, after dividing the inequality by $r$, it suffices to prove the corresponding inequality for the derivatives, i.\,e.
	\begin{align*}
	&\frac{1}{r}\left(2-2\left(1-\kappa\right)\left(\frac{R}{r}\right)^{2\kappa}\right)-4C\kappa(2-2\kappa)\frac{1}{R^2}\left(\frac{r}{R}\right)^{-1-2\kappa}\\&~~~~~~~~~~~+2C(2-2\kappa)(4\kappa-2)R^{-4-4\kappa}r^{1-4\kappa}\geq 0.
	\end{align*}
	The last term on the left hand side is always positive. Thus the left hand side is greater than
	\begin{align*}
	\frac{\kappa}{r}-4C\kappa(2-2\kappa)\frac{1}{R}\left(\frac{R}{r}\right)^{2\kappa}\geq \left(\frac{1}{2}-4C\frac{1}{R}\right)\frac{1}{r},
	\end{align*}
	which is positive for sufficienlty large $R$.
\end{proof}

Therefore, for both ranges of $\kappa$, the right hand side of (\ref{eqn:Bzero_contradiction}) is bounded below by
\begin{align*}
-\int_{r_+}^R\frac{1}{\ell^2}\frac{r^2}{r^2+a^2}|u|^2\,\dd r-2\int_R^{\infty}\left(\log\frac{r}{R}+\frac{1}{2\kappa}\left(\frac{R^{2\kappa}}{r^{2\kappa}}-1\right)\right)\tilde V_a^h\frac{r^2+a^2}{\Delta_-}|u|^2\,\dd r<0
\end{align*}
for non-trivial $u$,
a contradiction. Thus we have shown the following

\begin{lemma}
	\begin{align*}
	\frac{\del B}{\del\alpha}(\alpha(0),\omega_R(0))\neq 0.
	\end{align*}
\end{lemma}

\subsection{Behaviour for small $\epsilon>0$ for Neumann boundary conditions}
\label{subsec:small_behaviour_Neumann}

The main new idea of this section can be found in the proof of Proposition~\ref{propn:alpha_decreasing_Neumann}, where the insights of Section~\ref{subsec:PerturbingN} are essential to overcome the difficulties outlined at the beginning of the previous section.

\begin{propn}
	\label{propn:del_omega_Neumann}
	For sufficiently large $|m|$,
	\begin{align*}
	\omega_R(0)\frac{\del\omega_R}{\del\epsilon}(0)< 0.
	\end{align*}
\end{propn}

\begin{proof}
	We define an appropriate modified microlocal energy current
	\begin{align*}
	\tilde Q_T:=\Im\left(r^{-\frac{1}{2}+\kappa}\left(r^{\frac{1}{2}-\kappa}\right)'\overline{\omega u}\right).
	\end{align*}
	Let $\epsilon>0$, then $\tilde Q_T(r_+)=\tilde Q_T(\infty)=0$. This yields
	\begin{align}
	\label{eqn:int_whole}
	\int_{r_+}^{\infty}\left(\epsilon\frac{\Delta_-}{r^2+a^2}h^2\left\lvert\frac{\dd }{\dd r}\left(h^{-1}u\right)\right\lvert^2-\frac{r^2+a^2}{\Delta_-}\Im(\Vmod_a\overline\omega)|u|^2\right)\,\dd r=0.
	\end{align}
	Similarly to Section~\ref{subsec:Crossing}, we obtain
	\begin{align}
	\begin{aligned}
	\label{eqn:Im_ineq}
	-\Im\left(\Vmod_a\overline\omega\right)&>\frac{\epsilon}{(r^2+a^2)^2}\left(K(r)+V_+(r^2+a^2)^2-\frac{\alpha}{\ell^2}\Delta_-(r^2+a^2)\right)\\&~~~~~~~~+\epsilon \left(\frac{1}{2}-\kappa\right)\frac{\Delta_-}{r^2+a^2}r^{1/2-\kappa}\frac{\dd}{\dd r}\left(\frac{\Delta_-}{r^2+a^2}r^{-3/2+\kappa}\right)
	\end{aligned}
	\end{align}
	with the additional term due to the twisting. Again
	\begin{align*}
	K(r)=|\omega|^2(r^2+a^2)^2-\Xi^2 a^2m^2-\Delta_-a^2|\omega|^2.
	\end{align*}
	Recall from Section~\ref{subsec:Crossing} that
	\begin{align}
	\label{eqn:K_growth}
	\frac{\dd K}{\dd r}(r)=|\omega|^2\left(4\left(1-\frac{a^2}{\ell^2}\right)r^3+2a^2M+2a^2\left(1-\frac{a^2}{\ell^2}\right)\right)>0.
	\end{align}
	By Lemma~\ref{lemma:twisting_positive_g}, there is an $R>r_+$ such that
	\begin{align*}
	\frac{r^2+a^2}{\Delta_-}\left\lvert V_++V_{\alpha}+\left(\frac{1}{2}-\kappa\right)\frac{\Delta_-}{r^2+a^2}r^{1/2-\kappa}\frac{\dd}{\dd r}\left(\frac{\Delta_-}{r^2+a^2}r^{-3/2+\kappa}\right)\right\lvert<\frac{C}{2r^2}
	\end{align*}
	for any $C>0$. Thus, by an application of Lemma~\ref{lemma:Hardy_twisted} as in the proof of Lemma~\ref{lemma:VariationInequalityND}, we have
	\begin{align*}
	\begin{split}
	&\int_R^{\infty}\left(\epsilon\frac{\Delta_-}{r^2+a^2}h^2\left\lvert\frac{\dd }{\dd r}\left(h^{-1}u\right)\right\lvert^2-\frac{r^2+a^2}{\Delta_-}\Im(\Vmod_a\overline\omega)|u|^2\right)\,\dd r\\&~~~~~~~~~~~~~~~~~~~~~~>\int_R^{\infty}\frac{\epsilon}{(r^2+a^2)^2}K(r)|u|^2-\int_{R-1}^R\epsilon C'|u|^2\,\dd r
	\end{split}
	\end{align*}
	for sufficiently large $R$ and a large constant $C'>0$.
	
	For the sake of contradiction, suppose that $K(r_+)\geq 0$. Then, of course,
	\begin{align}
	\begin{split}
	\label{eqn:int_R}
	&\int_R^{\infty}\left(\epsilon\frac{\Delta_-}{r^2+a^2}h^2\left\lvert\frac{\dd }{\dd r}\left(h^{-1}u\right)\right\lvert^2-\frac{r^2+a^2}{\Delta_-}\Im(\Vmod_a\overline\omega)|u|^2\right)\,\dd r\\&~~~~~~~~>\int_R^{\infty}\frac{\epsilon}{(r^2+a^2)^2}K(r)|u|^2+\int_{R-1}^R\epsilon\left(\frac{K(r)}{(r^2+a^2)^2}-C'\right) |u|^2\,\dd r
	\end{split}
	\end{align}
		By (\ref{eqn:K_growth}), this means that $K>0$ on $(r_+,\infty)$. Since $\epsilon\mapsto\omega(\epsilon)$ is continuous and
	\begin{align*}
	|\omega(0)|^2\geq Cm^2,
	\end{align*}
	$|\omega|^2$ scales as $m^2$, so $\dd K/\dd r$ can be chosen to be as large as possible by increasing $m^2$, in particular, it can be used to overcome the potentially negative derivative of the remaining terms of the right hand side of (\ref{eqn:Im_ineq}) on $(r_+,R)$ and in (\ref{eqn:int_R}) on $(R-1,R)$. Using (\ref{eqn:int_whole}) and (\ref{eqn:int_R}), we conclude $u=0$, a contradiction.
\end{proof}

From now on we fix $m$ -- see Remark~\ref{rk:large_m}.

\begin{propn}
	\label{propn:alpha_decreasing_Neumann}
	\begin{align*}
	\frac{\del\alpha}{\del\epsilon}(0)>0.
	\end{align*}
\end{propn}

\begin{proof}
	The proof proceeds as in Section~\ref{subsec:Crossing}, adapting the idea used already in Section~\ref{subsec:PerturbingN}. 
		Set $u_{\epsilon}:=\del u/\del\epsilon$. For an $f$ as in (\ref{f_perturbation}) with an $R$ to be determined,
	\begin{align*}
	v(r):=u_{\epsilon}(r)-\frac{\del\kappa}{\del\epsilon}f(r)u(r)=u_{\epsilon}(r)+\frac{1}{2\kappa}\frac{\del\alpha}{\del\epsilon}f(r)u(r)
	\end{align*}
	satisfies the Neumann boundary condition at infinity. As $fu$ extends smoothly to the horizon, the behaviour of $v$ at $r=r_+$ is dominated by $u_{\epsilon}$.
	Using the $h$ of Lemma~\ref{lemma:twisting_positive_g} yields the ODE
	\begin{align*}
	& h^{-1}\frac{\dd}{\dd r}\left(\frac{\Delta_-}{r^2+a^2}h^2\frac{\dd }{\dd r}\left(\frac{v}{h}\right)\right)-\frac{r^2+a^2}{\Delta_-}\tilde V^h_a v\\&~~~~~~~=\frac{r^2+a^2}{\Delta_-}\frac{\del\tilde V_a}{\del\epsilon}u+\frac{1}{2\kappa}\frac{\del\alpha}{\del\epsilon}h^{-1}\frac{\dd}{\dd r}\left(\frac{\Delta_-}{r^2+a^2}h^2\frac{\dd f}{\dd r}\right)\left(h^{-1}u\right)\\
	&~~~~~~~~~~~+2\frac{1}{2\kappa}\frac{\del\alpha}{\del\epsilon}h^{-1}\frac{\Delta_-}{r^2+a^2}h^2\frac{\dd f}{\dd r}\frac{\dd}{\dd r}\left(h^{-1}u\right)+\frac{1}{2\kappa}\frac{\del\alpha}{\del\epsilon}fh^{-1}\frac{\dd}{\dd r}\left(\frac{\Delta_-}{r^2+a^2}h^2\frac{\dd}{\dd r}\left(h^{-1}u\right)\right).
	\end{align*}
	Observe that as in the proof of Proposition~\ref{propn:alpha_Dirichlet}
	\begin{align*}
	\frac{\dd}{\dd r}\left(\frac{\Delta_-}{r^2+a^2}h^2\frac{\dd}{\dd r}\left(h^{-1}u_{\epsilon}\right)\right)h^{-1}\overline u&=\frac{\dd}{\dd r}\left(\frac{\Delta_-}{r^2+a^2}h\left[\frac{\dd}{\dd r}\left(h^{-1}u_{\epsilon}\right)\overline u-u_{\epsilon}\frac{\dd}{\dd r}\left(h^{-1}\overline u\right)\right]\right)\\
	&~~~~~~~~+u_{\epsilon}h^{-1}\frac{\dd}{\dd r}\left(\frac{\Delta_-}{r^2+a^2}h^2\frac{\dd}{\dd r}\left(h^{-1}\overline u\right)\right).
	\end{align*}
	This yields
	\begin{align}
	\begin{split}
	\label{eqn:integral_var_epsilon}
	-2\kappa|u(r_+)|^2&=\int_{r_+}^{\infty}\left(\frac{r^2+a^2}{\Delta_-}2\kappa\Re\left(\frac{\del\tilde V}{\del\epsilon}\right)\Bigg\lvert_{\epsilon=0}+\frac{\del\alpha}{\del\epsilon}\frac{\dd}{\dd r}\left(\frac{\Delta_-}{r^2+a^2}h^2\frac{\dd f}{\dd r}\right)\left\lvert h^{-1}u\right\lvert^2\right)\,\dd r\\
	&~~~~~-\frac{\del\alpha}{\del\epsilon}\int_{r_+}^{\infty}f\left(\frac{\Delta_-}{r^2+a^2}h^2\left\lvert\frac{\dd}{\dd r}\left(h^{-1}u\right)\right\lvert^2+\tilde V^h_a\frac{r^2+a^2}{\Delta_-}|u|^2\right)\,\dd r\\
	&~~~~~-\frac{\del\alpha}{\del\epsilon}\int_{r_+}^{\infty}f\frac{\Delta_-}{r^2+a^2}h^2\left\lvert\frac{\dd}{\dd r}\left(h^{-1}u\right)\right\lvert^2\,\dd r.
	\end{split}
	\end{align}
	The expression for $\Re\left(\frac{\del\tilde V}{\del\epsilon}\right)\Big\lvert_{\epsilon=0}$ can be taken from (\ref{eqn:real_part_V_epsilon}). First one can eliminate the explicit $\lambda$ dependence via Lemma~\ref{propn:IfThen} and one obtains a lower bound on the right hand side using Proposition~\ref{propn:del_omega_Neumann}.
	Then suppose for the sake of contradiction that $\del\alpha/\del\epsilon\leq 0$. It follows immediately from Section~\ref{subsec:PerturbingN} that the right hand side of (\ref{eqn:integral_var_epsilon}) is positive, a contradiction.
\end{proof}

\subsection{The continuity argument for Neumann boundary conditions}
\label{subsec:Continuity_Neumann}

To apply the continuity argument to the Neumann case, we need to take the two steps outlined in the introduction to Section~\ref{subsec:Continuity}. The second step merely relied on continuity properties of $A$ and the monotonicity properties of $\omega(\epsilon)$ and $\alpha(\epsilon)$ established in Sections~\ref{subsec:Crossing} and \ref{subsec:small_behaviour_Neumann}, respectively; in particular, it did not rely directly on properties of the functional. Hence this part of the argument can be carried out almost verbatim. Therefore, we only need to deal with the first step here.

We make the analogous definitions for $\Aa_{\alpha,r_+}$ and $\Phi$ as in Section~\ref{subsec:Continuity}. For $f=r^{-1/2+\kappa}g$, $x\mapsto g(1/x)\in C_0^{\infty}[0,1/r_+)$, we define
\begin{align*}
\Aa_{\alpha,r_+}(f):=\{a>0\,:\,\tL_{\alpha,r_+,a}(f)<0\}
\end{align*}
and
\begin{align*}
\Aa_{\alpha,r_+}:=&\bigcup_{g\in C_0^{\infty}[0,1/r_+)}\Aa_{\alpha,r_+}(r\mapsto r^{-1/2+\kappa}g(1/r)))\\
=&\{a>0\,:\,\exists (x\mapsto g(1/x))\in C_0^{\infty}[0,1/r_+):\,\tL_{\alpha,r_+,a}(r^{-1/2+\kappa}g)<0\}.
\end{align*}
Moreover, we define
\begin{align*}
\Phi:\,(5/4,9/4)\times(0,\infty)\rightarrow(0,\infty),~\Phi(\alpha,r_+):=\inf\Aa_{\alpha,r_+}
\end{align*}
if $\Aa_{\alpha,r_+}\neq\emptyset$.
Instead of showing monotonicity for $\Phi$, we will define a left-continuous function $\Psi$ that can play the r\^ole of $\Phi$ in the continuity argument. For each $\alpha\in (5/4,9/4)$, there will be a value $\Phi(\alpha,r_+)$ for $a$ such that there is a real mode solution satisfying the Neumann boundary condition for $\alpha$ and this $a$.  The function $\Psi(\cdot,r_+)$ will essentially look like $\Phi(\cdot,r_+)$, but will be modified on potential jump points to achieve left-continuity. The arguments of Sections~\ref{subsec:PerturbingN} and~\ref{subsec:small_behaviour_Neumann} (which depend only on the existence of a Neumann mode solution) can be repeated for $\Psi(\alpha,r_+)$ instead of $\Phi(\alpha,r_+)$, thus we can substitue $\Phi$ by $\Psi$ in the remainder of the proof of Section~\ref{subsec:Continuity}.

\begin{lemma}
	\label{lemma:Psi_for_Phi}
	There is a left-continuous function $\Psi(\cdot,r_+)$ such that there is a real mode solution satisfying the Neumann boundary condition for each $5/4<\alpha<9/4$ and each $a=\Psi(\alpha,r_+)$. 
\end{lemma}

To prove this lemma, we need a monotonicity result about the twisted functional. Note that for a fixed $g$ (where $(x\mapsto g(1/x))\in C_0^{\infty}[0,1/r_+))$, the function
\begin{align}
\label{eqn:continuity_functional_Neumann}
(\alpha,a)\mapsto\tL_{\alpha,r_+,a}(r^{-1/2+\kappa}g)
\end{align}
is continuous.

\begin{lemma}
	\label{lemma:negative_derivative_twisted_functional}
	Let $5/4<\kappa_0<9/4$. Fix all spacetime parameters. Let $u_0:=r^{-1/2+\kappa_0}g_0$ be a solution to the radial ODE at $\kappa_0$. Define $u(r,\kappa):=r^{-1/2+\kappa}g_0$. Then
	\begin{align*}
	\frac{\del}{\del\kappa}\tL_{\alpha,r_+,a}(u(r,\kappa))\bigg\lvert_{\kappa=\kappa_0}>0.
	\end{align*}
\end{lemma}

\begin{proof}
	We start from the identity
	\begin{align*}
	h^{-1}\frac{\dd}{\dd r}\left(\frac{\Delta_-}{r^2+a^2}h^2\frac{\dd}{\dd r}\left(\frac{u}{h}\right)\right)-\tilde V_a^h\frac{r^2+a^2}{\Delta_-}u=\frac{\dd}{\dd r}\left(\frac{\Delta_-}{r^2+a^2}\frac{\dd u}{\dd r}\right)-\tilde V_a\frac{r^2+a^2}{\Delta_-}u
	\end{align*}
	where we always take $h=r^{-1/2+\kappa}$. Set $u_{\kappa}:=\del u/\del\kappa$. Let $f$ be as in (\ref{f_perturbation}) with an $R>r_++1$ to be determined and set $v:=u_{\kappa}-f u$. Then we have
	\begin{align*}
	&\frac{\del}{\del\kappa}\left(h^{-1}\frac{\dd}{\dd r}\left(\frac{\Delta_-}{r^2+a^2}h^2\frac{\dd}{\dd r}\left(\frac{u}{h}\right)\right)-\tilde V_a^h\frac{r^2+a^2}{\Delta_-}u\right)\\&~~~~~~~~=\frac{\dd}{\dd r}\left(\frac{\Delta_-}{r^2+a^2}\frac{\dd u_{\kappa}}{\dd r}\right)-\tilde V_a\frac{r^2+a^2}{\Delta_-}u_\kappa-\frac{\del\tilde V_a}{\del\kappa}\frac{r^2+a^2}{\Delta_-}u
	\\&~~~~~~~~=h^{-1}\frac{\dd}{\dd r}\left(\frac{\Delta_-}{r^2+a^2}h^2\frac{\dd}{\dd r}\left(\frac{u_{\kappa}}{h}\right)\right)-\tilde V_a^h\frac{r^2+a^2}{\Delta_-}u_{\kappa}-\frac{\del\tilde V_a}{\del\kappa}\frac{r^2+a^2}{\Delta_-}u
	\\&~~~~~~~~=h^{-1}\frac{\dd}{\dd r}\left(\frac{\Delta_-}{r^2+a^2}h^2\frac{\dd}{\dd r}\left(\frac{v}{h}\right)\right)-\tilde V_a^h\frac{r^2+a^2}{\Delta_-}v
	-\frac{\del\tilde V_a}{\del\kappa}\frac{r^2+a^2}{\Delta_-}u\\
	&~~~~~~~~~~~+h^{-1}\frac{\dd}{\dd r}\left(\frac{\Delta_-}{r^2+a^2}h^2\frac{\dd}{\dd r}\left(h^{-1}fu\right)\right)-\tilde V^h_a\frac{r^2+a^2}{\Delta_-}fu.
	\end{align*}
	Multiplying by $\overline u$, integrating over $(r_+,\infty)$, integrating by parts as in Section~\ref{subsec:PerturbingN} and evaluating at $\kappa=\kappa_0$ yields:
	\begin{align*}
	&\int_{r_+}^{\infty}\frac{\del}{\del\kappa}\left(h^{-1}\frac{\dd}{\dd r}\left(\frac{\Delta_-}{r^2+a^2}h^2\frac{\dd}{\dd r}\left(\frac{u}{h}\right)\right)-\tilde V_a^h\frac{r^2+a^2}{\Delta_-}u\right)\overline u\,\dd r\,\bigg\lvert_{\kappa=\kappa_0}\\
	=&2\kappa\int_{r_+}^{\infty}\frac{1}{r^2+a^2}\left(\frac{\del\lambda}{\del\alpha}-\frac{a^2}{\ell^2}\right)|u|^2\,\dd r\,\bigg\lvert_{\kappa=\kappa_0}\\
	&-\int_{r_+}^{\infty}\left(\frac{2\kappa}{\ell^2}\frac{r^2}{r^2+a^2}|u|^2-\frac{\dd}{\dd r}\left(\frac{\Delta_-}{r^2+a^2}h^2\frac{\dd f}{\dd r}\right)\left\lvert h^{-1}u\right\lvert^2\right)\,\dd r\,\bigg\lvert_{\kappa=\kappa_0}\\
	&-2\int_{r_+}^{\infty}f\left(\frac{\Delta_-}{r^2+a^2}h^2\left\lvert\frac{\dd}{\dd r}\left(h^{-1}u\right)\right\lvert^2+\tilde V^h_a\frac{r^2+a^2}{\Delta_-}|u|^2\right)\,\dd r\,\bigg\lvert_{\kappa=\kappa_0}.
	\end{align*}
	By repeating the proof of Section~\ref{subsec:PerturbingN}, one shows that the right hand side is negative. For the left hand side, we compute:
	\begin{align*}
	&\frac{\del}{\del\kappa}\left(\frac{\Delta_-}{r^2+a^2}h^2\left\lvert\frac{\dd}{\dd r}\left(h^{-1}u\right)\right\lvert^2+\tilde V^h_a\frac{r^2+a^2}{\Delta_-}|u|^2\right)\\
	=&\frac{\del^2}{\del r\del\kappa}\left(\frac{\Delta_-}{r^2+a^2}h^2\frac{\dd}{\dd r}\left(h^{-1}u\right)h^{-1}\overline u\right)\\
	&~~~-\frac{\del}{\del\kappa}\left(h^{-1}\frac{\dd}{\dd r}\left(\frac{\Delta_-}{r^2+a^2}h^2\frac{\dd}{\dd r}\left(\frac{u}{h}\right)\right)\overline u-\tilde V_a^h\frac{r^2+a^2}{\Delta_-}u\overline u\right)\\
	=&-\frac{\del}{\del\kappa}\left(h^{-1}\frac{\dd}{\dd r}\left(\frac{\Delta_-}{r^2+a^2}h^2\frac{\dd}{\dd r}\left(\frac{u}{h}\right)\right)-\tilde V_a^h\frac{r^2+a^2}{\Delta_-}u\right)\overline u\\
	&~~~-\left(h^{-1}\frac{\dd}{\dd r}\left(\frac{\Delta_-}{r^2+a^2}h^2\frac{\dd}{\dd r}\left(\frac{u}{h}\right)\right)-\tilde V_a^h\frac{r^2+a^2}{\Delta_-}u\right)\frac{\del\overline u}{\del\alpha}\\
	&~~~+\frac{\del^2}{\del r\del\kappa}\left(\frac{\Delta_-}{r^2+a^2}h^2\frac{\dd}{\dd r}\left(h^{-1}u\right)h^{-1}\overline u\right)
	\end{align*}
	Again we have
	\begin{align*}
	\left(h^{-1}\frac{\dd}{\dd r}\left(\frac{\Delta_-}{r^2+a^2}h^2\frac{\dd}{\dd r}\left(\frac{u}{h}\right)\right)-\tilde V_a^h\frac{r^2+a^2}{\Delta_-}u\right)\,\bigg\lvert_{\kappa=\kappa_0}=0.
	\end{align*}
	Moreover, 
	\begin{align*}
	\frac{\del}{\del\kappa}\left(\frac{\Delta_-}{r^2+a^2}h^2\frac{\dd}{\dd r}\left(h^{-1}u\right)h^{-1}\overline u\right)\,\bigg\lvert_{\kappa=\kappa_0}\sim r^{1+2\kappa_0}\frac{\dd}{\dd r}\left(r^{1/2-\kappa_0}u_0\right)\log r.
	\end{align*}
	Therefore,
	\begin{align*}
	\frac{\del}{\del\kappa}\tL_{\alpha,r_+,a}(u(r,\kappa))\bigg\lvert_{\kappa=\kappa_0}=-\int_{r_+}^{\infty}\frac{\del}{\del\kappa}\left(h^{-1}\frac{\dd}{\dd r}\left(\frac{\Delta_-}{r^2+a^2}h^2\frac{\dd}{\dd r}\left(\frac{u}{h}\right)\right)-\tilde V_a^h\frac{r^2+a^2}{\Delta_-}u\right)\overline u\,\dd r,
	\end{align*}
	whence positivity.
\end{proof}

\begin{cor}
	\label{cor:local_monotonicity_right}
	For all $\alpha\in (5/4,9/4)$, there is a $\delta>0$ such that $\Phi(\cdot,r_+)$ is decreasing in $[\alpha,\alpha+\delta)$
\end{cor}

\begin{proof}
	Lemma~\ref{lemma:negative_derivative_twisted_functional} shows that monotonicity is an open property.
\end{proof}

Now we can prove the main lemma.

\begin{proof}[Proof of Lemma~\ref{lemma:Psi_for_Phi}]
	Let $\epsilon>0$. By Corollary~\ref{cor:local_monotonicity_right}, there is an indexing set $X$ and disjoint half-open intervals $I_{\beta}$, $\beta\in X$, containing their left endpoints, such that one has $\cup_{\beta\in X}I_{\beta}=[5/4+\epsilon,9/4)$ and $\Phi(\cdot,r_+)\big\lvert_{ I_\beta}$ is decreasing for all $\beta\in X$.
	For $\alpha\in\cup_{\beta\in X}\mathring I_{\beta}$, set
	\begin{align*}
	\Psi(\alpha,r_+):=\Phi(\alpha,r_+).
	\end{align*}
	Let $\alpha_0\in\del I_{\beta_1}\cap I_{\beta_2}$. Choose a sequence $(\alpha_k)\subseteq I_{\beta_1}$ such that $\alpha_k\rightarrow\alpha_0$. As the sequence $(\Phi(\alpha_k,r_+))$ is monotonically decreasing and bounded below, it is convergent. We set $a_0:=\lim_k\Phi(\alpha_k,r_+)$. Let $f_{\alpha_k}$ be the unique solution to the radial ODE with parameters $\alpha_k$, $a_k=\Phi(\alpha_k,r_+)$ and $\omega_k=\Omega_R(\alpha_k)$ -- see definition (\ref{eqn:defn_omega_+}). Let $f_{\alpha_0}$ be the unique solution corresponding to the parameters $\alpha_0$ and $a_0$. Since all $f_{\alpha_k}$ satisfy the Neumann boundary condition, continuity of the reflection and transmission coefficients yields that $f_{\alpha_0}$ satisfies the Neumann boundary condition as well. We set
	\begin{align*}
	\Psi(\alpha_0,r_+):=a_0.
	\end{align*}
	As we can repeat this construction for all $\epsilon>0$ and all jump points $\alpha_0$, we obtain a function $\Psi(\cdot,r_+)$ defined in $(5/4,9/4)$, whose values correspond to parameters $a$ with periodic mode solutions.
		
		Since we have left-continuity at the jump points by construction, it remains to show that $\Psi(\cdot,r_+)$ is left-continuous in $\alpha\in\cup_{\beta\in X}\mathring I_{\beta}$, which can be proved as Lemma~\ref{lemma:CL}: Suppose not.
	Then there is an $\epsilon>0$ such that, for all $\delta>0$, there is an $\alpha'<\alpha$ with
	\begin{align*}
	\alpha-\alpha'<\delta
	\end{align*}
	and
	\begin{align*}
	\Psi(\alpha',r_+)-\Psi(\alpha,r_+)\geq \epsilon.
	\end{align*}
	Then there is an $a$ between $\Psi(\alpha,r_+)$ and $\Psi(\alpha',r_+)$ such that there is an $g$ with $a\in\Aa_{\alpha,r_+}(r^{-1/2+\kappa}g)$,
	but, for each $\delta$, there is an $\alpha'$ with $a\notin\Aa_{\alpha',r_+}(r^{-1/2+\kappa}g)$.
	Therefore, since $\tL_{\alpha,r_+,a}(r^{-1/2+\kappa}g)<0$ and due to the continuity (\ref{eqn:continuity_functional_Neumann}),
	 there is a $\delta>0$ such that for all $\alpha-\alpha'<\delta$, we have $\tL_{\alpha',r_+,a}(r^{-1/2+\kappa}g)<0$, i.\,e. $a\in\Aa_{\alpha',r_+}(r^{-1/2+\kappa}g)$, a contradiction.
\end{proof}

\section{Acknowledgements}

I would like to thank my supervisors Mihalis Dafermos and Gustav Holzegel for proposing this project, helpful discussions, guidance and advice on the exposition. I would also like to thank Yakov Shlapentokh-Rothman for useful conversations and advice. I gratefully acknowledge the financial support of EPSRC, the Cambridge Trust and the Studienstiftung des deutschen Volkes.

\begin{appendices}

\section{The angular ODE}
\label{sec:AngularODE}

Assume throughout the section that $m\neq0$. Recall equations (\ref{eqn:AngularODE1}) and (\ref{eqn:AngularODE2}). 
We will only give details for $\alpha\leq 0$. The other case can be treated analogously.
Define $x:=\cos\theta$. Then the equation becomes
\begin{align*}
 \frac{\dd}{\dd x}\left(\Delta_{\theta}(1-x^2)\frac{\dd S}{\dd x}\right)-&\left(\frac{\Xi^2}{\Delta_{\theta}}\frac{m^2}{1-x^2}-\left(\frac{\Xi}{\Delta_{\theta}}a^2\omega^2-2ma\omega\frac{\Xi}{\Delta_{\theta}}\frac{a^2}{\ell^2}-\frac{\alpha}{\ell^2}a^2\right)x^2\right)S+\lambda S=0.
\end{align*}
Set
\begin{align*}
 K(x):=\frac{\dd}{\dd x}\left(\left(1-\frac{a^2}{\ell^2}x^2\right)\left(1-x^2\right)\right)=4\frac{a^2}{\ell^2}x^3-2x\left(1+\frac{a^2}{\ell^2}\right).
\end{align*}
Using the language of Theorem~\ref{thm:RegularSing}, we see that at $\pm 1$, we have
\begin{align*}
 f_0=1,~~~~~~~~~~g_0=-m^2/4.
\end{align*}
Thus for $m\neq 0,1$, we have two zeros which do not differ by an integer. Then we know that solutions are linear combinations of $(x\mp 1)^{-|m|/2}$ and $(x\mp 1)^{|m|/2}$ near $\pm 1$.

\begin{propn}
Suppose that for some fixed $\omega_0,\alpha_0\in\RR$, we have an eigenvalue $\lambda_0$. Then, for $\kappa$ sufficiently close to $\kappa_0$, we can uniquely find a complex analytic function $\lambda(\omega,\alpha)$ of eigenvalues for the angular ODE with parameter $(\omega,\alpha)\in\CC\times\RR$ such that $\lambda_0=\lambda(\omega,\alpha)$.
\end{propn}

\begin{proof}
We can use the proof in \citep{ShlapentokhGrowing}.
If $S$ is an eigenfunction, we clearly must have
\begin{align*}
S\sim (1\mp x)^{|m|/2}
\end{align*}
as $x\rightarrow \pm 1$. For any $\omega,\alpha$ and $\lambda$, we can uniquely define a solution $S(\theta,\omega,\alpha,\lambda)$ by requiring that
\begin{align*}
S(x,\omega,\alpha,\lambda)(1+x)^{-|m|/2}
\end{align*}
is holomorphic at $x=-1$ and
\begin{align*}
\left(S(\cdot,\omega,\alpha,\lambda)(1+\cdot)^{-|m|/2}\right)(x=-1)=1.
\end{align*}
Then we have holomorphic functions $F(\omega,\alpha,\lambda)$ and $G(\omega,\alpha,\lambda)$ such that
\begin{align*}
S(x,\omega,\alpha,\lambda)\sim F(\omega,\alpha,\lambda)(1-x)^{-|m|/2}+G(\omega,\alpha,\lambda)(1-x)^{|m|/2}
\end{align*}
as $x\rightarrow 1$. Since $\lambda_0$ is an eigenvalue, we have $F(\omega_0,\alpha_0,\lambda_0)=0$. We want to appeal to the implicit function theorem and define our function $\lambda(\omega,\alpha)$ uniquely near $(\omega_0,\alpha_0)$. Suppose (for the sake of contradiction) that
\begin{align*}
\frac{\del F}{\del\lambda}(\omega_0,\alpha_0,\lambda_0)=0.
\end{align*}
Set $S_{\lambda}:=\del S/\del\lambda$. Since $\del F/\del\lambda=0$, $S_{\lambda}$ satisfies the boundary conditions of eigen\-functions. Moreover, we have
\begin{align*}
&\frac{\dd}{\dd x}\left(\Delta_{\theta}(1-x^2)\frac{\dd S_{\lambda}}{\dd x}\right)-\left(\frac{\Xi^2}{\Delta_{\theta}}\frac{m^2}{1-x^2}-\left(\frac{\Xi}{\Delta_{\theta}}a^2\omega_0^2-2ma\omega_0\frac{\Xi}{\Delta_{\theta}}\frac{a^2}{\ell^2}-\frac{\alpha}{\ell^2}a^2\right)x^2\right)S_{\lambda}\\&~~~~~~~~~~~~~~~~~~~+\lambda_0 S=-S.
\end{align*}
Multiplying both sides by $\overline S$, integrating over $(0,\pi)$ with measure $\sin\theta\,\dd\theta$, integrating by parts and using that $\overline S$ satisfies the angular ODE implies
\begin{align*}
\int_0^{\pi}|S|^2\sin\theta\,\dd\theta=0,
\end{align*}
which is a contradiction. The proof for $\alpha<0$ proceeds similarly.
\end{proof}

\begin{propn}
\label{propn:ImOmega}
If $\omega_I>0$, then
\begin{align*}
-\Im(\lambda\overline{\omega})>0.
\end{align*}
\end{propn}

\begin{proof}
Let $\alpha\leq 0$. Multiplying the ODE by $\overline{\omega S}$, integrating by parts and taking imaginary parts gives
\begin{align*}
-\int_0^{\pi}\Im(\lambda\overline{\omega})\sin\theta\,\dd\theta
	&=\int_0^{\pi}\omega_I\left(\Delta_{\theta}\left\lvert\frac{\dd S}{\dd\theta}\right\lvert^2+\left[\frac{\Xi^2}{\Delta_{\theta}}\frac{m^2}{\sin^2\theta}-\frac{\alpha}{\ell^2}a^2\cos^2\theta\right]|S|^2\right)\sin\theta\,\dd\theta\\
	&~~~~~~~~~~+\int_0^{\pi}\frac{\Xi}{\Delta_{\theta}}\cos^2\theta\,\Im(a^2\omega^2\overline{\omega})|S|^2\sin\theta\,\dd\theta,
\end{align*}
which is positive for $\omega_I>0$. For $\alpha>0$, the proof proceeds almost verbatim.
\end{proof}

\begin{propn}
\label{propn:LambdaDeriv}
When $\omega$ is real, we have
\begin{align*}
\frac{\del\lambda}{\del\alpha}=-\frac{a^2}{\ell^2}\int_0^{\pi}\cos^2\theta|S|^2\sin\theta\,\dd\theta
\end{align*}
for $\alpha\leq 0$ and
\begin{align*}
\frac{\del\lambda}{\del\alpha}=\frac{a^2}{\ell^2}\int_0^{\pi}\sin^2\theta|S|^2\sin\theta\,\dd\theta
\end{align*}
for $\alpha>0$.
\end{propn}

\begin{proof}
Let $S_{\alpha}:=\del S/\del\alpha$. First, let $\alpha\leq 0$. 
First one differentiates (\ref{eqn:AngularODE1}) with respect to $\alpha$, then multiplies by $\overline S$ and then integrates by part. Since $\omega\in\RR$, $\overline S$ satisfies the angular ODE, which yields the result. Similarly we obtain the result for $\alpha>0$.
\end{proof}

\section{Twisted derivatives and the modified potential}
\label{sec:twisted_derivative}

To deal with the slow decay or even growth of modes satisfying the Neumann boundary condition, we need to use renormalised derivatives
\begin{align*}
h\frac{\dd}{\dd r}\left(h^{-1}\cdot\right)
\end{align*}
with a sufficiently regular function $h$. 
Defining the modified potential
\begin{align*}
\tilde V^h_a:=\tilde V-\frac{1}{h}\frac{\Delta_-}{r^2+a^2}\frac{\dd}{\dd r}\left(\frac{\Delta_-}{r^2+a^2}\frac{\dd h}{\dd r}\right),
\end{align*}
we obtain a twisted expression for the radial ODE:

\begin{lemma}	
	\label{lemma:EL_twisted_g}
	For all $f\in C^1$ that are piecewise $C^2$,
	\begin{align*}
	h^{-1}\frac{\dd}{\dd r}\left(\frac{\Delta_-}{r^2+a^2}h^2\frac{\dd}{\dd r}\left(h^{-1}f\right)\right)-\tilde V^h_a\frac{r^2+a^2}{\Delta_-}f=\frac{\dd}{\dd r}\left(\frac{\Delta_-}{r^2+a^2}\frac{\dd f}{\dd r}\right)-\tilde V_a\frac{r^2+a^2}{\Delta_-}f.
	\end{align*}
\end{lemma}

By virtue of twisting, the modified potential can be chosen to be positive for large $r$:

\begin{lemma}
	\label{lemma:twisting_positive_g}
	Let $h:=r^{-1/2+\kappa}$.
	If $|m|$ is sufficiently large, then there is an $R>r_+$ such that $\tilde V^h_a>0$ for $r>R$. The choice of $R$ is independent of $a$ and $\alpha$. Moreover $\tilde V_a^h=\OO(1)$ as $r\rightarrow\infty$.
\end{lemma}

\begin{proof}
	We look at the asymptotic behaviour of the different parts of $\tilde V_a$:
	\begin{align*}
	V_0-\omega^2&\sim\frac{1}{\ell^2}\left(\lambda+a^2\omega^2-2ma\omega\Xi\right)>\frac{1}{\ell^2}m^2\Xi^2>0\\
	V_+&=\frac{2\Delta_-}{(r^2+a^2)^2}\frac{r^2}{\ell^2}+\frac{\Delta_-}{(r^2+a^2)^4}\left(a^4\Delta_-+(r^2-a^2)2Mr\right)\\&\sim\frac{2\Delta_-}{(r^2+a^2)^2}\frac{r^2}{\ell^2}+\frac{a^4\Delta_-^2}{(r^2+a^2)^4}
	\end{align*}
	One easily computes that
	\begin{align*}
	v(r):=\frac{2-\alpha}{\ell^2}\frac{\Delta_-}{(r^2+a^2)^2}r^2-h^{-1}\frac{\Delta_-}{r^2+a^2}\frac{\dd}{\dd r}\left(\frac{\Delta_-}{r^2+a^2}\frac{\dd h}{\dd r}\right)=\OO(1),
	\end{align*}
	which yields the result.
\end{proof}

\section{The twisted Euler-Lagrange equation}
\label{sec:Twisted_Euler_Lagrange}

We give here the derivation of the weak twisted Euler-Lagrange equation.

\begin{proof}[Proof of Lemma~\ref{lemma:ELDN}]
	The following proof can be extracted from \citep{Evans}. We give the extension to twisted derivatives here for the sake of completeness.
	The minimiser $f_{a}$ is a minimiser of the functional
	\begin{align*}
	\tL_{a}(f):=\int_{r_+}^{\infty}\left(\frac{\Delta_-}{r^2+a^2}r^{-1+2\kappa}\left\lvert\frac{\dd}{\dd r}\left(r^{\frac{1}{2}-\kappa}f\right)\right\lvert^2+\Vmod_a\frac{r^2+a^2}{\Delta_-}|f|^2\right)\,\dd r
	\end{align*}
	under the constraint
	\begin{align*}
	\mathcal{J}(f)=0,
	\end{align*}
	where
	\begin{align*}
	\JJ(f)=\int_{r_+}^{\infty} G(r,f)\,\dd r,~~~~G(r,f)=\frac{1}{r^2}\left(|f|^2-r_+\right).
	\end{align*}
	Moreover, define $g(r,f):={2f}/{r^2}$.
	Fix $\psi_1\in \uH_{\kappa}^{1}(r_+,\infty)$. We assume in a first step that $g(r,f_a)$ is not identically zero almost everywhere on $(r_+,\infty)$. Then we can find a $\psi_2\in \uH_{\kappa}^{1}(r_+,\infty)$ such that
	\begin{align*}
	\int_{r_+}^{\infty}g(r,f_{a})\psi_2(r)\,\dd r\neq 0.
	\end{align*}
	Define $j(\tau,\sigma):=\JJ(f_{a}+\tau\psi_1+\sigma\psi_2)$	for $\tau,\sigma\in\RR$. Clearly, $j(0,0)-0$. Since $\frac{\del g(r,f_a+\tau\psi_1+\sigma\psi_2)}{\del\tau}\psi_1$ and $\frac{\del g(r,f_a+\tau\psi_1+\sigma\psi_2)}{\del\tau}\psi_2$ are integrable on $(r_+,\infty)$, $j$ is in $C^1$. In particular, we have
	\begin{align*}
	\frac{\del j}{\del\sigma}(0,0)=\int_{r_+}^{\infty}g(r,f_{a})\psi_2(r)\,\dd r\neq 0.
	\end{align*}
	By the Implicit Function Theorem, there is a $\kappa:\,\RR\rightarrow\RR$ such that $\kappa(0)=0$ and
	\begin{align*}
	j(\tau,\kappa(\tau))=0.
	\end{align*}
	In other words, the function $f_{a}+\chi(\tau)$, where
	\begin{align}
	\label{eqn:implicit_curve}
	\chi(\tau):=\tau\psi_1+\kappa(\tau)\psi_2,
	\end{align}
	satisfies the integral constraint. Thus, setting $i(\tau):=\tL_a(f_{a}+\chi(\tau))$,
	we obtain $i'(0)=0$. Note here that $i$ is differentiable in $\tau$ since $f_{a}\in \uH_{\kappa}^1(r_+,\infty)$. We have
	\begin{align*}
	\frac{\dd i}{\dd\tau}\bigg\lvert_{\tau=0}&=2\int_{r_+}^{\infty}\bigg(\frac{\Delta_-}{r^2+a^2}r^{-1+2\kappa}\frac{\dd}{\dd r}\left(r^{\frac{1}{2}-\kappa}f_{a}\right)\left(\frac{\dd}{\dd r}\left(r^{\frac{1}{2}-\kappa}\psi_1\right)+\kappa'(0)\frac{\dd}{\dd r}\left(r^{\frac{1}{2}-\kappa}\psi_2\right)\right)\\&~~~~~~+\Vmod_a\frac{r^2+a^2}{\Delta_-}f_{a}(\psi_1+\kappa'(0)\psi_2)\bigg)\,\dd r.
	\end{align*}
	From (\ref{eqn:implicit_curve}), we deduce
	\begin{align*}
	\kappa'(0)=-\frac{\int_{r_+}^{\infty}g(r,f_{a})\psi_1\,\dd r}{\int_{r_+}^{\infty}g(r,f_{a})\psi_2\,\dd r}.
	\end{align*}
	Setting 
	\begin{align*}
	\lambda:=2\frac{\int_{r_+}^{\infty}\left(\frac{\Delta_-}{r^2+a^2}r^{-1+2\kappa}\frac{\dd}{\dd r}\left(r^{\frac{1}{2}-\kappa}f_{a}\right)\frac{\dd}{\dd r}\left(r^{\frac{1}{2}-\kappa}\psi_2\right)+\Vmod_a\frac{r^2+a^2}{\Delta_-}f_{a}\psi_2\right)\,\dd r}{\int_{r_+}^{\infty}g(r,f_{a})\psi_2\,\dd r}
	\end{align*}
	yields that
	\begin{align*}
	\begin{split}
	\int_{r_+}^{\infty}\bigg(\frac{\Delta_-}{r^2+a^2}r^{-1+2\kappa}\frac{\dd}{\dd r}\left(r^{\frac{1}{2}-\kappa}f_{a}\right)\frac{\dd}{\dd r}\left(r^{\frac{1}{2}-\kappa}\psi_1\right)+\Vmod_a\frac{r^2+a^2}{\Delta_-}f_{a}\psi_1\bigg)\,\dd r=\lambda\int_{r_+}^{\infty}\frac{f_a}{r^2}\psi\,\dd r
	\end{split}
	\end{align*}
	for all $\psi\in \uH_{\kappa}^1(r_+,\infty)$.  We have $f_a\in \uH_{\kappa}^1(r_+,\infty)$, whence $\lambda=-\nu_a$.
	
	It remains to deal with the case $g(r,f_a)=0$ a.\,e. This, however, would yield that $f=0$ in contradiction to the norm constraint.
\end{proof}

\end{appendices}

\bibliographystyle{alphadin}

\addcontentsline{toc}{section}{References}
{\footnotesize
	\bibliography{literatureSuperradiance}}

\end{document}